\documentclass[11pt, oneside]{article}  	
\usepackage{geometry}        		
\usepackage{longtable,booktabs,array}
  
\usepackage{titletoc}

\input{preamble}

\title{A variational framework for modal estimation}
\author{Tâm Le Minh$^{1,2}$, Julyan Arbel$^2$, Florence Forbes$^2$ and Hien Duy Nguyen$^{3,4}$}
\date{}

\begin{document}
\maketitle
\begin{center}
\textit{\noindent
$^1$School of Medicine and Dentistry, Griffith University, Nathan, QLD 4101, Australia \\
$^2$Univ. Grenoble Alpes, Inria, CNRS, Grenoble INP, LJK, 38000 Grenoble, France \\
$^3$School of Computing, Engineering and Mathematical Sciences, La Trobe University, Bundoora, VIC 3086, Australia \\
$^4$Institute of Mathematics for Industry, Kyushu Univ., Nishi Ward, Fukuoka 819-0395, Japan
}
\end{center}


\paragraph{Abstract.} \sloppy{Multivariate mode estimation arises in many statistical problems such as inverse problems, multimodal sampling, and density-based clustering, but becomes challenging in moderate to high dimensions, especially when the underlying density is not directly evaluable. We introduce GERVE (Gibbs-measure Entropy-Regularized Variational Estimation), a sample-based method for estimating multivariate modes by approximating Gibbs distributions directly from samples, without estimating or evaluating the density. GERVE uses Gaussian-mixture variational annealing and natural-gradient optimization, producing a mixture concentrated in high-density regions whose component responsibilities also provide a clustering of the observations. We prove theoretical guarantees in two regimes: as the Gibbs temperature goes to zero, the optimal variational mixture concentrates around the global modes of the population density; at fixed positive temperature, we prove existence, consistency, and asymptotic normality of empirical maximizers and propose a bootstrap procedure for uncertainty quantification. Simulations and a real-data experiment show that GERVE accurately recovers modes and produces meaningful clusters. 
}


\section{Introduction}

Understanding the structure of an unknown probability distribution $p$ on $\mathbb{R}^d$, observed through samples, is a central task in statistics. Modes, the local maxima of $p$, reveal data high-density regions. Detecting multiple modes exposes heterogeneity, supports clustering and segmentation, and offers alternative explanations in inverse problems \citep{carreira2006acceleration, chacon2015population}. Recently, it has also emerged as a preliminary requirement for new efficient sampling methods of multimodal high-dimensional distributions \citep{noble2025learned}.

Classical approaches to mode estimation are commonly grouped into three main families. The first comprises geometric methods that detect high-density regions from the sample \citep{chernoff1964estimation, dalenius1965mode, venter1967estimation, sager1978estimation}. These methods are simple and robust but tend to degrade in high dimensions. The second fits a smooth density and then locates its modes \citep{parzen1962estimation, loftsgaarden1965nonparametric, kronmal1968estimation, samworth2018recent, liu2020empirical}. This route is principled but suffers from smoothing bias and high computational cost. The third consists of mean-shift algorithms that ascend the gradient of a kernel density estimate (KDE) without requiring fitting a global model  \citep{fukunaga1975estimation, comaniciu2002mean, carreira2007gaussian, arias2016estimation}. These methods are efficient but bandwidth-sensitive and can miss well-separated modes. To our knowledge, no existing method simultaneously targets multiple density modes, provides valid uncertainty quantification, and scales with dimension via simple algorithmic updates, without requiring a full density estimate \citep[contrary to][]{genovese2016non, ameijeiras2021multimode}.

We propose a different perspective: rather than searching directly for local maxima of an unknown function, we recast the task of mode estimation as that of approximating an associated Gibbs measure from samples. Let $f : \mathcal{S} \to \mathbb{R}$ be a measurable function, where $\mathcal{S} \subset \mathbb{R}^d$ is compact with non-empty interior, and fix a temperature parameter $\omega > 0$. Define
\begin{equation*}
g_\omega(d\xb) = \frac{\exp(f(\xb)/\omega)}{Z_{\omega}(f)} d\xb, \qquad \text{with~} Z_{\omega}(f) := \int_{\mathcal{S}} \exp(f(\yb)/\omega) d\yb,
\end{equation*}
and assume $Z_{\omega}(f) < \infty$. The Gibbs measure $g_\omega$ preserves the modes of $f$ and is the maximizer, over the unconstrained family of probability densities on $\mathcal{S}$, of the entropy-regularized variational objective \citep{donsker1976asymptotic}:
\begin{equation}
\mathcal{L}_\omega(q) = \mathrm{E}_{q} \left[f(\Xb)\right] + \omega \mathcal H_{\mathcal S}(q),
\qquad
\text{with~} \mathcal H_{\mathcal S}(q) = - \int_{\mathcal S} q(\yb)\log q(\yb) d\yb.
\label{eq:population_objective}
\end{equation}
Our starting point is therefore a statistical characterization of the modal structure, rather than a new search heuristic.

In the statistical setting, however, the target function is an unknown density $f = p$ (assumed nonnegative and supported on a subset of $\mathcal{S}$), observed only through i.i.d. samples $\Xb_{1:N} \sim p$. Since the density $p$ is unknown, the term $\mathrm{E}_{q} \left[p(\Xb)\right]$ is not tractable, so direct evaluation of the population objective $\mathcal{L}_\omega(q)$ is impossible. Crucially, because both $p$ and $q$ are densities, we have the identity
\begin{equation}
\mathrm{E}_{q} \left[p(\Xb)\right] = \mathrm{E}_{p} \left[q(\Xb)\right].
\label{eq:gerve_identity}
\end{equation}
This rewrites the variational objective in a form that is accessible from data. Replacing the expectation under $p$ by its empirical counterpart leads to the objective
\begin{equation}
\widehat{\mathcal{L}}_{\omega,N}(q) = \frac1N\sum_{i=1}^N q(\Xb_i) + \omega \mathcal H_{\mathcal S}(q).
\label{eq:empirical_objective}
\end{equation}
We therefore define \textbf{GERVE} (\textit{Gibbs-measure Entropy-Regularized Variational Estimation}) from this principle: a likelihood-free framework for modal estimation that approximates, from samples alone, a Gibbs measure associated with the unknown density by maximizing $\widehat{\mathcal{L}}_{\omega,N}$ over a tractable variational family. This sample-based formulation is the key distinction from the natural variational annealing (NVA) framework of \cite{leminh2025natural}, where the objective function is assumed to be either known or directly evaluable. Here, the identity \eqref{eq:gerve_identity} gives a statistical criterion that can be directly estimated from samples. To optimize it, we adapt the Gaussian-mixture variational annealing machinery of \cite{leminh2025natural}. This empirical formulation does not only avoid explicit density estimation, but also provides the basis for asymptotic analysis and bootstrap uncertainty quantification.

To approximate the empirical objective in multimodal settings, we optimize $\widehat{\mathcal{L}}_{\omega,N}(q)$ over a Gaussian-mixture family. This choice provides an approximation class that can represent multiple high-density regions while component means and covariances capture location and anisotropy. However, Gaussian mixtures are used here as a variational family for Gibbs-measure approximation rather than as a likelihood-based generative model for the data. Therefore, GERVE does not aim to recover a ``true'' number of Gaussian components, but to infer the modal structure of the underlying density, which in multivariate settings need not coincide with the number of components of a generative mixture \citep{kabata2025singularities}.
A simple illustration of this fact is the bivariate three-component Gaussian mixture exhibiting four modes presented in \cite{carreira2007gaussian} and \cite{leminh2025natural}. Following the variational implementation of \cite{leminh2025natural}, we use natural-gradient updates based on exponential-family representations of Gaussian mixtures \citep{lin2019fast}.

Annealing provides the mechanism that links this variational formulation to modal recovery. At large temperatures, entropy encourages diffuse mixtures that cover several candidate high-density regions. When the temperature decreases, components concentrate more sharply around dominant modes of the target distribution. An overcomplete mixture can be useful in practice: additional components help cover candidate high-density regions early in annealing, but later redundant components may receive negligible weight or concentrate near the same mode. Overspecifying the variational mixture therefore enriches the approximation class for mode recovery, rather than representing an assumption about a ``true'' number of generative mixture components, as in mixture-based clustering.

Our main contributions are:
\begin{enumerate}
    \item \textit{A sample-based variational formulation for mode estimation.}
    We show that Gibbs-measure approximation for an unknown density can be formulated through an empirical entropy-regularized objective that is directly computable from samples, without requiring explicit density estimation. This provides the statistical foundation of GERVE.
    \item \textit{A statistical characterization of the variational modal target.}
    We study two complementary regimes. At the population level, optimal Gaussian-mixture approximations concentrate around the global modes as the temperature vanishes. At fixed positive temperature, empirical maximizers are consistent and asymptotically normal. We also develop a matched-mode bootstrap for uncertainty quantification at fixed temperature.
    \item \textit{A practical variational implementation.}
    We adapt Gaussian-mixture variational annealing and natural-gradient optimization to the proposed sample-based criterion. The resulting procedure explores multiple high-density regions and connects to mean shift in a fixed-covariance special case. We evaluate it for mode estimation and modal clustering, and apply it to collision-hotspot detection.
\end{enumerate}

Our approach is related to modal clustering \citep{menardi2016review,chacon2020modal}, which defines clusters as basins of attraction of density modes. GERVE uses variational approximation, yielding clusters that locally agree with basins near modes under mild conditions (Suppl.~\ref{app:agreement}), while providing computational efficiency and natural uncertainty quantification.

\noindent\textbf{Outline.} The remainder is organized as follows. Section~\ref{sec:foundations} recalls the variational foundations of Gibbs-measures and shows that Gaussian mixture approximations can be used for mode seeking by specifying their behavior under annealing (Thm.~\ref{thm:annealing_main}).  Section~\ref{sec:framework} shows that this mechanism can be extended to mode estimation when the density to maximize is only available through samples. The main principles of GERVE are introduced and the asymptotic behavior of the empirical maximizers  is established (Thm.~\ref{thm:consistency_main}). Section~\ref{sec:algorithms} presents algorithmic details (Alg.~\ref{alg:gerve_general}). Section~\ref{sec:uq} gives the fixed-temperature bootstrap procedure (Alg.~\ref{alg:bootstrap_uuq}) with guarantees (Thm.~\ref{thm:bootstrap-params-max-cor_main} and~\ref{thm:modes-CLT-bootstrap_main}). Sections~\ref{sec:simulations} and \ref{sec:real_datasets} report simulations and a case study on UK road-collision hotspots. Section~\ref{sec:conclusion} discusses limitations and future directions. All proofs and additional details are in Supplementary Material. A Python implementation for GERVE can be found at \texttt{github.com/tam-leminh/gerve}.

\noindent\textbf{Notation.} We consider that $p$ (or $f$) has a finite number $I$ of global modes $\{\xb_i^\star\}_{i=1}^I$. When $f : \mathcal{S} \to \mathbb{R}$ and there is no ambiguity on the domain, $f$ is $\mathcal{C}^k$ means $f \in \mathcal{C}^k(\mathcal{S})$, that is $f$ is continuously differentiable on $\mathcal{S}$ to the order $k$. For a distribution $q$ on ${\mathcal A} \subseteq \mathbb{R}^d$ and a set $
{\mathcal B} \subset {\mathcal A}$, we write $q(
{\mathcal B}) := \int_{
{\mathcal B}} q(\xb) d\xb$ and $q^{
{\mathcal B}} := q \mathds 1_{{\mathcal B}} / q({\mathcal B})$, its truncation (or conditional law) on $
{\mathcal B}$. The set $\{q_\varthetab, \varthetab \in \Theta\}$ denotes a variational family of densities parameterized by $\varthetab$ where $\varthetab=\lambdab$ for single Gaussians $q_{\lambdab}=\mathcal N(\mub,\Sigmab)$, or 
$\varthetab=\Lambdab$ for 
Gaussian mixtures 
$q_{\Lambdab}=\sum_{k=1}^K \pi_k \mathcal N(\mub_k,\Sigmab_k)$.
When there is no ambiguity, we write for the objective $\mathcal{L}_\omega(\varthetab) := \mathcal{L}_\omega(q_\varthetab)$. Unless explicitly specified, $\lVert \cdot \rVert$ denotes the Euclidean norm for vectors and the corresponding operator norm for matrices.

\section{Population variational formulation for mode seeking}
\label{sec:foundations}
\label{sec:objective}

For a measurable function $f:\mathcal S\to\mathbb R$, the associated Gibbs measure $g_\omega$ preserves its modes. 
The goal of finding the modes of $f$ can then be turned into that of approximating $g_\omega$, which in turn can be achieved through a variational approximation principle. This population setting, in which $f$ is known, corresponds to the variational-annealing setting of \cite{leminh2025natural}. We recall it here to establish the population properties that motivate the sample-based formulation of Section~\ref{sec:framework}.
The classical variational identity
\[
\mathcal{L}_\omega(q)
 =  \omega \log Z_\omega(f) - \omega \KL \left(q \| g_\omega\right),
\]
with $\KL\left(q \| g_\omega\right) := \int_{\mathcal{S}} q(\xb) \log (q(\xb)/g_\omega(\xb)) d\xb$,
implies that $g_\omega$ uniquely maximizes the objective $\mathcal{L}_\omega$ when $Z_\omega(f)<\infty$ \citep{donsker1976asymptotic}. Restricting the search for $q$ in a tractable variational family gives an approximation for $g_\omega$ by minimization of $\KL \left(q \| g_\omega\right)$. 

As a natural choice for approximating multimodal landscapes, we study how Gaussian mixtures approximate $g_\omega$ when $\omega$ decreases to $0$. Since $g_\omega$ is defined only on $\mathcal{S}$, admissible variational distributions are required to have support in $\mathcal{S}$. We therefore work with mixtures of truncated Gaussians. The next result, proved in Supplementary Material~\ref{app:thm_annealing}, describes how optimal mixtures concentrate on the global modes when the number of components $K$ is at least the number of global modes $I$.

\begin{theorem}[Gaussian mixture concentration on global modes]
\label{thm:annealing_main}
Assume $f \in \mathcal{C}^3(\mathcal{S})$ and bounded on $\mathcal{S}$, and for each mode in $\{\xb_i^\star\}_{i=1}^{I}$, $\xb_i^\star \in \mathrm{int}(\mathcal{S})$ and $\Hb_i := -\nabla^2 f(\xb_i^\star)\succ 0$. Consider $\mathcal{Q}$ to be the variational family of Gaussian mixtures truncated on $\mathcal{S}$, with a fixed number of components $K \ge I$ and upper-bounded covariances. Let $q_{\omega}^\star$ denote any maximizer of $\mathcal{L}_\omega$ in $\mathcal{Q}$. Choose disjoint neighborhoods $U_1, \dots, U_I$ of the modes and let $U_0 := S \setminus \bigcup_{i=1}^I U_i$. 

Then, as $\omega\to 0$:
\begin{enumerate}

\item[(i)] We have $\KL(q_{\omega}^\star \Vert  g_\omega) \rightarrow 0$. In particular, $q_{\omega}^\star(U_0) \rightarrow 0$, and, for each neighborhood $U_i$, $i = 1, \dots, I$,

\begin{equation*}
    q_{\omega}^\star (U_i) \rightarrow c_i,  \qquad \text{where}\,\, c_i \propto (\det \Hb_{i})^{-1/2}, \,\sum_{i=1}^I c_i = 1. 
\end{equation*}
\item[(ii)] For all $i = 1, \dots, I$, the conditional laws of $q_{\omega}^\star$ and $g_\omega$ on $U_i$ satisfy $\KL\big(q_{\omega}^{\star,U_i} \Vert  g_{\omega}^{U_i}\big) \rightarrow 0$. In particular, the mean and covariance of $q_{\omega}^{\star,U_i}$ satisfy
\begin{equation*}
\mathbb{E}_{q_{\omega}^{\star,U_i}}[\Xb] = \xb_i^\star + o(1),
\qquad
\mathrm{Cov}_{q_{\omega}^{\star,U_i}}(\Xb) = o(1).
\end{equation*}

\end{enumerate}
\end{theorem}

This result shows that, when $K \geq I$, optimal Gaussian-mixture approximations place all their mass in neighborhoods of the global modes of $f$ as $\omega \to 0$. The local shape of $g_\omega$ at a particular mode $\xb_i^\star$ has been studied in Proposition 3 of~\cite{leminh2025natural}. There, the optimal Gaussian approximation in a neighborhood of $\xb_i^\star$ is of the form $\mathcal{N}(\mub_\omega, \Sigmab_\omega)$ satisfying:
\begin{equation*}
    \mub_\omega = \xb_i^\star + o(\sqrt{\omega}), \qquad \Sigmab_\omega = \omega \Hb_i^{-1} + o(\omega).
\end{equation*}

These expansions clarify how Gibbs measures can be used for mode seeking. As $\omega$ decreases, mixture mass concentrates near the global modes of $f$. The temperature parameter governs the trade-off between exploration (through smooth and diffuse mixtures at large $\omega$) and exploitation (concentrated components at small $\omega$). These population results provide the motivation for annealing in a variational setting: decreasing the temperature progressively concentrates the optimal variational approximation around the global modes, motivating the use of a decreasing temperature schedule for mode localization. We next introduce a statistical formulation in which $f$ is unobserved and only accessible through samples drawn from it.

\section{GERVE: sample-based variational mode estimation}

\label{sec:framework}

We observe i.i.d.\ samples $\Xb_{1:N}\sim p$ supported on a bounded domain $\mathrm{supp}(p) := \{ \xb \in \mathbb{R}^d : p(\xb) > 0\}$ and we choose the compact $\mathcal S$ such that $\mathrm{supp}(p) \subset \mathcal S$. Replacing $f$ with $p$, we distinguish the \emph{population} and \emph{empirical} objectives respectively referring to $\mathcal{L}_\omega(q)$ and $\widehat{\mathcal{L}}_{\omega,N}(q)$ defined by~\eqref{eq:population_objective}, with $f := p$, and~\eqref{eq:empirical_objective}.

Instead of performing optimization on truncated distributions on $\mathcal{S}$, GERVE maximizes $\widehat{\mathcal{L}}_{\omega,N}$ over untruncated Gaussian mixture families on $\mathbb{R}^d$,
while all integrals in the objective are kept over $\mathcal S$. Supplementary Material~\ref{app:truncation_approximation} discusses that if $\mathcal{S}$ is taken large enough, then the objective gap between the truncated and untruncated distributions is uniformly small, which leads to an approximation bound between their population maximizers under an argmax separation condition. This simplifies optimization, as it removes parameter-dependent normalizers in the natural-gradient updates (details in Section \ref{sec:setup}). 

For a single Gaussian ${\mathcal N}(\mub,\Sigmab)$, we consider the compact parameter set
\begin{equation}
    \Theta := \Bigl\{\lambdab=(\Sigmab^{-1}\mub, -\frac{1}{2}\Sigmab^{-1}) : \lVert\mub\rVert_\infty \le \mub_{\max}, \sigma_{\min}^2 \Ib \preceq \Sigmab \preceq \sigma_{\max}^2 \Ib \Bigr\}.
\label{eq:compact_gau}
\end{equation}
Proposition~\ref{prop:trunc_gap} shows that if the outside mass $\varepsilon(\lambdab) = 1 - \int_{\mathcal{S}} q_\lambdab(\xb) d\xb$ is uniformly small on $\Theta$, then the objective gap is uniformly small.  The same reasoning applies component-wise to mixtures on the parameter set
\begin{equation}
\Theta_K :=
\Bigl\{
\Lambdab=(v_1,\dots,v_{K-1},\lambdab_1,\dots,\lambdab_K):
\lambdab_k\in\Theta,\ 
|v_k-v_\ell|\le v_{\max}\ \text{for all } k,\ell\in[K]
\Bigr\},
\label{eq:compact_mixtures}
\end{equation}
where $v_k:=\log(\pi_k/\pi_K)$ for $k<K$ and $v_K:=0$. The weight constraint ensures that $\log(\pi_k/\pi_\ell)$ remains bounded, for all $k,\ell\in[K]$, which prevents degeneracies while preserving flexibility and compactness. In fact, throughout, $\Theta_{K}$ shall be understood as its restriction to the closure $\overline{\ensuremath{\Theta_{K}^{\mathrm{lex}}}}$ of the lexicographically ordered parameter space
$\Theta_{K}^{\mathrm{lex}}$ (defined in Suppl.~\ref{subsec:Lexicographic-order-Theta}), that enforces a canonical labelling and thus provides sufficient
identifiability for our technical results.

In the multimodal optimization setting, the main source of bias is due to the entropy-induced repulsion in mixtures. We deliberately exploit this repulsion early in training to mitigate mode collapse, a known issue in variational Gaussian mixtures \citep{wu2019solving, jerfel2021variational, soletskyi2025theoretical}. The induced bias progressively diminishes under annealing.

\label{sec:annealing_asymptotics}

The following theorem establishes that GERVE remains statistically consistent at fixed~$\omega = \omega_0 > 0$, as $N \to \infty$. An asymptotic normality result is also derived. These results are stated for Gaussian mixtures in $\Theta_K$ (eq.~\eqref{eq:compact_mixtures}), but can be applied to single Gaussians in $\Theta$ (eq.~\eqref{eq:compact_gau}) as a special case. The proof is provided in Supplementary Material~\ref{app:thm_consistency}. 

\begin{theorem}[Asymptotics of empirical maximizers] \label{thm:consistency_main}
    At fixed $\omega_0 > 0$, if $p$ is bounded and continuous on $\mathcal{S}$, then a population maximizer $\Lambdab^\star_{\omega_0}$ of $\mathcal{L}_{\omega_0}$ exists on $\Theta_K$. If it is unique, then any empirical maximizer $\widehat{\Lambdab}_{\omega_0,N}\in\argmax \widehat{\mathcal{L}}_{\omega_0,N}$ satisfies, as $N \to \infty$:
\begin{equation*}
\widehat{\Lambdab}_{\omega_0,N} \xrightarrow[]{P} \Lambdab^\star_{\omega_0}.
\end{equation*}
Furthermore, if $\Lambdab^\star_{\omega_0} \in \operatorname{int}(\Theta_K)$ and $\Hb^\star_{\omega_0} := -\nabla_{\Lambdab}^2\mathcal{L}_{\omega_0}(\Lambdab^\star_{\omega_0}) \succ 0$, then
\[
\sqrt{N} \big(\widehat\Lambdab_{\omega_0,N}-\Lambdab^\star_{\omega_0}\big)
 \xrightarrow[]{\mathcal{D}} 
\mathcal{N} (0, \Wb_{\omega_0}),
\]
where $\Wb_{\omega_0} := \big(\Hb_{\omega_0}^\star\big)^{-1} \Vb_{\omega_0} \big(\Hb_{\omega_0}^\star\big)^{-1}$, and $\Vb_{\omega_0}:=\operatorname{Var} \big(\nabla_{\Lambdab} q_{\Lambdab}(\Xb)\big)\big|_{\Lambdab=\Lambdab^\star_{\omega_0}}$.
\end{theorem}

\begin{remark*}
Maximizer uniqueness is subtle in mixture models. In particular, unless the parameter space is restricted so as to rule out degeneracies, a component with vanishing weight can have its associated mean and covariance drift without affecting the mixture density, and if several components share the same $(\mub,\Sigmab)$ then redistributing weights within that tied group also leaves the density unchanged. Consequently, convergence at the parameter level cannot be guaranteed in general, even when the induced densities converge to a maximizer.
To formalize consistency without uniqueness, one can replace the uniqueness assumption in Theorem~3.1 by the weaker requirement that
\[
\mathcal{D}_{\omega_0}=\argmax_{\Lambdab\in\Theta_K}\mathcal{L}_{\omega_0}(\Lambdab)
\]
is nonempty and compact. The set-valued argmax theorem \citep[Cor.~5.58]{van1998asymptotic} then yields
\[
\mathrm{dist}\bigl(\widehat{\Lambdab}_{\omega_0,N},\mathcal{D}_{\omega_0}\bigr)\xrightarrow{{P}}0,
\qquad
\mathrm{dist}(\xb,\mathcal A)=\inf_{\ab\in \mathcal A}\|\xb-\ab\|.
\]
The asymptotic normality statement, however, holds only under local identifiability, for example \(\Hb^\star_{\omega_0}\succ 0\), which rules out duplicated components at the maximizer.
In practice, such non-uniqueness can be useful with an overcomplete mixture: some component weights may become negligible, while several component means may coalesce near the same prominent mode, reducing sensitivity to moderate overspecification of $K$.
\end{remark*}

Theorems~\ref{thm:annealing_main} and~\ref{thm:consistency_main} characterize two complementary principles of GERVE. At the population level, optimal variational mixtures concentrate around the global modes of $p$ as $\omega \to 0$, while at fixed positive temperature, empirical maximizers consistently estimate their population counterparts as $N \to \infty$. These results provide both a population motivation for annealing in mode localization and statistical guarantees for the sample-based variational criterion. In practice, these two principles lead to a mode-seeking procedure in which the variational approximation progressively focuses on high-density regions.

Beyond mode estimation, GERVE naturally induces a clustering structure. 
With possibly overcomplete mixtures, the mixture mass is able to concentrate around distinct 
modes as entropy decreases. In practice, redundant components either vanish or collapse, so the 
effective number of clusters adapts automatically. Unlike likelihood-based Gaussian mixtures clustering, GERVE targets density modes under annealing.

In practice, once fitted, cluster assignments are straightforward 
via responsibilities (posterior component probabilities).
This connects to \emph{modal clustering} \citep{menardi2016review,chacon2020modal}, where clusters are attraction basins of density modes. Classical approaches rely on nonparametric density estimates and gradient flows, while GERVE uses a parametric variational approximation. The resulting clusters need not coincide exactly with attraction basins, but they agree locally near each mode under mild conditions, see Supplementary Material~\ref{app:agreement} for a statement and a counterexample.
This provides a variational approach to modal clustering that can accommodate an overcomplete initial choice of $K$. A caveat is that annealing emphasizes dominant modes: secondary modes may disappear as $\omega\to 0$, so retaining them requires stopping at a positive temperature.

\section{Computational implementation and variants}
\label{sec:algorithms}

We implement the sample-based variational formulation of Section~\ref{sec:framework} by adapting the Gaussian-mixture natural-gradient and annealing machinery of \cite{leminh2025natural}. The major computational difference is that the objective and its gradients are now estimated directly from the observed sample through $\widehat{\mathcal{L}}_{\omega,N}$, rather than from evaluations of a known target function.

\subsection{Variational optimization with natural gradients}
\label{sec:setup} 
The empirical objective $\widehat{\mathcal{L}}_{\omega, N}$ (eq.~\eqref{eq:empirical_objective}) can be optimized over $\Theta$ or $\Theta_K$ using natural gradients $
\widetilde{\nabla}_\varthetab \widehat{\mathcal{L}}_{\omega, N}(\varthetab) := \Fb(\varthetab)^{-1} \nabla_\varthetab \widehat{\mathcal{L}}_{\omega, N}(\varthetab)$, which incorporates the geometry of the variational parameter space through the Fisher information matrix $\Fb(\varthetab)$ (\citealp{amari1998natural}).
When $q_\varthetab$ belongs to an exponential family with sufficient statistic $ \Tb(\Xb)$, the natural gradient equals the gradient in expectation parameters $\Mb = \mathbb{E}_{q_\varthetab}[\Tb(\Xb)]$: 
\begin{equation}
\widetilde{\nabla}_\varthetab \widehat{\mathcal{L}}_{\omega, N}(\varthetab) = \nabla_\Mb \widehat{\mathcal{L}}_{\omega, N}(\varthetab).
\label{eq:natural_expectation_duality}
\end{equation}
This identity allows computation of natural gradients without explicit estimation and inversion of the Fisher matrix. It also holds for untruncated Gaussian mixtures via the Minimal Conditional Exponential Family representation (MCEF; \citealp{lin2019fast}).

Algorithm~\ref{alg:gerve_general} describes natural-gradient ascent on
$\widehat{\mathcal{L}}_{\omega,N}(\vartheta)$ along a temperature schedule $\omega_t$, $t=1,\ldots,T$. Standard stochastic-approximation arguments give convergence to stationary sets at fixed temperature under suitable conditions (see Supplementary Material~\ref{app:snga_convergence}). Annealing is motivated by the population concentration results of Section~\ref{sec:foundations}, with the temperature decreased gradually to allow the optimization to adapt between successive temperatures. A projection step can optionally be used in practice to keep the iterates within the compact parameter space $\Theta$ or $\Theta_K$.

\LinesNumbered
\begin{algorithm}[H]
\caption{GERVE (general form)}
\label{alg:gerve_general}
\textsc{Input:} samples $\Xb_{1:N}$.\\
\textsc{Set:} initial $\varthetab_1$, stepsizes $(\rho_t)_t$, schedule $(\omega_t)_t$, (optional) batch size $B$.\\
\For{$t=1 {:} T$}{
(Optional) \textsc{Sample} a mini-batch $\{\Xb^{(t)}_i\}_{i=1}^B$ from $\Xb_{1:N}$ (or use full data).\\
\textsc{Compute} an (unbiased) natural-gradient estimate $\widetilde{\nabla}\widehat{\mathcal L}_{\omega_t,N}(\varthetab_t)$.\\
\textsc{Update} $\varthetab_{t+1} \leftarrow \varthetab_t + \rho_t \widetilde{\nabla}\widehat{\mathcal L}_{\omega_t,N}(\varthetab_t)$.\\
(Optional) \textsc{Project} to a compact set.
}
\textsc{Return:} $\varthetab_{T+1}$.
\end{algorithm}

\subsection{Connection with Gaussian mean shift}
\label{sec:fgc_alg}

Consider the family of Gaussian distributions with a fixed isotropic covariance  $ \sigma^2 \Ib$,  $\{q_\mub=\mathcal N(\mub, \sigma^2 \Ib), \mub \in \Rset^d\}$. For a large enough $\mathcal{S}$, we have $\nabla_\mub \mathcal H_{\mathcal{S}}(q_\mub) \approx 0$, so the natural-gradient step can be simplified into: 
\begin{equation}
\mub_{t+1}  =  \mub_t + \frac{\rho_t}{N} \sum_{i=1}^N (\Xb_i-\mub_t) q_{\lambdab_t}(\Xb_i).
\label{eq:GERVE_update_fc_gaussian}
\end{equation}
Using a batch $\{\Xb^{(t)}_i\}_{i=1}^B$ gives an unbiased estimator of the RHS (see Suppl.~\ref{app:fcg_derivation} for the derivation).
Denoting by 
$p_h(\mub)=\frac{1}{N}\sum_{i=1}^N \varphi_h(\Xb_i-\mub)$ the KDE with a Gaussian kernel $\varphi_h$ with bandwidth $h$, the following result states that GERVE recovers the mean-shift algorithm when $h=\sigma^2$.
\begin{proposition}[Equivalence to Gaussian mean-shift]
\label{prop:gerve_meanshift_equivalence}
Update rule~\eqref{eq:GERVE_update_fc_gaussian} performs gradient ascent on $p_h$, and normalizing the step $\rho_t = p_h(\mub)^{-1}$ yields the classical mean-shift fixed-point update:
\[
\mub_{t+1} = \frac{\sum_{i=1}^N \varphi_h(\Xb_i-\mub_t) \Xb_i}{\sum_{j=1}^N \varphi_h(\Xb_j-\mub_t)}.
\]
\end{proposition}

Viewing fixed-covariance GERVE as mean-shift puts us under classical KDE theory. Operationally, $\mub_t$ is a particle ascending the smoothed density $p_h$, and the (fixed) covariance $h \Ib$ is the bandwidth: large $h$ leads to smooth and global but biased updates, while small $h$ yields sharp but noisy updates, sensitive to sampling. 

This regime is thus a well-understood nonparametric baseline that grounds richer GERVE variants. This can typically be used to set the covariance bounds in learned-covariance variants of GERVE. For example, if $p$ is $\mathcal{C}^3$, the gradient of the Gaussian KDE has bias $O(h)$ and fluctuation $O_p \big((N h^\frac{d+2}{2})^{-1/2}\big)$, yielding the optimal plug-in rate $h\asymp N^{-2/(d+6)}$ \citep{parzen1962estimation,romano1988weak,tsybakov1990recursive,genovese2016non}.

\subsection{Gaussian mixture GERVE}
\label{sec:mixtures}
When considering Gaussian mixtures represented by $\Theta_K$, it is convenient to adopt the following parametrization. Let $q_\Lambdab=\sum_{k=1}^K \pi_k \mathcal N(\mub_k,\Sb_k^{-1})$, where $\Sb_{k} := \Sigmab_{k}^{-1}$ with  logits $v_k=\log(\pi_k/\pi_K)$ and $\Lambdab=(v_{1:K-1},\lambdab_{1:K}) \in \Theta_K$.
Natural-gradient steps lead to component-wise updates coupled via the entropy term (full derivation in Suppl.~\ref{app:mix_derivation}):
\begin{align*}
\Sb_{k,t+1} &= \Sb_{k,t}\Big(1 - \frac{\rho_t}{N}\sum_{i=1}^N \Big((\Xb_i-\mub_{k,t})(\Xb_i-\mub_{k,t})^T \Sb_{k,t}-\Ib\Big) q_{\lambdab_{k,t}}(\Xb_i)\Big) \\
&\quad- \frac{2 \rho_t\omega_t}{\pi_{k,t}} \nabla_{\Sb^{-1}_k} \mathcal{H}_{\mathcal{S}}(q_\Lambdab)|_{\Lambdab = \Lambdab_t},\\
\mub_{k,t+1} &= \mub_{k,t} + \frac{\rho_t}{N} \Sb_{k,t+1}^{-1}\Sb_{k,t} \sum_{i=1}^N (\Xb_i-\mub_{k,t}) q_{\lambdab_{k,t}}(\Xb_i) + \frac{\rho_t\omega_t}{\pi_{k,t}} \Sb_{k,t+1}^{-1} \nabla_{\mub_k} \mathcal{H}_{\mathcal{S}}(q_\Lambdab)|_{\Lambdab = \Lambdab_t},\\
v_{k,t+1} &= v_{k,t} + \frac{\rho_t}{N}\sum_{i=1}^N \big(q_{\lambdab_{k,t}}(\Xb_i)-q_{\lambdab_{K,t}}(\Xb_i)\big) + \rho_t\omega_t \nabla_{\pi_k} \mathcal{H}_{\mathcal{S}}(q_\Lambdab)|_{\Lambdab = \Lambdab_t}.
\end{align*}
Learned covariances adapt to anisotropy and remain positive definite at fixed temperature. The fitted mixture becomes increasingly concentrated around modes as $\omega$ decreases. Furthermore, the entropy term induces repulsion between components at large $\omega$ that vanishes as $\omega \to 0$. Therefore, a slower decreasing annealing schedule allows for longer exploration by delaying mode collapse and improving mode capture in rugged and multimodal landscapes. Overcomplete mixtures may allocate multiple components to the same mode or assign negligible weight to redundant components. 

Gradient derivations and Monte Carlo estimators for the entropy terms are in Supplementary Material~\ref{app:mix_derivation}. Diagonal, isotropic, or fixed-covariance variants are obtained by modifying $\Theta_K$, and derivations follow the same pattern.

\subsection{Complexity and practical guidelines}
\label{sec:practical_playbook}

With $T$ iterations, batch size $B$, and a Gaussian mixture with $K$ components, the computational cost is $O\big((d^2 B + d^3) K T\big)$. Monte Carlo estimation of the entropy gradients adds a cost in $K^2$. The $d^3$ factor reflects the matrix multiplications and inversions that arise when updating full precision matrices. Restricting the variational family reduces this cost to $O(d B K T)$ for diagonal or fixed covariances. These restricted families often preserve the qualitative behavior of the algorithm and they improve scalability, which brings the cost in line with classical mode-seeking and clustering procedures. Further details appear in Supplementary Material~\ref{app:complexity}.

We now summarize practical guidance for hyperparameter selection. For the temperature schedule, the initial value of $\omega$ should be large enough to encourage exploration and mitigate premature mode collapse. The temperature should then decrease gradually toward a target value $\omega^\dagger$, allowing the optimization to adapt between successive temperatures, with $\omega^\dagger=0$ for mode estimation and $\omega^\dagger=\omega_0$ with $\omega_0>0$ for clustering.

For stepsizes, a safe choice satisfies the Robbins--Monro conditions $\sum_t \rho_t = \infty$ and $\sum_t \rho_t^2 < \infty$. In practice, when the temperature $\omega_t$ decreases over time, we find that mean updates are stabilized by the scaling $\rho_t \propto \omega_t^{-\beta}$ with $\beta \in (0,1)$. This choice often accelerates convergence, although it can violate the Robbins--Monro conditions.

Compactness constraints require explicit bounds for the covariance parameters in $\Theta$ and for the mixture parameters in $\Theta_K$. The upper bound $\sigma^2_{\max}$ should be large enough so that a single component can cover the entire dataset. For mode estimation, a suitable rule for the lower bound is $\sigma^2_{\min}$ of the order of magnitude of $N^{-2/(d+6)}$, which follows classical mean-shift bandwidth theory. For clustering, $\sigma^2_{\min}$ should be comparable to intra-class variances. For the logits that define the weights in $\Theta_K$ (eq.~\eqref{eq:compact_mixtures}), a sufficiently large $v_{\max}$ allows unbalanced mixture weights while preserving compactness, so that redundant components may receive negligible weight.

\section{Bootstrap uncertainty quantification for mode estimation}
\label{sec:uq}

At fixed temperature, the quantities targeted by the procedure are the component means of the population variational optimum. Uncertainty around these targets can be provided at a fixed stopping temperature $\omega_0>0$, using a bootstrap principle. 
A GERVE baseline fit is performed on the initial sample (constant $\omega_t = \omega_0$). GERVE is then refit $L$ times on resampled data. Each set of resulting modes is matched to the baseline modes, providing some empirical spread as confidence regions and a stability score.

\subsection{Bootstrap procedure}
\label{sec:uq-proc}

The bootstrap procedure is summarized in Algorithm~\ref{alg:bootstrap_uuq} for 
Gaussian mixture GERVE. Each fit, $\ell = 0, \dots, L$, uses an overcomplete $K$-component mixture followed by pruning (for some threshold $\epsilon \ll 1$) and merging of near-duplicates. The first fit identifies $\widehat{K}$ ``baseline'' modes. Each bootstrap fit $\ell = 1, \dots, L$ produces its own set of $\widehat K^{(\ell)}$ modes, which are matched to the baseline modes via minimum-cost assignment (Hungarian algorithm, \citealp[]{kuhn1955hungarian}). Then, we form confidence regions from the empirical distribution of matched modes, and monitor mode recovery with the stability score
\begin{equation}
s_k  =  \frac{1}{L}\sum_{\ell=1}^L \mathds{1}\{\text{mode }k\text{ is matched in replicate }\ell\}.
\label{eq:stability_score_def}
\end{equation}

\LinesNumbered
\begin{algorithm}[H]
  \caption{Bootstrap uncertainty quantification for mode estimation} \label{alg:bootstrap_uuq}
  \textsc{Input:} samples $\Xb_{1:N}$, temperature $\omega_0$, number of replicates $L$. \\
  \textsc{Fit} GERVE at $\omega_0$ on $\Xb_{1:N}$, prune/merge to obtain baseline modes $\{\widehat\mub_k^{
  (0)}\}_{k=1}^{\widehat K}$. \\
  \For{$\ell=1\!:\!L$}{
    \textsc{Sample} $J^{(\ell)}_{1:N} \overset{\text{i.i.d.}}{\sim} \mathrm{Unif}([N])$, set $\Xb^{(\ell)}_i \gets \Xb_{J^{(\ell)}_i}$ for $i=1\!:\!N$. \\
    \textsc{Fit} GERVE at $\omega_0$ on $\Xb^{(\ell)}_{1:N}$, prune/merge to obtain $\{\widehat\mub^{(\ell)}_j\}_{j=1}^{\widehat K^{(\ell)}}$. \\
    \textsc{Match} $\{\widehat\mub^{(\ell)}_j\}_{j=1}^{\widehat K^{(\ell)}}$ to $\{\widehat\mub_k^{(0)}\}_{k=1}^{\widehat K}$ (Hungarian). \\
    \textsc{Record} for each $k$: the matched $\widehat\mub^{(\ell)}_k$. 
  } 
  \textsc{Return:} per-mode confidence regions from $(\widehat\mub^{(\ell)}_k)_{\ell=1}^L$ and stability scores $(s_k)_{k=1}^{\widehat K}$. 
\end{algorithm}

Next, we establish a bootstrap validity result for the locally stable component-matching map. The practical implementation in Algorithm~\ref{alg:bootstrap_uuq} additionally incorporates pruning and merging, which introduce post-processing beyond the matching step considered in the following analysis.

\subsection{Statistical validity}
\label{sec:uq-theory}

Let $\widehat{\mathcal L}_{\omega_0,N}$ be the empirical objective at fixed $\omega_0>0$ on a compact parameter set $\Theta_K$ (eq.~\eqref{eq:compact_mixtures}), and let 
$\widehat\Lambdab_{\omega_0,N} \in \argmax_{\Lambdab\in\Theta_K}\widehat{\mathcal L}_{\omega_0,N}(\Lambdab)$ with a unique population maximizer $\Lambdab^\star_{\omega_0}$.
Write $P^\ast$ for the conditional probability given the observed sample $\Xb_{1:N}$.
Bootstrap limits are understood conditionally on $\Xb_{1:N}$, in $P$-probability. Define the bootstrap criterion and maximizer
\[
\widehat{\mathcal L}^\ast_{\omega_0,N}(\Lambdab)=E_{P_N^\ast}[q_{\Lambdab}(\Xb)]+\omega_0 \mathcal H_{\mathcal S}(q_{\Lambdab}),\qquad
\widehat\Lambdab^\ast_{\omega_0,N}\in\argmax_{\Lambdab\in\Theta_K}\widehat{\mathcal L}^\ast_{\omega_0,N}(\Lambdab).
\] 
In addition to unicity, we assume that $\Lambdab^\star_{\omega_0} \in \operatorname{int}(\Theta_K)$ and that the population maximizer satisfies $\Hb^\star_{\omega_0} = -\nabla_{\Lambdab}^2\mathcal{L}_{\omega_0}(\Lambdab^\star_{\omega_0}) \succ 0$, so that Theorem~\ref{thm:consistency_main} fully applies.

The following statements guarantee the validity of bootstrap for parameters (Thm.~\ref{thm:bootstrap-params-max-cor_main}), matched modes (Thm.~\ref{thm:modes-CLT-bootstrap_main}), and their robustness to inexact maximization (Prop.~\ref{prop:inexact_main}). All results are proved in Supplementary Material~\ref{app:proofs_uq}. In Supplementary Material~\ref{app:uq}, we also investigate the consistency of stability scores and propose strategies to handle nonconvex objectives. 

\begin{theorem}[Parameter-level bootstrap validity]
\label{thm:bootstrap-params-max-cor_main}
We have, as $N\to \infty$: 
\[
\sqrt{N}\big(\widehat\Lambdab^\ast_{\omega_0,N}-\widehat\Lambdab_{\omega_0,N}\big) \xrightarrow[]{\mathcal D} \mathcal N(0,\Wb_{\omega_0})
\quad\text{conditionally on }\Xb_{1:N},
\]
where $\Wb_{\omega_0} := \big(\Hb_{\omega_0}^\star\big)^{-1} \Vb_{\omega_0} \big(\Hb_{\omega_0}^\star\big)^{-1}$, and $\Vb_{\omega_0}:=\operatorname{Var} \big(\nabla_{\Lambdab} q_{\Lambdab}(\Xb)\big)\big|_{\Lambdab=\Lambdab^\star_{\omega_0}}$.
\end{theorem}

For the theoretical analysis, we fix deterministic target locations $\ub_{1:K_0}$, which act as reference anchors for defining the matching map. In practice, these targets correspond to the subset of baseline mode estimates we intend to track, so $K_0$ is taken no larger than the number $\widehat{K}$ of resolved baseline modes returned by the initial fit. To move from parameters to matched modes, we then define the matching map $\mathcal{M}$ that stacks the $K_0$ component means closest to their respective targets $\ub_{1:K_0}$ (see Suppl.~\ref{app:uq_notation_separation} for more details).
The following guarantees only require a local separation margin.

\begin{assumption}[Local separation for matching]
\label{as:sep-match_main}
There exists $\delta>0$ such that, at
$\Lambdab^\star_{\omega_0}$, each target $\ub_j$ has a unique nearest
component mean separated from all other component means by a distance gap
of at least $\delta$, and the corresponding nearest-component indices are
pairwise distinct (see Suppl.~\ref{app:uq_notation_separation} for a formal definition).
\end{assumption}

When separation fails, bootstrap confidence regions may inflate and the stability score $s_k$ may drop, which is informative about ambiguous modes (see Suppl.~\ref{app:uq} for a consistency result for $s_k$).

\begin{theorem}[Matched-modes bootstrap validity and confidence regions]
\label{thm:modes-CLT-bootstrap_main}
Under Assumption~\ref{as:sep-match_main}, $\mathcal M$ is $\mathcal C^1$ at $\Lambdab^\star_{\omega_0}$ with Jacobian $\Jb$, and as $N\to \infty$:
\[
\sqrt N\big(\mathcal M(\widehat\Lambdab_{\omega_0,N})-\mathcal M(\Lambdab^\star_{\omega_0})\big)\xrightarrow[]{\mathcal D} \mathcal N(0,\Cb_{\mathcal M}),\quad
\sqrt N\big(\mathcal M(\widehat\Lambdab^\ast_{\omega_0,N})-\mathcal M(\widehat\Lambdab_{\omega_0,N})\big)\xrightarrow[]{\mathcal D} \mathcal N(0,\Cb_{\mathcal M}),
\]
where $\Cb_{\mathcal M}=\Jb \Wb_{\omega_0} \Jb^T$. For each matched mode $j$, the marginal covariance is $\Cb_j$ and bootstrap confidence regions are asymptotically valid.
\end{theorem}

Finally, we allow inexact baseline and bootstrap fits. Let $\widetilde\Lambdab_{\omega_0,N}$ and $\widetilde\Lambdab^\ast_{\omega_0,N}$ denote any approximate solutions that are consistent, that is $\|\widetilde\Lambdab_{\omega_0,N} - \Lambdab^\star_{\omega_0}\| = o_{P}(1)$ and $\|\widetilde\Lambdab^\ast_{\omega_0,N} - \Lambdab^\star_{\omega_0}\| = o_{P^\ast}(1)$, and first-order stationary up to $o_P(N^{-1/2})$ and $o_{P^\ast}(N^{-1/2})$, respectively, that is
$\|\nabla_\Lambdab \widehat{\mathcal L}_{\omega_0,N}(\widetilde\Lambdab_{\omega_0,N})\| = o_P(N^{-1/2})$ and
$\|\nabla_\Lambdab \widehat{\mathcal L}^\ast_{\omega_0,N}(\widetilde\Lambdab^\ast_{\omega_0,N})\| = o_{P^\ast}(N^{-1/2})$.

\begin{proposition}[Robustness to inexact maximization]
\label{prop:inexact_main}
The conclusions of Theorems~\ref{thm:bootstrap-params-max-cor_main} and~\ref{thm:modes-CLT-bootstrap_main} hold with $\widetilde\Lambdab_{\omega_0,N}$ and $\widetilde\Lambdab^\ast_{\omega_0,N}$ in place of $\widehat\Lambdab_{\omega_0,N}$ and $\widehat\Lambdab^\ast_{\omega_0,N}$.
\end{proposition}

\section{Simulation studies}
\label{sec:simulations}

We investigate the ability of GERVE to recover modes and induce useful clusterings, and the effects of moderate overspecification of the variational mixture. On simulated data, we illustrate (i) how clustering emerges as a byproduct of variational mode estimation and (ii) how mode recovery improves with sample size.

\subsection{Clustering example}
\label{sec:simu_clustering}

To illustrate the clustering properties of GERVE, we simulate $N = 6000$ points from a 2D, 3-component Gaussian mixture with means at the nodes of an equilateral triangle, with isotropic covariance $\sigma^2 \Ib$, $\sigma^2 = 0.25$. The data exhibits $J=3$ high-density regions that we aim to recover. Details on the setup and a density plot can be found in Supplementary Material~\ref{app:simulations_cluster}.

We run GERVE with Gaussian mixtures (full covariances) under two regimes: matched $K = 3$ and overcomplete $K = 7$. To showcase GERVE's robustness to component mean initialization, we initialize means in $[-2,0]^2$. We start with broad covariances and equal weights. We use a continuously decreasing annealing schedule, and compensate with step size $\rho_t \propto \omega_t^{-0.7}$ (full hyperparameters and schedules in Suppl.~\ref{app:simulations_cluster}). After training, assign by highest responsibility and prune weights $<10^{-3}$. Baselines are GMM-EM (Gaussian mixture model fit by the EM algorithm) with full covariances and $k$-means for $K\in\{3,7\}$.

Figure~\ref{fig:clusters_k7} shows the obtained partitions when $K=7$. With this overcomplete mixture, GERVE still yields 3 effective clusters: redundant components coalesce at modes or are pruned, whereas GMM-EM returns 7 clusters, so does $k$-means (Suppl.~\ref{app:simulations_cluster}, Fig.~\ref{fig:clusters_k3}). Because GERVE targets modal structure rather than a clustering likelihood, it is more robust to over-specification. With $K=3$, all methods recover 3 clusters (Suppl.~\ref{app:simulations_cluster}, Fig.~\ref{fig:clusters_k3}). 

\begin{figure}
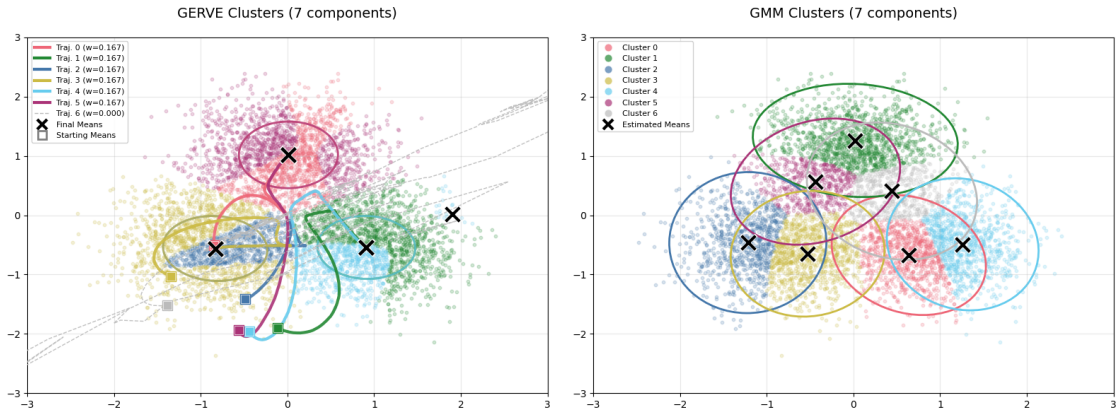

\centering
\includesvg[width=0.49\linewidth]{figures/cluster_gerve_7.svg}
\includesvg[width=0.49\linewidth]{figures/cluster_gmm_7.svg}
\caption{Overspecified clustering ($K=7$) of a Gaussian triangle mixture sample. Data points are colored according to the component with the highest posterior responsibility in the fitted mixture. Ellipses represent the component covariances. Left: GERVE returns 3 effective clusters (the component with dashed trajectory has vanishing weight, the remaining six components form three clusters by grouping equivalent components). Right: GMM-EM returns a partition into 7 clusters. 
}
\label{fig:clusters_k7}
\end{figure}

In Supplementary Material~\ref{app:real_datasets}, we assess robustness against classical and state-of-the-art methods \citep{chacon2020modal} with a benchmark on standard UCI datasets. We also provide an ablation study and report runtimes.

\subsection{Mode-finding performance}

We use two settings: (A) the same symmetric Gaussian triangle setting as previously, with sharper components ($\sigma^2=0.1$), (B) a skewed variant with the same three-mode geometry, introducing skewness and anisotropy, see Supplementary Section~\ref{app:simulations_mode} for more details. In both settings, true global modes are approximately located at the component modes.

Over $n_{\text{rep}}=100$ replicates, each method outputs a set of mode estimates. We report three metrics (see Suppl.~\ref{app:simulations_mode}):
(i) \emph{Mode recovery} ($\mathsf{MR}_\epsilon$), the number of true modes recovered within a tolerance $\epsilon$, we set $\epsilon=10^{-2}$; 
(ii) \emph{Hungarian matching} sum ($\mathsf{HM}$), minimum-assignment cost between true and estimated modes;
(iii) \emph{Nearest-neighbor} sum ($\mathsf{NN}$), aggregate distance from each estimate to the closest mode.

We run GERVE (Gaussian mixture with full covariances) with $K\in\{3,7\}$ and a vanishing annealing schedule. We focus on baselines that can return multiple candidate modes in multivariate settings, so that their outputs can be compared directly with the set of modes produced by GERVE: Gaussian mean-shift with fixed-point updates, with modes aggregated from $K$ initializations; the Feature Significance method via the \texttt{feature} R package, which detects significant curvature points, then estimates modes as the centroids of $k$-means; KDE+NVA, an alternative method that first estimates the density by Gaussian KDE and then applies Natural Variational Annealing (NVA)~\citep{leminh2025natural} to the estimated density; and the Local Mode Testing Algorithm (LMTA) of \citet{genovese2016non}, which identifies KDE candidate modes and retains those with significantly negative curvature using data splitting and an elementary-symmetric-polynomial bootstrap. For each sample size $N\in\{2^{10},2^{12},\dots,2^{20}\}$ and value of $K$, we perform a light hyperparameter grid search where applicable. Hyperparameters are selected using the known true modes, with the aim of comparing attainable mode-recovery performance. To isolate the effect of the preliminary KDE step, KDE+NVA uses the selected Gaussian mean-shift bandwidths and GERVE hyperparameters, while LMTA uses its own bandwidth-selection criterion (Suppl.~\ref{app:simulations_mode}). Due to computational cost, KDE+NVA and LMTA results are reported only up to $N=2^{14}$ and $N=2^{16}$, respectively. We report means for $\mathsf{MR}_\epsilon$ and $\mathsf{NN}$ and medians for $\mathsf{HM}$, with 95\% confidence intervals.

Fig.~\ref{fig:gerve_ms_fs} summarizes performance on the two settings with respect to sample size $N$ for $K\in\{3,7\}$.
On the Gaussian triangle (A), GERVE’s $\mathsf{MR}_\epsilon$ increases towards the number of true modes and $\mathsf{HM}$ decreases as $N$ grows, which provides evidence that mode recovery improves as the sample size increases. With $K=3$, GERVE's $\mathsf{NN}$ is comparable to that of mean-shift, Feature Significance and LMTA, showcasing its mode-location accuracy. With $K=7$, GERVE and KDE+NVA improve $\mathsf{MR}_\epsilon$ and $\mathsf{HM}$ while keeping $\mathsf{NN}$ competitive. This exhibits the variational mixture robustness to overspecification. In addition, GERVE's KDE-free and mini-batch implementation allows it to be evaluated at the largest sample sizes considered, whereas KDE+NVA and LMTA become computationally prohibitive earlier. On the skew triangle (B), all methods struggle to recover the modes with the $\epsilon = 10^{-2}$ tolerance. Nevertheless, the distance-based accuracy of GERVE, mean-shift, and LMTA slightly improves with $N$. GERVE remains among the most competitive methods, with increasing $\mathsf{MR}_\epsilon$ and decreasing $\mathsf{HM}$ and $\mathsf{NN}$ as the sample size grows.

\begin{figure}
\centering
\includegraphics[width=\linewidth]{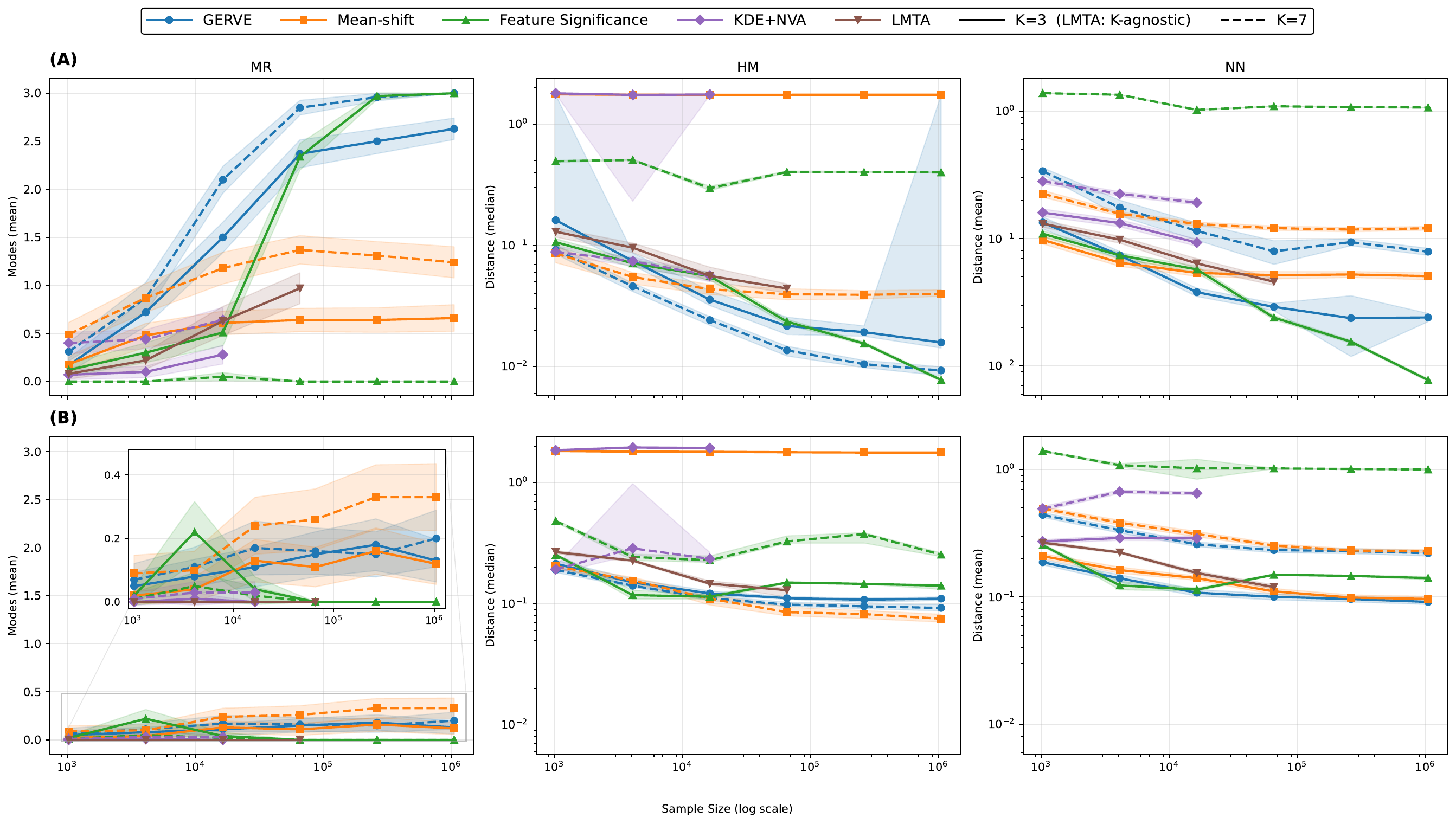}
\caption{Mode estimation vs.\ sample size $N$ for the (A) Gaussian triangle, (B) skew triangle examples ($I=3$). Curves: means ($\mathsf{MR}_\epsilon$, \textsf{NN}) or medians (\textsf{HM}) over $n_{\text{rep}}=100$; bands: 95\% confidence intervals.
}
\label{fig:gerve_ms_fs}
\end{figure}

To summarize the key aspects highlighted in this simulation study: (i) Clustering emerges from mode estimation and remains robust to moderate overspecification; (ii) For mode estimation, GERVE is competitive in location accuracy, improves as the sample size increases, and moderate overspecification helps relative to the other baselines.

\section{Application: UK road collision hotspots detection}
\label{sec:real_datasets}

Identifying collision hotspots (spatial intensity modes) and quantifying their uncertainty for intervention planning is crucial in road safety studies. Count-based metrics and KDE are common \citep{xie2008kernel, okabe2009kernel}. Spatial scans add significance testing and Bayesian models add uncertainty, but at the cost of parametric assumptions and heavier computation \citep{kulldorff1997spatial, aguero2006spatial}. GERVE offers a likelihood-free alternative that estimates mode locations and provides localization uncertainty summaries.

We analyse UK \texttt{STATS19} collisions from the Department for Transport Road Safety Data portal (\textit{www.gov.uk/government/statistics/road-safety-data}), restricting to A-roads in Greater London (longitudes $[-0.54,0.33]$, latitudes $[51.28,51.70]$) and severity labelled ``fatal'' or ``serious'' from 2020 to 2024. 

We fit an overcomplete mixture, targeting $10$ hotspots and using $K=20$ components to enable overcompleteness, and resolve modes by pruning and merging (details in Suppl.~\ref{app:uk}). The stopping temperature $\omega^\dagger$ is chosen at the elbow, the largest $\omega$ where the resolved-mode count peaks (Suppl.~\ref{app:uk}, Fig.~\ref{fig:w0_component_counts_plot}). At $\omega^\dagger$, we compute $L=500$ bootstrap resamples, match bootstrap modes to the baseline by assignment with adaptive gates (Suppl.~\ref{app:uk}), and project to a metric system (EPSG:27700). For each mode $k$, we report a stability score $s_k$ and a $95\%$ Gaussian approximation ellipse computed from the empirical covariance of its matched
bootstrap centres.

Stability and ellipses can be used to target actions. Figure~\ref{fig:collision_with_ellipses_2024_zoomed} maps 2020-2024 hotspots with ellipses and $s_k$ in the zoomed-in window $[-0.15,0.15]\times[-0.075,0.075]$. High-stability locations ($s_k\ge 0.7$, arbitrary threshold) such as Shoreditch (ID: 12) and Elephant \& Castle (ID: 1) have ellipses of order $100$-$150$ m, which supports investigation targeted at the intersection. Medium-stability locations ($0.4 \le s_k < 0.7$) such as Aldgate (ID: 8) and Clapham HS (ID: 16) show wider ellipses that span multiple junctions, which suggests area-wide investigation. Table~\ref{tab:hotspots_2024_full} in Supplementary Material~\ref{app:uk} lists the entries, ranked by stability score.

\begin{figure}
\centering
\includegraphics[width=0.99\linewidth]{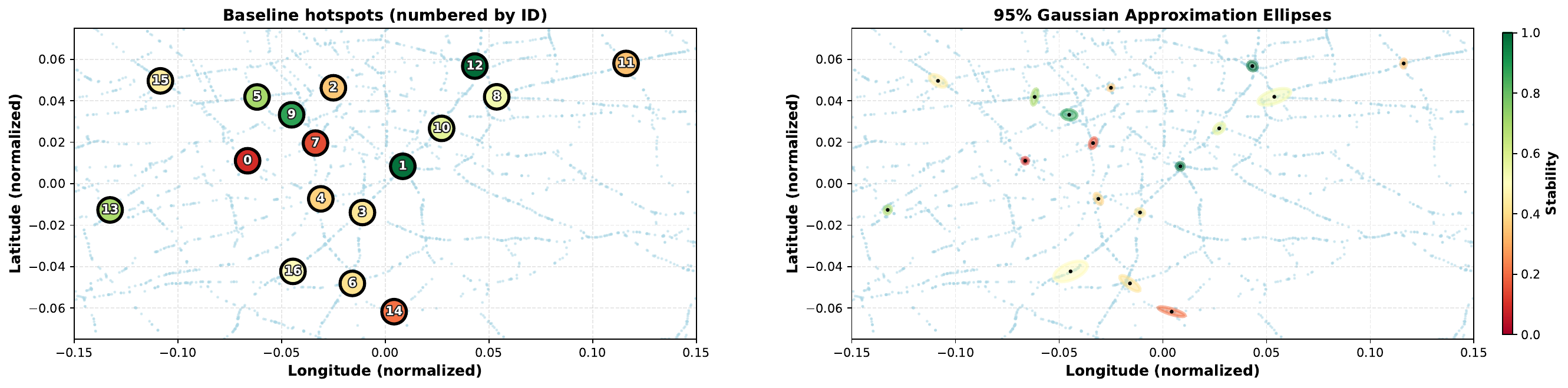}
\caption{Baseline collision hotspots identified by GERVE for the Greater London Area between 2020 and 2024, with stability scores. The normalized coordinates are bounded within $[-0.7,0.7]\times[-0.35,0.35]$ (normalized window) box, but all hotspots are in the $[-0.15,0.15]\times[-0.075,0.075]$ box (zoomed-in window), see Supplementary Material~\ref{app:uk}, Figure~\ref{fig:collision_with_ellipses_2024} for the full study area. Left: with hotspot IDs for reference to Table~\ref{tab:hotspots_2024_full}, Supplementary Material~\ref{app:uk}. Right: 95\% Gaussian approximation ellipses, in normalized coordinates. }
\label{fig:collision_with_ellipses_2024_zoomed}
\end{figure}

To compare like-for-like, we use mean-shift, operating on the same object as GERVE. With a flat kernel and a bandwidth sweep to a data-driven reference $h_0$ (see Suppl.~\ref{app:uk}), mean-shift enumerates many local maxima at small bandwidths and merges aggressively at larger ones. In the zoomed-in window, mean-shift at $0.05h_0$ yields 122 centres, while GERVE reports 17. The contrast reflects design: entropy-regularized mixtures impose global competition for support, so weaker peaks are absorbed as temperature falls. This yields a coarser and more stable partition that aligns with major junctions. A visual comparison appears in Supplementary Material~\ref{app:uk}, Figure~\ref{fig:meanshift_vs_gerve_combined}.

For temporal stability, we fit GERVE on 2014-2019 and 2020-2024 and label a hotspot persistent if its 95\% Gaussian approximation ellipse overlaps an ellipse from the other period. Eight of the 17 hotspots persist, which indicates stable risk across years. Non-persistent hotspots (Piccadilly Circus, Gunter Grove, Clapham HS) may reflect temporally varying factors such as pandemic-related traffic shifts, local construction, or policy changes. Supplementary Material~\ref{app:uk}, Figure~\ref{fig:hotspot_comparison_2019_vs_2024} overlays the two periods.

\section{Conclusion}
\label{sec:conclusion}

GERVE is a likelihood-free variational approach to mode estimation and modal clustering. It optimizes an entropy-regularized empirical objective with natural gradients under annealing, fitting mixtures that concentrate on high-density regions and induce clusters without explicit density estimation. As temperature goes to zero, optimal variational mixtures concentrate around the global modes of the population density. At fixed positive temperature, empirical maximizers exist, are consistent, and satisfy a central limit theorem. A matched-mode bootstrap procedure yields per-mode confidence sets.

Empirically, GERVE recovers modal structure and remains robust to moderate overspecification as temperature decreases, with some component weights becoming negligible or several components concentrating near the same mode. Diagonal-covariance variants preserve most of the behavior at lower cost, and an ablation study shows that annealing, moderate component overspecification, and short burn-in phases are useful levers.

GERVE is most suitable when uncertainty in mode locations is important, when the cluster count is misspecified, or when modal clustering is preferable to partition-based grouping. For problems where the number of clusters is known and computational speed is paramount, established task-specific clustering algorithms remain preferable. Three limitations merit attention. First, modal clustering may misalign with the group structure that matters in an application. Modes define clusters by attraction basins, but scientific groupings may depend on other criteria, so agreement is not guaranteed. Second, the optimization procedure can converge to local rather than global maxima. Careful initialization, temperature schedules, and overcomplete mixtures reduce this risk but do not completely eliminate it. Third, computation can be costly with full covariance matrices in high dimensions. In our experiments, using diagonal covariances allows one to control the cost while preserving mode separation.

The theory and experimental results position GERVE as a principled statistical framework for likelihood-free modal inference. By connecting modal clustering, mean-shift, variational annealing, and uncertainty quantification through a sample-based variational criterion, GERVE provides a common framework for modal analysis from samples. Looking ahead, three directions are natural. First, develop data-driven annealing schedules with principled stopping. Second, establish pathwise guarantees for empirical maximizers and for the validity of the bootstrap procedure under annealing. Third, characterize mode-component correspondence and label stability as temperature varies. Pursuing these directions will link the annealing procedure to stronger statistical guarantees and broaden the method’s practical reach.

\section*{Acknowledgements} This work is supported by the French ANR and JST-CREST Bayes-Duality grant (ANR-21-JSTM-0001). 

\newpage

\bibliographystyle{apalike}
\bibliography{biblio}

\newpage

\appendix

\begin{appendix}

\renewcommand{\thesection}{S\arabic{section}} 

\renewcommand{\thefigure}{S\arabic{figure}}
\setcounter{figure}{0}
\renewcommand{\thetable}{S\arabic{table}}
\setcounter{table}{0}
\renewcommand{\theequation}{S\arabic{equation}}
\setcounter{equation}{0}

\startcontents[appendices]
\printcontents[appendices]{l}{1}{\section*{Supplementary material}\setcounter{tocdepth}{3}}

\section{Theory details}
\label{app:theory}

\subsection{Notation and setup}

We work on a domain $\mathcal S \subset \mathbb R^d$, with finite Lebesgue measure.
Bold symbols denote vectors/matrices (e.g., $\xb,\mub,\Sigmab$).
Given a measurable $f:\mathcal S\to\mathbb R$ and a temperature $\omega>0$, define a normalized Gibbs measure as
\begin{equation*}
g_\omega(d\xb)=\frac{\exp(f(\xb)/\omega)}{Z_\omega(f)} d\xb,
\end{equation*}
with partition function
\begin{equation*}
Z_\omega(f)=\int_{\mathcal S}\exp(f(\yb)/\omega) d\yb
\end{equation*}
assumed finite when invoked.

We approximate $g_\omega$ on $\mathcal{S}$ with a tractable $q$ (restricted to $\mathcal{S}$):
(i) a single Gaussian $q_{\lambdab}=\mathcal N(\mub,\Sigmab)$,
(ii) a $K$-component Gaussian mixture
\begin{equation*}
q_{\Lambdab}(\xb)=\sum_{k=1}^K \pi_k \mathcal N(\xb\mid \mub_k,\Sigmab_k),
\qquad
\sum_{k=1}^K \pi_k=1, \pi_k>0.
\end{equation*}
Mixture weights are parameterized by logits $v_k=\log(\pi_k/\pi_K)$.
Unless stated otherwise, covariance eigenvalues are constrained to $[\sigma_{\min}^2,\sigma_{\max}^2]$.

For any density $q$, $\mathbb E_q[\cdot]$ denotes the expectation under $q$ and $\KL(q\|p)$ the Kullback--Leibler divergence. For $q$ on $\mathcal S$, define the entropy on $\mathcal S$:
\begin{equation*}
\mathcal H_{\mathcal S}(q)=-\int_{\mathcal S} q(\xb)\log q(\xb) d\xb,
\end{equation*}
and the variational objective
\begin{equation*}
\mathcal L_\omega(q) = \mathbb E_q[f(\Xb)] + \omega\mathcal H_{\mathcal S}(q).
\end{equation*}
By the classical variational identity,
\begin{equation*}
\mathcal L_\omega(q) = \omega\log Z_\omega(f)-\omega\KL\left(q || g_\omega\right),
\end{equation*}
so $g_\omega$ uniquely maximizes $\mathcal L_\omega$ whenever $Z_\omega(f)<\infty$.

For variational parameters $\varthetab$, let $\Fb(\varthetab)$ be the Fisher information matrix.
The natural gradient is
\begin{equation*}
\widetilde{\nabla}_\varthetab \mathcal L_\omega(\varthetab)
=\Fb(\varthetab)^{-1}\nabla_\varthetab \mathcal L_\omega(\varthetab).
\end{equation*}
If $q_\varthetab$ is in an exponential family with sufficient statistic $\Tb(\Xb)$ (or in an MCEF for mixtures \citep{lin2019fast}),
the natural gradient equals the standard gradient with respect to the expectation parameters  $\Mb=\mathbb E_{q_\varthetab}[\Tb(\Xb)]$:
\begin{equation*}
\widetilde{\nabla}_\varthetab \mathcal L_\omega(\varthetab)=\nabla_\Mb \mathcal L_\omega(\varthetab)] \; . 
\end{equation*}

Let $\{\xb_i^\star\}_{i=1}^I$ denote the set of global maximizers (``modes'') of $f$ (or of $p$ when $f=p$) on $\mathcal S$.
At a nondegenerate mode $\xb_i^\star$ (i.e. $\nabla^2 f(\xb_i^\star) \prec 0$), define the positive-definite curvature
\begin{equation*}
\Hb_i := -\nabla^2 f(\xb_i^\star) \succ 0.
\end{equation*}

For symmetric matrix $A$, we write $\mathrm{eig}_{\min}(A),\mathrm{eig}_{\max}(A)$ for extremal eigenvalues.
Norms are Euclidean for vectors and operator norms for matrices unless noted.

For a set $\mathcal{A} \subset \mathbb{R}^d$, $\mathrm{int}(\mathcal{A})$ denotes its interior, $\partial \mathcal{A}$ its boundary, and $\mathcal{A}^c := \mathbb{R}^d \setminus \mathcal{A}$ its complement. For an element $\xb \in \mathbb{R}^d$, we write $\mathrm{dist}(\xb,{\cal A}):=\inf_{\ab\in {\mathcal A}}\lVert \xb-\ab\rVert_2$. For $r > 0$, let $B(\xb, r)$ denote the ball centred on $\xb$ with radius $r$.

\subsection{Truncation approximation and compact sets}
\label{app:truncation_approximation}

In the optimization setting~\citep{leminh2025natural}, the target $f$ is known or directly evaluable, so $\mathbb{E}_{q_\varthetab}[f(\xb)]$ is computable. In GERVE, $f$ is replaced by an unknown density $p$ and we observe only i.i.d.\ samples $\Xb_{1:N}\sim p$. We suppose that $p$ is smooth, bounded, and supported on a bounded domain $\mathrm{supp}(p) \subset \mathcal{S}$. 

Using symmetry,
\begin{equation*}
\mathcal L_\omega(q_\varthetab)
= \int_\mathcal{S} p(\xb) q_\varthetab(\xb) d\xb + \omega \mathcal H_{\mathcal S}(q_\varthetab)
= \mathbb E_{p}\left[q_\varthetab(\Xb)\right] + \omega \mathcal H_{\mathcal S}(q_\varthetab),
\end{equation*} 
so the first term is an expectation under $p$. Let $P_N := \frac{1}{N}\sum_{i=1}^N \delta_{\Xb_i}$ be the empirical measure. A natural plug-in estimator replaces $\mathbb E_{p}[q_\varthetab(\Xb)]$ by $\mathbb E_{P_N}[q_\varthetab(\Xb)]$, yielding the empirical variational objective
\begin{equation}
\widehat{\mathcal{L}}_{\omega,N}(\varthetab)
= \mathbb E_{P_N}\left[q_\varthetab(\Xb)\right] + \omega \mathcal H_{\mathcal S}(q_\varthetab)
= \frac{1}{N}\sum_{i=1}^N q_\varthetab(\Xb_i) + \omega \mathcal{H}_{\mathcal S}(q_\varthetab).
\label{eq:empirical_objective_param}
\end{equation} 
This is the learning objective optimized by GERVE. 

\paragraph{Optimization on untruncated Gaussians.}
Since $p$ is supported on $\mathcal{S}$ rather than on $\mathbb{R}^d$, the corresponding Gibbs measure is defined only on $\mathcal{S}$, and one might a priori restrict $q$ to $\mathcal{S}$. $q_{\lambdab}$ denotes the untruncated Gaussian on $\mathbb{R}^d$ with natural parameter $\lambdab$. Define
\begin{equation*}
\psi_{\mathcal S}(\lambdab):=\int_{\mathcal S} q_{\lambdab}(\xb) d\xb.
\end{equation*}
The $\mathcal S$-truncated version is $q^{\mathcal S}_{\lambdab}(\xb):=q_{\lambdab}(\xb) \mathds{1}_{\mathcal S}(\xb)/\psi_{\mathcal S}(\lambdab)$.

Optimizing directly over $q^{\mathcal S}_{\lambdab}$ is cumbersome because $\psi_{\mathcal S}(\lambdab)$ varies with the parameter. The next result (proven in Section~\ref{app:proof_truncation_prop}) shows that on a large $\mathcal S$, optimizing over the untruncated family while integrating over $\mathcal S$ incurs only a controlled error on a compact set.

To define such a compact parameter space for Gaussians, fix a mean bound $\mub_{\max}<\infty$ and eigenvalue bounds $0<\sigma_{\min}^2<\sigma_{\max}^2<\infty$, and consider the compact set of natural parameters
\begin{equation*}
\Theta := \Bigl\{\lambdab=(\Sigmab^{-1}\mub,-\frac{1}{2} \Sigmab^{-1}) : \mub\in\mathbb R^d, \sigma_{\min}^2 \Ib \preceq \Sigmab \preceq \sigma_{\max}^2 \Ib,\lVert\mub\rVert_\infty \le \mub_{\max} \Bigr\}.
\end{equation*}
For $\lambdab \in \Theta$, let $\mub(\lambdab)$ and $\Sigmab(\lambdab)$ be the mean and covariance matrix encoded by $\lambdab$. 

For notational convenience, let 
\begin{equation*}
F_\omega(\lambdab) := \mathcal L_\omega(q_{\lambdab}),
\qquad G_\omega(\lambdab) := \mathcal L_\omega(q^{\mathcal S}_{\lambdab}),
\end{equation*}
for $\lambdab\in\Theta$.
Let $\varepsilon(\lambdab):=1- \psi_{\mathcal S}(\lambdab)$, $\varepsilon_{\max}:=\sup_{\lambdab\in\Theta}\varepsilon(\lambdab) <1$, and $p_{\max} = \sup_{\xb \in \mathcal{S}} p(\xb)$.
In the following proposition, we bound $\left\lvert G_\omega(\lambdab) - F_\omega(\lambdab) \right\rvert$ uniformly on $\Theta$.

\begin{proposition}[Truncation gap]
\label{prop:trunc_gap}
Let $C_{q}=\sup_{\lambdab \in \Theta}\sup_{\xb\in\mathcal S}\big|\log q_\lambdab(\xb)\big|<\infty$.
Define, for $t\in[0,1)$,
\begin{equation*}
\delta_\omega(t):= p_{\max} t +\omega \big\lvert\log(1-t)\big\rvert + \omega C_{q} t.
\end{equation*}

Then for all $\lambdab\in\Theta$,
\begin{equation}
\big\lvert G_\omega(\lambdab)-F_\omega(\lambdab) \big\rvert
\le \delta_\omega\big(\varepsilon(\lambdab)\big),
\label{eq:box_gap}
\end{equation}
and
\begin{equation*}
\sup_{\lambdab\in\Theta}
\big\lvert G_\omega(\lambdab)-F_\omega(\lambdab) \big\rvert
\le \delta_\omega(\varepsilon_{\max}).
\end{equation*}
\end{proposition}

This bound implies that if $\mub$ stays away from $\partial\mathcal{S}$, the error from using untruncated Gaussians is negligible. Now, we derive a tail bound and, combined with an argmax stability argument, we obtain the next theorem, proven in Section~\ref{app:proof_truncation_thm}.

\begin{theorem}[Truncation approximation on a bounded support]
\label{thm:trunc_equiv}
Let $\delta_\omega(\cdot)$ be as in Proposition~\ref{prop:trunc_gap}.
For any fixed $\omega>0$, suppose there exists $M_0>0$ such that,
for all $\lambdab\in\Theta$,
\[
\mub(\lambdab)\in\operatorname{int}(\mathcal S),
\qquad
\mathrm{dist}(\mub(\lambdab),\partial\mathcal S)
\ge
M_0\sqrt{\mathrm{eig}_{\max}(\Sigmab(\lambdab))}.
\]

Then for some $C > 0$,
\begin{equation*}
    \varepsilon_{\max}\le 2^{d/2} e^{-M_0^2/4}, \qquad \sup_{\lambdab\in\Theta} \big|G_\omega(\lambdab)-F_\omega(\lambdab)\big|\le\delta_\omega(\varepsilon_{\max}).
\end{equation*}

Moreover, let $m_\omega(\rho):=F_\omega(\lambdab_\omega^\star)- \sup_{\|\lambdab-\lambdab_\omega^\star\|\ge\rho}F_\omega(\lambdab)$ be the value gap. 
Consequently, if $m_\omega(\rho) > 0$ for some $\rho>0$,
then every maximizer of $G_\omega$ lies in $B(\lambdab_\omega^\star,\rho)$ whenever $\delta_\omega(\varepsilon_{\max})<m(\rho)/3$. 
\end{theorem}

\paragraph{Case of mixtures.}
Theorem~\ref{thm:trunc_equiv} is stated for single Gaussians.
For mixtures, the same argument applies componentwise. On the compact
parameter set $\Theta_K$ defined in equation~\eqref{eq:compact_mixtures},
the component parameters satisfy the same spectral bounds as in $\Theta$,
while the bounded pairwise log-odds ensure that the mixture weights are
uniformly bounded away from zero. Consequently, the mixture log-density
$\log q_{\Lambdab}$ is uniformly controlled on $\mathcal S$.

If the component means lie sufficiently far inside $\mathcal S$ relative
to their covariance scales, then the outside masses
$\varepsilon(\mub_k,\Sigmab_k)$ are uniformly small. Hence replacing
truncated components by untruncated ones induces a uniformly small
perturbation of $\widehat{\mathcal L}_{\omega,N}$ and
$\mathcal L_\omega$, and the corresponding population maximizers remain
close under the argmax-separation condition of
Theorem~\ref{thm:trunc_equiv}.

\subsection{Lexicographic order on $\Theta$}\label{subsec:Lexicographic-order-Theta}

Fix a one-to-one coordinate map $\tau:\Theta\to\mathbb{R}^{d_{\lambda}}$.
For $a,b\in\mathbb{R}^{d_{\lambda}}$ we write $a\prec_{\mathrm{lex}} b$
if there exists $j\in\{1,\dots,d_{\lambda}\}$ such that $a_{i}=b_{i}$
for all $i<j$ and $a_{j}<b_{j}$.
For $\lambdab,\lambdab'\in\Theta$ we write $\lambdab\prec_{\mathrm{lex}}\lambdab'$
if $\tau(\lambdab)\prec_{\mathrm{lex}}\tau(\lambdab')$.
We define the lexicographically ordered parameter space for $K$-component mixtures by
\[
\Theta_{K}^{\mathrm{lex}}
=\Bigl\{\Lambdab=(v_{1},\dots,v_{K-1},\lambdab_{1},\dots,\lambdab_{K})\in\Theta_{K}:
\lambdab_{1}\prec_{\mathrm{lex}}\lambdab_{2}\prec_{\mathrm{lex}}\cdots\prec_{\mathrm{lex}}\lambdab_{K}\Bigr\}.
\]

Let $\mathfrak{S}_{K}$ denote the permutation group on $\{1,\dots,K\}$.
The mixture model is invariant under relabelling of components, in the sense that for every
$\sigma\in\mathfrak{S}_{K}$ and every $\Lambdab=(\vb,\lambdab_{1},\dots,\lambdab_{K})\in\Theta_{K}$ there exists a relabelled parameter
\[
\sigma\cdot\Lambdab
:=\bigl(\vb^{\sigma},\lambdab_{\sigma(1)},\dots,\lambdab_{\sigma(K)}\bigr)\in\Theta_{K},
\]
such that the mixing weights are permuted accordingly,
$\pi(\vb^{\sigma})=(\pi_{\sigma(1)}(\vb),\dots,\pi_{\sigma(K)}(\vb))$,
and
\[
f(\cdot;\sigma\cdot\Lambdab)=f(\cdot;\Lambdab).
\]

Assume additionally that the component parameters are pairwise distinct, that is, $\lambdab_{i}\neq\lambdab_{j}$ for $i\neq j$.
Then, for every $\Lambdab\in\Theta_{K}$, there exists a unique permutation $\sigma_{\Lambdab}\in\mathfrak{S}_{K}$ such that
\[
\Lambdab^{\mathrm{ord}}:=\sigma_{\Lambdab}\cdot\Lambdab\in\Theta_{K}^{\mathrm{lex}},
\qquad
f(\cdot;\Lambdab^{\mathrm{ord}})=f(\cdot;\Lambdab).
\]
In particular, whenever a maximizer (or stationary point) exists over $\Theta_{K}$ for a permutation-invariant criterion,
there is an equivalent maximizer (or stationary point) lying in $\Theta_{K}^{\mathrm{lex}}$.
Thus, without loss of generality, one may restrict uniqueness arguments to $\Theta_{K}^{\mathrm{lex}}$.

Furthermore, suppose that the model is identifiable up to label switching on $\Theta_{K}$, namely,
\[
f(\cdot;\Lambdab)=f(\cdot;\Lambdab')
\Longrightarrow
\exists\,\sigma\in\mathfrak{S}_{K}:\ \Lambdab'=\sigma\cdot\Lambdab.
\]
Then the parametrization is identifiable on $\Theta_{K}^{\mathrm{lex}}$, in the sense that
$\Lambdab,\Lambdab'\in\Theta_{K}^{\mathrm{lex}}$ and $f(\cdot;\Lambdab)=f(\cdot;\Lambdab')$ imply $\Lambdab=\Lambdab'$.

To justify the existence and uniqueness of $\sigma_{\Lambdab}$, observe first that because $\tau$ is one-to-one and
$\prec_{\mathrm{lex}}$ is a strict total order on $\mathbb{R}^{d_{\lambda}}$, the induced relation $\prec_{\mathrm{lex}}$
on $\Theta$ is also a strict total order. In particular, for any $\lambdab\neq\lambdab'$ in $\Theta$, exactly one of
$\lambdab\prec_{\mathrm{lex}}\lambdab'$ or $\lambdab'\prec_{\mathrm{lex}}\lambdab$ holds.

Fix $\Lambdab=(\vb,\lambdab_{1},\dots,\lambdab_{K})\in\Theta_{K}$ and assume that $\lambdab_{i}\neq\lambdab_{j}$ for $i\neq j$.
Define the rank of component $i$ by
\[
r(i):=1+\bigl|\{j\in\{1,\dots,K\}:\ \lambdab_{j}\prec_{\mathrm{lex}}\lambdab_{i}\}\bigr|.
\]
Since the $\lambdab_{i}$ are pairwise distinct and the order is total, the ranks $\{r(i)\}_{i=1}^{K}$ are exactly
$\{1,\dots,K\}$. Hence there exists a unique permutation $\sigma_{\Lambdab}\in\mathfrak{S}_{K}$ such that
\[
r\bigl(\sigma_{\Lambdab}(k)\bigr)=k,\qquad k=1,\dots,K,
\]
or equivalently,
\[
\lambdab_{\sigma_{\Lambdab}(1)}\prec_{\mathrm{lex}}\lambdab_{\sigma_{\Lambdab}(2)}\prec_{\mathrm{lex}}\cdots\prec_{\mathrm{lex}}\lambdab_{\sigma_{\Lambdab}(K)}.
\]
Define the ordered representative
\[
\Lambdab^{\mathrm{ord}}
=\sigma_{\Lambdab}\cdot\Lambdab
=\bigl(\vb^{\sigma_{\Lambdab}},\lambdab_{\sigma_{\Lambdab}(1)},\dots,\lambdab_{\sigma_{\Lambdab}(K)}\bigr).
\]
By construction, $\Lambdab^{\mathrm{ord}}\in\Theta_{K}^{\mathrm{lex}}$.
Uniqueness of $\sigma_{\Lambdab}$ follows because the sorting permutation of a list of distinct elements under a strict total order
is unique (see, e.g., \citealp[Sec.~8.1]{cormen2009clrs}).

By the assumed relabelling invariance, for every $\sigma\in\mathfrak{S}_{K}$ we have $f(\cdot;\sigma\cdot\Lambdab)=f(\cdot;\Lambdab)$,
and in particular $f(\cdot;\Lambdab^{\mathrm{ord}})=f(\cdot;\Lambdab)$.
Consequently, any permutation-invariant criterion $Q$ (for instance,
$Q(\Lambdab)=\int \psi(\xb,f(\xb;\Lambdab))\,d\xb$ or $Q(\Lambdab)=\sum_{i=1}^{n}\log f(\Xb_{i};\Lambdab)$, or any other function of
$f(\cdot;\Lambdab)$) satisfies
\[
Q(\Lambdab^{\mathrm{ord}})=Q(\Lambdab).
\]
Hence, if $\Lambdab^{\star}$ is a maximizer of $Q$ over $\Theta_{K}$, then $\bigl(\Lambdab^{\star}\bigr)^{\mathrm{ord}}$ is also a maximizer and lies
in $\Theta_{K}^{\mathrm{lex}}$.

Finally, to prove identifiability on $\Theta_{K}^{\mathrm{lex}}$, let $\Lambdab,\Lambdab'\in\Theta_{K}^{\mathrm{lex}}$ satisfy
$f(\cdot;\Lambdab)=f(\cdot;\Lambdab')$. By identifiability up to label switching, there exists $\sigma\in\mathfrak{S}_{K}$ such that
$\Lambdab'=\sigma\cdot\Lambdab$, and therefore
\[
(\lambdab'_{1},\dots,\lambdab'_{K})=(\lambdab_{\sigma(1)},\dots,\lambdab_{\sigma(K)}).
\]
But both $(\lambdab_{1},\dots,\lambdab_{K})$ and $(\lambdab'_{1},\dots,\lambdab'_{K})$ are strictly increasing under $\prec_{\mathrm{lex}}$.
Since the $\lambdab_{k}$ are distinct, the set $\{\lambdab_{1},\dots,\lambdab_{K}\}$ admits a unique strictly increasing ordering.
Therefore the only permutation $\sigma$ for which $(\lambdab_{\sigma(1)},\dots,\lambdab_{\sigma(K)})$ is strictly increasing is
$\sigma=\mathrm{id}$. It follows that $\lambdab'_{k}=\lambdab_{k}$ for all $k$, and hence $\Lambdab'=\Lambdab$.
This shows that the parametrization is identifiable on $\Theta_{K}^{\mathrm{lex}}$.

\subsection{Modal basins-responsibility cells agreement}
\label{app:agreement}

To investigate the link between GERVE's induced clustering and modal clustering, we compare two partitions near each mode $\xb_i^\star$ of a $\mathcal{C}^2$ density: the responsibility cells $A_k(\theta)$ of an isotropic Gaussian mixture
$q_\thetab(\xb)=\sum_{k=1}^K \pi_k q_{\mub_k,\sigma^2}(x)$, with $\thetab = (\mub_1, \dots, \mub_K)$, and the modal basins $\{\mathcal B_i\}$ of $p$.
We give simple, local conditions under which they agree on inner neighborhoods of the modes.

Let $\{\xb_i^\star\}_{i=1}^I$ be isolated modes of $p$ in $\mathrm{supp}(p) \subset \mathcal{S}$.
For $r>0$, write $B_i(r):=B(\xb_i^\star,r)$ the ball centred on $\xb_i^\star$ and with radius $r$, and $\tilde r:=r/2$.
Weights satisfy $\pi_k\in[\pi_{\min},\pi_{\max}]$ with $\pi_{\min}>0$, $\sum_k \pi_k=1$.
For all $k \in [K]$, define posterior component responsibilities $\mathsf{resp}_k(\xb;\thetab)=\dfrac{\pi_k q_{\mub_k,\sigma^2 I}(\xb)}{\sum_j \pi_j q_{\mub_j,\sigma^2 I}(\xb)}$ and
cells $\mathcal{A}_k(\thetab)=\{\xb: \mathsf{resp}_k(\xb; \thetab)\ge \mathsf{resp}_j(\xb;\thetab),  \forall j \in [K]\}$. For $i \in [I]$, the modal basin $\mathcal{B}_i$ is the set of all initial points $\xb$ such that the solution $\widetilde{\xb}(t)$ of the gradient-ascent flow $\dot{\xb} = \nabla p(\xb)$ with $\widetilde{\xb}(0) = \xb$ satisfies $\widetilde{\xb}(t) \rightarrow \xb^\star_{i}$ as $t \to \infty$

\begin{assumption}[Local strong log-concavity and separation]\label{as:local}
For each $i$ there exists $r_i>0$ and $\kappa_i>0$ such that
$-\nabla^2 \log p(\xb)\succeq \kappa_i \Ib$ for all $\xb\in B_i(r_i)$, and the closed balls
$\overline{B_i(r_i)} \subset \mathrm{supp}(p)$ are pairwise disjoint.
\end{assumption}

\begin{lemma}[Basin containment]\label{lem:basin}
Under Assumption~\ref{as:local}, $B_i(r_i)\subset \mathcal B_i$ for each $i$.
\end{lemma}
\begin{proof}[Proof of Lemma~\ref{lem:basin}]
Let $\widetilde{\xb}(t)$ be a trajectory of $\dot{\xb} = \nabla p (\xb)$ with $\widetilde{\xb}(0) \in B_i(r_i)$. Choose the Lyapunov function $V(\xb) = \lVert \xb - \xb_{i}^\star \rVert^2/2$. 

On $\mathrm{supp}(p)$, using $\dot{\xb} = \nabla p(\xb) = p(\xb) \nabla \log p(\xb)$, we have
\begin{equation*}
    \dot{V}(\xb) = \langle \dot{\xb}, \xb - \xb_i^\star \rangle = p(\xb) \langle \nabla \log p(\xb), \xb - \xb_i^\star \rangle.
\end{equation*}
Since $\nabla \log p(\xb_i^\star) = \nabla p(\xb_i^\star) = 0$, strong log-concavity (Assumption~\ref{as:local}) implies, 
for $\xb \in B_i(r_i) \setminus \{\xb_i^\star\}$,
\begin{equation*}
\langle \nabla\log p(\xb), \xb - \xb_i^\star \rangle \le - \kappa_i \|\xb - \xb_i^\star\|^2 < 0,
\end{equation*}
thus, with $m_i = \inf_{\xb \in \overline{B_i(r_i)}} p(\xb) > 0$,
\begin{equation*}
\dot{V}(\xb) \le -2 \kappa_i p(\xb) V(\xb) \le -2 \kappa_i m_i V(\xb).
\end{equation*}

By Gr\"onwall's inequality, we have
\begin{equation}
    V(\widetilde{\xb}(t)) \le  V(\widetilde{\xb}(0)) \exp(-2 \kappa_i m_i t),
    \label{eq:gronwall}
\end{equation}
as long as $\widetilde{\xb}(t)$ is in $B_i(r_i)$.

Finally, we prove by contradiction that $\widetilde{\xb}(t) \in B_i(r_i)$ for all $t \ge 0$. Suppose the existence of $T_{\mathrm{out}} = \inf_{t > 0} \{t : \widetilde{\xb}(t) \not\in B_i(r_i) \} < \infty$. In this case, by continuity of $\xb$, we have $\lVert \widetilde{\xb}(T_{\mathrm{out}}) - \xb_i^\star \rVert = r_i$.

But since $V(\widetilde{\xb}(0)) < r_i^2/2$, \eqref{eq:gronwall} gives for $t \in [0, T_{\mathrm{out}})$, 
\begin{equation*}
    V(\widetilde{\xb}(t)) < \frac{r_i^2}{2} \exp(-2 \kappa_i m_i t) < \frac{r_i^2}{2}.
\end{equation*}
This implies $\|\widetilde{\xb}(t) - \xb_i^\star\| < r_i$ for all $t \in [0, T_{\mathrm{out}})$, and this contradicts $\lVert \widetilde{\xb}(T_{\mathrm{out}}) - \xb_i^\star \rVert = r_i$.

Thus, we have $\widetilde{\xb}(t) \in B_i(r_i)$ for all $t \ge 0$ and~\eqref{eq:gronwall} ensures $V(\widetilde{\xb}(t)) \rightarrow 0$ as $t \to \infty$, hence $\widetilde{\xb}(t) \to \xb_i^\star$ and we conclude that $\widetilde{\xb}(0) \in \mathcal{B}_i$.
\end{proof}

\begin{assumption}[One-per-mode configuration]\label{as:oneper}
There is a configuration \(\thetab=(\mub_1,\ldots,\mub_I)\) with exactly one centre
\(\mub_i\in B_i(\tilde r_i)\) for each \(i\), and no other centres in \(\bigcup_{i=1}^I B_i(r_i)\).
Therefore, by convention, we relabel components so that centre \(i\) lies in \(B_i(\tilde r_i)\) and write \(\mathcal{A}_i(\thetab)\) for the \(i\)th cell.
\end{assumption}

\begin{assumption}[Local margin]\label{as:margin}
There exists \(\Delta>0\) such that for all \(i\) and all \(\xb\in B_i(\tilde r_i)\),
\begin{equation*}
\|\xb-\mub_j\| \ge \|\xb-\mub_i\|+\Delta\qquad \forall j\neq i.
\end{equation*}
\end{assumption}

\begin{assumption}[Scale separation]\label{as:sigma}
There exists \(\varepsilon\in(0,1)\) such that 
\begin{equation*}
\sigma \le \frac{\Delta}{\sqrt{2 \log \Psi(\varepsilon)}},
\qquad
\text{where } \Psi(\varepsilon) := \frac{(I-1) \pi_{\max}}{\varepsilon \pi_{\min}} > 1.
\end{equation*}
\end{assumption}

\begin{theorem}[Responsibility on balls]\label{thm:dom}
Under Assumptions~\ref{as:oneper}-\ref{as:sigma}, for each \(i\),
\begin{equation*}
B_i(\tilde r_i) \subset \mathcal{A}_i(\thetab),
\qquad
\text{and}\quad r_i(\xb;\thetab) \ge (1+\varepsilon)^{-1}, \quad \forall \xb\in B_i(\tilde r_i).
\end{equation*}
\end{theorem}
\begin{proof}[Proof of Theorem~\ref{thm:dom}]
Fix \(i\) and let \(\xb\in B_i(\tilde r_i)\). By Assumption~\ref{as:margin},
\(\|\xb-\mub_j\|-\|\xb-\mub_i\|\ge\Delta\) and
\(\|\xb-\mub_j\|+\|\xb-\mub_i\|\ge 2\|\xb-\mub_i\|+\Delta\). Combining with Assumption~\ref{as:sigma}, this gives
\begin{equation*}
\frac{q_{\mub_j,\sigma^2\Ib}(\xb)}{q_{\mub_i,\sigma^2\Ib}(\xb)}
=\exp \Big(- \frac{(\|\xb-\mub_j\|-\|\xb-\mub_i\|)(\|\xb-\mub_j\|+\|\xb-\mub_i\|)}{2\sigma^2}\Big)
\le \exp \Big(- \frac{\Delta^2}{2\sigma^2}\Big) \le \frac{1}{\Psi(\varepsilon)}.
\end{equation*}
Summing with weights gives
\(\sum_{j\ne i}  \frac{\pi_j q_{\mub_j}}{\pi_i q_{\mub_i}}\le  \frac{(I-1) \pi_{\max}}{\pi_{\min} \Psi(\varepsilon)} = \varepsilon\), hence
\(\mathsf{resp}_i(\xb;\thetab)\ge (1+\varepsilon)^{-1} > 1/2\), which means $\xb\in \mathcal{A}_i(\thetab)$.
\end{proof}

\begin{remark} Assumption~\ref{as:sigma} is in fact a sufficient condition for \(\mathsf{resp}_i(\xb;\thetab)\ge (1+\varepsilon)^{-1}\), which means competing components contribute a total probability mass $\le \varepsilon$ relative to the dominant component $i$. 
\end{remark}

\begin{corollary}[Local partition agreement]\label{cor:agree}
Under Assumptions~\ref{as:local}-\ref{as:sigma}, for each \(i\),
\begin{equation*}
\mathcal{A}_i(\thetab)\cap B_i(\tilde r_i) = \mathcal B_i\cap B_i(\tilde r_i) = B_i(\tilde r_i).
\end{equation*}
\end{corollary}
\begin{proof}[Proof of Corollary~\ref{cor:agree}]
Theorem~\ref{thm:dom} gives \(B_i(\tilde r_i)\subset \mathcal{A}_i(\thetab)\).
Lemma~\ref{lem:basin} gives \(B_i(\tilde r_i)\subset \mathcal B_i\).
\end{proof}

Agreement is guaranteed only on local neighborhoods of the modes $\bigcup_{i=1}^I B_i(\tilde r_i)$. In general, the responsibility cells $\{\mathcal{A}_k\}$ and the modal basins $\{\mathcal{B}_i\}$ globally disagree. Figure~\ref{fig:region_agreement} showcases a counterexample, highlighting the differences between modal clustering and high-responsibility assignments in a Gaussian mixture model. Remarkably, a mode may not even be in its corresponding component responsibility region, notably when components significantly overlap.

\begin{figure}
    \centering
\includegraphics[width=0.60\linewidth]{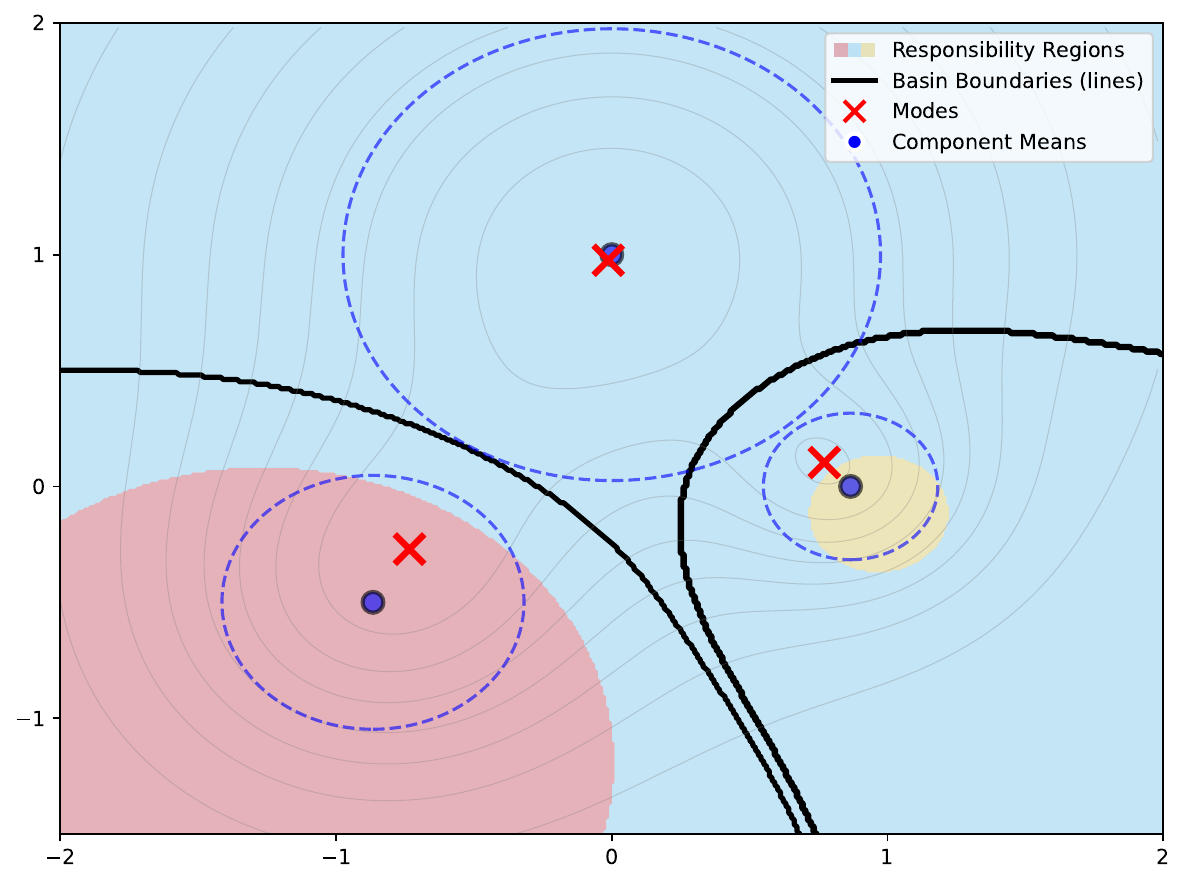}
    \caption{Responsibility regions (in color resp. blue, pink, yellow) and (approximate) modal basins for a Gaussian mixture with 3 components, means at $\mub_1=\big(-\cos(\pi/6),-0.5\big)$, $\mub_2=(0,1)$, $\mub_3=\big(\cos(\pi/6),0\big)$, mixture weights $\boldsymbol{\pi}=(0.16, 0.80, 0.04)$ and isotropic covariances $\Sigmab_k=\sigma_k^2 \Ib$ with variances $(\sigma_1^2,\sigma_2^2,\sigma_3^2)=(0.30, 0.95, 0.10)$. Covariances are represented as circles of radii $\sigma_k$.}
    \label{fig:region_agreement}
\end{figure}

\subsection{Natural-gradient convergence}
\label{app:snga_convergence}

Under the Robbins--Monro and bounded iterate conditions, we can show that stochastic natural-gradient ascent (SNGA) is consistent. Convergence is obtained via a direct application of the ODE method for stochastic approximation \citep[Chp. 2.1]{borkar2008stochastic}.  

\begin{assumption}
\label{ass:snatgrad}
Consider the update
\begin{equation}
\varthetab_{t+1} = \varthetab_t + \rho_t  \psib_t,
\label{eq:snga}
\end{equation}
where $\psib_t$ is an unbiased stochastic estimate of the natural gradient
$\widetilde\nabla \widehat{\mathcal{L}}_{\omega,N}(\varthetab_t)$ (equivalently, the gradient in the expectation-parameter space for exponential families or mixtures thereof). Let $\mathcal{F}_t := \sigma(\varthetab_1, \psib_1, \dots, \psib_{t})$.

Assume:
\begin{enumerate}
\item[(A1)] (\textbf{Robbins--Monro}) $\sum_t \rho_t = \infty$, $\sum_t \rho_t^2 < \infty$.
\item[(A2)] (\textbf{Mean field}) $\widetilde\nabla \widehat{\mathcal{L}}_{\omega,N}$ is locally Lipschitz and of at most linear growth.
\item[(A3)] (\textbf{Noise}) $\mathbb{E}[\psib_{t} \mid \mathcal{F}_{t-1}] = \widetilde\nabla \widehat{\mathcal{L}}_{\omega,N}(\varthetab_{t})$ and $\mathbb{E}\|\psib_{t}\|^2 \le C (1+\|\varthetab_{t}\|^2)$.
\item[(A4)] (\textbf{Stability}) The iterates $(\varthetab_t)_{t\ge 1}$ are a.s. bounded.
\end{enumerate}
\end{assumption}

\begin{theorem}[Convergence of SNGA at fixed $\omega$]
\label{thm:snatgrad}
Let $\mathcal{S}_\omega = \{\varthetab : \widetilde\nabla \widehat{\mathcal{L}}_{\omega,N}(\varthetab) = 0\}$. Under Assumption~\ref{ass:snatgrad}, we have, as $t\to \infty$:
\begin{equation*}
\mathrm{dist}(\varthetab_t, \mathcal{S}_\omega)  \to  0 \quad\text{a.s.}
\end{equation*}
Equivalently, every limit point of $(\varthetab_t)$ is almost surely stationary.
\end{theorem}

\begin{remark}[Natural-gradient regularity]
For exponential families and mixtures thereof (treated as MCEF models), the ``natural gradient = expectation gradient'' identity (eq.~\eqref{eq:natural_expectation_duality}) ensures that $\widetilde\nabla \widehat{\mathcal{L}}_{\omega,N}$ inherits Lipschitz continuity from $\widehat{\mathcal{L}}_{\omega,N}$ on compact $\Theta$, without explicit inversion of the Fisher information matrix.
\end{remark}

\begin{remark}[Projections]
GERVE performs optimization in a compact parameter space ($\Theta$ for single Gaussians, $\Theta_K$ for Gaussian mixtures), so the iterates must be projected to ensure that they remain in this space. In this case, the above convergence result is not directly applicable. However, in practice, projection is not always needed, typically when the compact space contains the solution to the unconstrained optimization problem.
\end{remark}

\subsection{Uncertainty quantification for mode estimation}\label{app:uq}

Here, we develop the theoretical guarantees of the mode-level uncertainty quantification (UQ) procedure proposed in Section~\ref{sec:uq}. In this section, we mathematically translate all the notions (including Assumption~\ref{as:sep-match_main}), defined in Section~\ref{sec:uq-theory}. These will be also used in the proofs in~\ref{app:proofs_uq}. Then, we investigate the consistency stability scores (Prop.~\ref{prop:stability-max}) and propose strategies to handle nonconvex objectives. All the results for this section and Section~\ref{sec:uq} of the main part are proven in Section~\ref{app:proofs_uq}.

\subsubsection{Notations and separation assumption}
\label{app:uq_notation_separation}

We start by mathematically translating all the notions defined in Section~\ref{sec:uq-theory}. These will be used throughout this section and in the proofs in~\ref{app:proofs_uq}. 

\paragraph{Data, measures, and convergence.}
$\Xb_{1:N}$ are i.i.d.\ from the population law $P$ on the sample space $\mathcal S$.
The empirical and bootstrap empirical measures are
\begin{equation*}
P_N= \frac1N\sum_{i=1}^N\delta_{\Xb_i},\qquad
P_N^\ast= \frac1N\sum_{i=1}^N\delta_{\Xb_i^\ast},\quad \Xb_1^\ast,\dots,\Xb_N^\ast \overset{\mathrm{i.i.d.}}{\sim} P_N \text{ given }\Xb_{1:N}.
\end{equation*}
We write $P^\ast(\cdot)=P(\cdot | \Xb_{1:N})$ for conditional (bootstrap) probability.
All bootstrap limits are understood \emph{conditionally on $\Xb_{1:N}$, in $P$-probability}.

\paragraph{Parameter space and objective.}
Fix $\omega_0>0$ and work on the compact parameter set $\Theta_K$ (eq.~\eqref{eq:compact_mixtures}). 
For a finite signed measure $G$ on $\mathcal S$, define the criterion
\begin{equation*}
\mathcal L_{\omega_0}(\Lambdab;G):=\int q_{\Lambdab}(\xb) dG(\xb)+\omega_0 \mathcal H_{\mathcal S}(q_{\Lambdab}),
\end{equation*}
and the population/empirical versions $\mathcal L_{\omega_0}(\Lambdab):=\mathcal L_{\omega_0}(\Lambdab;P)$ and
$\widehat{\mathcal L}_{\omega_0,N}(\Lambdab):=\mathcal L_{\omega_0}(\Lambdab;P_N)$.

\paragraph{Estimators and bootstrap refits.}
Let
\begin{equation*}
\widehat\Lambdab_{\omega_0,N}\in\argmax_{\Lambdab\in\Theta_K}\widehat{\mathcal L}_{\omega_0,N}(\Lambdab),
\qquad
\Lambdab^\star_{\omega_0}\in\argmax_{\Lambdab\in\Theta_K}\mathcal L_{\omega_0}(\Lambdab),
\end{equation*}
and define the bootstrap criterion and maximizer
\begin{equation*}
\widehat{\mathcal L}^\ast_{\omega_0,N}(\Lambdab):=\mathcal L_{\omega_0}(\Lambdab;P_N^\ast),
\qquad
\widehat\Lambdab^\ast_{\omega_0,N}\in\argmax_{\Lambdab\in\Theta_K}\widehat{\mathcal L}^\ast_{\omega_0,N}(\Lambdab).
\end{equation*}
We also allow a (possibly inexact) estimator $\widetilde\Lambdab_{\omega_0,N}$ of $\widehat\Lambdab_{\omega_0,N}$ (e.g., the output of a finite-iteration GERVE fit).

\paragraph{Quadratic expansion matrices.}
Write
\begin{equation*}
\Hb^\star_{\omega_0}:=-\nabla_{\Lambdab}^2\mathcal L_{\omega_0}(\Lambdab^\star_{\omega_0}),\qquad
\Vb_{\omega_0}:=\operatorname{Var} \big(\nabla_{\Lambdab}q_{\Lambdab}(X)\big)\big|_{\Lambdab=\Lambdab^\star_{\omega_0}},
\qquad
\Wb_{\omega_0}:=(\Hb^\star_{\omega_0})^{-1}\Vb_{\omega_0}(\Hb^\star_{\omega_0})^{-1}.
\end{equation*}

\paragraph{Mode means, matching, and projections.}
Here, we consider fixed targets $\ub_1,\dots,\ub_{K_0}\in\mathbb R^d$ ($K_0\le K$).
For $\Lambdab\in\Theta_K$, let $\mub_k(\Lambdab)\in\mathbb R^d$ be the $k$-th component mean and
\begin{equation*}
\mb(\Lambdab):=\big(\mub_1(\Lambdab)^T,\dots,\mub_K(\Lambdab)^T\big)^T\in\mathbb R^{Kd}.
\end{equation*}
Define the matching map $\mathcal{M}(\Lambdab)\in\mathbb R^{dK_0}$ by selecting an injection
$\pi:\{1,\dots,K_0\}\hookrightarrow\{1,\dots,K\}$ minimizing $\sum_{j=1}^{K_0}\|\mub_{\pi(j)}(\Lambdab)-\ub_j\|$,
with a fixed lexicographic tie-break:
\begin{equation*}
\mathcal{M}(\Lambdab)=\big(\mub_{\pi(1)}(\Lambdab)^T,\dots,\mub_{\pi(K_0)}(\Lambdab)^T\big)^T.
\end{equation*}
Let $\Eb_j\in\mathbb R^{d\times dK_0}$ extract the $j$-th $d$-block and set $\mathcal{M}_j(\Lambdab):=\Eb_j \mathcal{M}(\Lambdab)\in\mathbb R^d$.

\paragraph{Local separation and linearization.}
Let $\Pi_{K_0,K}$ denote the finite set of injections
$\pi:[K_0]\hookrightarrow[K]$. For $\pi\in\Pi_{K_0,K}$ and
$\Lambdab\in\Theta_K$, define the assignment cost
\begin{equation*}
    C_\pi(\Lambdab)
    :=
    \sum_{j=1}^{K_0}
    \|\mub_{\pi(j)}(\Lambdab)-\ub_j\|.
\end{equation*}
Let $\pi^\star(\Lambdab)$ denote a minimizer of $C_\pi(\Lambdab)$,
using the fixed lexicographic tie-break in the definition of
$\mathcal M$, and define the assignment gap
\begin{equation*}
    \Delta(\Lambdab)
    :=
    \min_{\pi\neq\pi^\star(\Lambdab)}
    \Bigl\{
    C_\pi(\Lambdab)
    -
    C_{\pi^\star(\Lambdab)}(\Lambdab)
    \Bigr\}.
\end{equation*}
For $\delta>0$, set
\[
    \mathcal G_\delta
    :=
    \{\Lambdab\in\Theta_K:\Delta(\Lambdab)\ge\delta\}.
\]
Local separation (Assumption~\ref{as:sep-match_main}) holds at
$\Lambdab^\star_{\omega_0}$ if there exists $\delta>0$ such that
$\Delta(\Lambdab^\star_{\omega_0})\ge\delta$.
In this case, the optimal assignment is locally constant and
$\mathcal M$ is $\mathcal C^1$ at
$\Lambdab^\star_{\omega_0}$
(Prop.~\ref{prop:mode-matching-smoothness}), with Jacobian
\[
    \Jb
    :=
    \nabla_{\Lambdab}
    \mathcal M(\Lambdab^\star_{\omega_0}),
\]
and we define
\begin{equation*}
    \Cb_{\mathcal M}
    :=
    \Jb\Wb_{\omega_0}\Jb^T,
    \qquad
    \Cb_j
    :=
    \Eb_j\Cb_{\mathcal M}\Eb_j^T.
\end{equation*}

\subsubsection{Consistency of stability scores}
\label{app:stability_scores}

To investigate the properties of the stability scores, we write their mathematical definition.
Let $\Phi(G)$ be a deterministic measurable selection from the (possibly set-valued) argmax
\begin{equation*}
\Phi(G) \in \argmax_{\Lambdab\in\Theta_K}\Big\{\int q_{\Lambdab} dG + \omega_0 \mathcal H_{\mathcal S}(q_{\Lambdab})\Big\}.
\end{equation*}
Fix disjoint neighborhoods $U_1,\dots,U_{K_0}$ of $\ub_1,\dots,\ub_{K_0}$.
Define $F_j(G)=\mathds{1}\{\Eb_j \mathcal{M}(\Phi(G))\in U_j\}$ and the targets
\begin{equation*}
\pi_{j,N} := P \big(F_j(P_N)=1\big),\qquad
\pi_j := \lim_{N\to\infty}\pi_{j,N}.
\end{equation*}
The bootstrap stability score defined in Section~\ref{sec:uq} can be reformulated
$s_j=\frac1L\sum_{\ell=1}^L F_j(P_N^{\ast(\ell)})$. In other words, 
\begin{equation*}
s_j  =  \frac1L\sum_{\ell=1}^L \mathds{1}\{\exists \text{ a matched component in }U_j \text{ for the }\ell\text{-th bootstrap fit}\}.
\end{equation*}

Let $\mathbb{E}^\ast$ and $\operatorname{Var}^\ast$ denote the conditional expectation and variance given $\Xb_{1:N}$. The following proposition, proven in Section~\ref{app:proof_stability}, establishes statistical guarantees on the stability scores.

\begin{proposition}[Consistency of stability scores]\label{prop:stability-max}
Let us index the stability scores $s_{j, L}$ by $L$. Under Assumptions~\ref{ass:consistency} and \ref{as:sep-match_main} and with $U_j$ chosen to have a positive margin at $\Lambdab^\star_{\omega_0}$ (Lemma~\ref{lem:indicator-continuity}), define
\begin{equation*}
\tau_{j,N} := \mathbb{E}^\ast \big[F_j(P_N^\ast)\big] = \mathbb{E}^\ast[s_{j, L}]\qquad\text{(for all $L$).}
\end{equation*}
Then:
\begin{enumerate}\itemsep3pt
\item[(i)] \textbf{(Conditional LLN).} For each $N$, conditionally on the data,
\begin{equation*}
s_{j,L} \xrightarrow{\text{a.s.}}  \tau_{j,N}\quad\text{as }L\to\infty,
\qquad
\operatorname{Var}^\ast(s_{j,L})=\operatorname{Var}^\ast \big(F_j(P_N^\ast)\big)/L \to 0 .
\end{equation*}
In particular, for any $\epsilon>0$,
\begin{equation*}
P^\ast \big(|s_{j,L}-\tau_{j,N}|>\epsilon\big) \le 2\exp \big(-2L\epsilon^2\big).
\end{equation*}
\item[(ii)] \textbf{(Large-$N$ limit).} As $N\to\infty$,
\begin{equation*}
\tau_{j,N} \xrightarrow{P} \pi_j,
\qquad
\pi_{j,N}:=P \big(F_j(P_N)=1\big) \to \pi_j,
\end{equation*}
where $\pi_j := F_j(P)\in\{0,1\}$ \emph{(equivalently, $\pi_j=P(F_j(P)=1)$)}.
\end{enumerate}
Consequently, for any sequence $L_N\to\infty$ as $N \to \infty$, we have $s_{j,L_N} \xrightarrow{P} \pi_j$.
\end{proposition}

\subsubsection{Handling local optima}

Although $\mathcal{L}_{\omega_0}$ is nonconvex in general, in practice we use the following guardrails to mitigate the issue:
\begin{enumerate}
    \item \emph{Annealing + overcompleteness.} We use an overcomplete model, and we start from a higher smoothing level and continue down to $\omega_0$. Small components are pruned. This helps GERVE explore different basins.
    \item \emph{Stability scores.} Per-mode stability scores $s_j$ are monitored, with values near 1 indicating robust recovery. For example, as a diagnostic, we can flag $s_j<0.5$ as potentially problematic and $s_j<0.3$ as likely unstable.
    \item \emph{Multiple initializations.} We can run several random starts at $\omega_0$ and keep the best solution (by objective). 
\end{enumerate}

\section{Proofs}
\label{app:proofs}

\subsection{Preliminaries: Empirical process properties}
\label{app:gc_property}

In this section, we establish the empirical-process properties used in
the asymptotic results of Sections~\ref{sec:framework} and~\ref{sec:uq}
and Supplementary Material~\ref{app:uq}. On the compact Gaussian-mixture
parameter space used throughout the asymptotic analysis, we show that
the density class and its first- and second-order derivative
classes are $P$-Donsker and $P$-Glivenko--Cantelli, with
square-integrable envelopes. The proof follows from standard
empirical-process results.

For i.i.d.\ observations $\Xb_1,\dots,\Xb_N\sim P$, write
\[
P_N f := \frac1N\sum_{i=1}^N f(\Xb_i).
\]
We recall that a class $\mathcal F$ is $P$-Glivenko--Cantelli if
\[
\sup_{f\in\mathcal F}|P_Nf-Pf|
\xrightarrow{P}0,
\]
and is $P$-Donsker if
\[
\mathbb G_N:=\sqrt N(P_N-P)
\xrightarrow[]{w} \mathbb G_P
\quad\text{in }\ell^\infty(\mathcal F),
\]
where $\xrightarrow[]{w}$ denotes weak convergence and $\mathbb G_P$ is the corresponding $P$-Brownian bridge.
An envelope of $\mathcal F$ is a measurable function $F$ satisfying
$|f|\le F$ for every $f\in\mathcal F$.

\begin{lemma}[Empirical-process regularity for Gaussian mixtures]
\label{lem:gc-donsker-mixture}
Suppose that $P$ is supported on a compact $\mathcal{S}\subset\mathbb{R}^d$.
For fixed $K$ and compact $\Theta_K\subset\mathbb{R}^{\dim(\Theta_K)}$, the classes
\[
\mathcal Q=\{q_{\Lambdab}:\Lambdab\in\Theta_K\},
\]
\[
\mathcal Q^{(u)}_\nabla
=
\{[\nabla_{\Lambdab}q_{\Lambdab}]_u:
  \Lambdab\in\Theta_K\},
\qquad u=1,\ldots,p,
\]
and
\[
\mathcal Q^{(u,v)}_{\nabla^2}
=
\{[\nabla^2_{\Lambdab}q_{\Lambdab}]_{uv}:
  \Lambdab\in\Theta_K\},
\qquad u,v=1,\ldots,p,
\]
are $P$-Donsker (hence Glivenko--Cantelli), with envelopes in
$L_2(P)$.
\end{lemma}

\begin{proof}
Since $P$ is supported on the compact $\mathcal{S}$, it suffices to control
the classes on $\mathcal{S}$. Let
$\widetilde\Theta_K:=\operatorname{conv}(\Theta_K)$, the convex hull of $\Theta_K$. The Gaussian
natural-parameter domain is open and convex and contains $\Theta_K$, therefore $\widetilde\Theta_K$ is a compact subset of this domain, that means, for any $\Lambdab,\widetilde\Lambdab\in\Theta_K$, the
segment joining them lies in $\widetilde\Theta_K$.

Gaussian-mixture densities are smooth in their parameters on this
domain. Since $\mathcal{S}\times\widetilde\Theta_K$ is compact,
$q_\Lambdab(\xb)$ and its derivatives with respect to $\Lambdab$ up
to the third order are uniformly bounded. Consequently,
$\mathcal Q$, $\mathcal Q_\nabla^{(u)}$, and
$\mathcal Q_{\nabla^2}^{(u,v)}$ admit bounded, hence $L_2(P)$,
envelopes.

Moreover, the mean value theorem gives
\[
|q_\Lambdab(\xb)-q_{\widetilde\Lambdab}(\xb)|
\le C_1\|\Lambdab-\widetilde\Lambdab\|,
\]
\[
|[\nabla_\Lambdab q_\Lambdab(\xb)]_u
 -[\nabla_\Lambdab q_{\widetilde\Lambdab}(\xb)]_u|
\le C_2\|\Lambdab-\widetilde\Lambdab\|,
\]
and
\[
|[\nabla_\Lambdab^2 q_\Lambdab(\xb)]_{uv}
 -[\nabla_\Lambdab^2 q_{\widetilde\Lambdab}(\xb)]_{uv}|
\le C_3\|\Lambdab-\widetilde\Lambdab\|,
\]
uniformly in $\xb \in \mathcal{S}$, where $C_1,C_2,C_3<\infty$.

Thus these are Lipschitz Euclidean parametric classes indexed by the
compact finite-dimensional set $\Theta_K\subset\Rset^{\dim(\Theta_K)}$, with
constant Lipschitz moduli. By Theorem~2.7.11 of
\citet{van1996weak},
\[
N_{[]}(\varepsilon,\mathcal F,L_2(P))
\le
N(c\varepsilon,\Theta_K,\|\cdot\|)
=O(\varepsilon^{-p})
\]
for each of the three classes $\mathcal F$ above and some
$c>0$. Thus,
\[
\int_0^\delta
\sqrt{\log N_{[]}(\varepsilon,\mathcal F,L_2(P))}\,d\varepsilon
<\infty.
\]
The bracketing Donsker theorem (\citealp{van1998asymptotic}, Thm.~19.5) implies that all three
classes are $P$-Donsker, hence $P$-Glivenko--Cantelli.
\end{proof}

\subsection{Proof of Theorem~\ref{thm:annealing_main}}
\label{app:thm_annealing}

We prove Theorem~\ref{thm:annealing_main}, establishing the consistency of the optimal variational solution $q_\omega^\star := \argmax_\mathcal{Q} \mathcal{L}_\omega(q)$, where $\mathcal{Q}$ is the family of truncated Gaussian mixtures of $K \ge I$ components with upper-bounded covariances, as $\omega \rightarrow 0$. We proceed in three steps. First, we identify the asymptotics of the Gibbs measure $g_\omega$ via a Laplace expansion near each mode of $f$ (Prop.~\ref{prop:laplace} and~\ref{prop:gibbs_mean_cov}). Then, we construct a truncated mixture $\widetilde{q}^\mathcal{S}_\omega$ on $\mathcal{S}$ such that $\KL(\widetilde{q}^\mathcal{S}_\omega || g_\omega) \rightarrow 0$ (Prop.~\ref{prop:kl_rate}). This implies $\KL(q_\omega^\star || g_\omega) \rightarrow 0$, so finally, combining with the Laplace expansions of $g_\omega$, we deduce the claims of Theorem~\ref{thm:annealing_main}.

\subsubsection*{Setup}

Let $f\in \mathcal{C}^3(\mathcal{S})$ be bounded with (finitely many) nondegenerate global modes $\{\xb_i^\star\}_{i=1}^I$, and let $f^\star := \sup_{\xb \in {\cal S}} f(\xb) = f(\xb_i^\star)$. Nondegeneracy means $\Hb_i = -\nabla^2 f(\xb_i^\star) \succ 0$, for all $i \in [I]$. These modes are isolated and in the interior of $\mathcal{S}$, so there exist disjoint open, bounded neighborhoods $U_1, \dots, U_I \subset \mathcal{S}$ of the $\xb_i^\star$ such that each $U_i$ contains no other critical point of $f$. Without loss of generality, assume that $f$ is strictly positive on $\bigcup_{i=1}^I U_i$. Define $U_0 :=\mathcal{S} \setminus \bigcup_{i=1}^I U_i$. 
The restricted objective over a family $\mathcal{Q}$ is
\begin{equation*}
\mathcal{L}_\omega(q) = \mathbb{E}_{q}[f(\Xb)] + \omega \mathcal{H}_\mathcal{S}(q), \qquad q \in \mathcal Q.
\end{equation*}
Here, the family $\mathcal{Q}$ is the $\mathcal{S}$-truncated Gaussian mixture family with $K \ge I$ components, and all eigenvalues of covariance matrices in $(0,\sigma_{\max}^2]$. Recall that for any Gaussian mixture $q_\Lambdab$ on $\mathbb R^d$, we write its $\mathcal{S}$-truncation
\begin{equation*}
q_\Lambdab^\mathcal{S}(\xb)
:= \frac{q_\Lambdab(\xb) \mathds{1}_\mathcal{S}(\xb)}
        {\int_\mathcal{S} q_\Lambdab(\yb) d\yb}.
\end{equation*} 

Let $q_\omega^\star$ denote any choice of maximizer of $\mathcal{L}_\omega(q)$ in $\mathcal{Q}$, for a fixed $\omega>0$. To simplify the notation, we denote the conditional distributions of $q_\omega^\star$ and $g_\omega$ on the regions $U_i$, $i = 0, \dots, I$, respectively by $\varphi_{i,\omega}$ and $\psi_{i,\omega}$, defined as follows:
\begin{equation}
\varphi_{i,\omega}(\xb)
:= q_\omega^{\star,U_i}(\xb) = \frac{q_\omega^\star(\xb)\mathds 1_{U_i}(\xb)}{\alpha_{i,\omega}},
\qquad
\psi_{i,\omega}(\xb)
:= g_\omega^{U_i}(\xb) = \frac{g_\omega(\xb)\mathds 1_{U_i}(\xb)}{\gamma_{i,\omega}},
\label{eq:varphi_psi}
\end{equation}
where
\begin{equation}
\alpha_{i,\omega} := \int_{U_i} q_\omega^\star(\xb) d\xb,
\qquad
\gamma_{i,\omega} := \int_{U_i} g_\omega(\xb) d\xb.
\label{eq:alpha_gamma}
\end{equation}

\subsubsection*{Step 1: Laplace expansion and local Gibbs measure properties}

We start by applying Laplace’s method around each isolated mode $\xb_i^\star$ of $f$ to the unrestricted Gibbs optimizer $g_\omega$. The two following propositions characterize the shape of the Gibbs measure under annealing. 

\begin{proposition}[Laplace expansion near a mode]
\label{prop:laplace}
Fix $i\in[I]$ and consider an open neighborhood $U_i$ of the mode
$\xb_i^\star$. Then, as $\omega\to0$,
\begin{equation*}
\int_{U_i}\exp\{f(\xb)/\omega\}\,d\xb
=
e^{f^\star/\omega}(2\pi\omega)^{d/2}
(\det\Hb_i)^{-1/2}\bigl(1+o(1)\bigr).
\end{equation*}
Moreover,
\begin{equation*}
Z_\omega(f)
=
e^{f^\star/\omega}(2\pi\omega)^{d/2}
\left\{\sum_{j=1}^I(\det\Hb_j)^{-1/2}\right\}
\bigl(1+o(1)\bigr).
\end{equation*}
and, for any bounded continuous function $h$,
\begin{equation*}
\int_{U_i}h(\xb)\,g_\omega(d\xb)
=
c_i h(\xb_i^\star)+o(1),
\qquad
c_i
:=
\frac{(\det\Hb_i)^{-1/2}}
{\sum_{j=1}^I(\det\Hb_j)^{-1/2}}.
\end{equation*}
Furthermore, $g_\omega \rightarrow \sum_{i=1}^I c_i\delta_{\xb_i^\star}$ weakly, and there exist $C,\eta>0$ such that $g_\omega(U_0)\le Ce^{-\eta/\omega}$.
\end{proposition}

\begin{proof}[Proof of Proposition~\ref{prop:laplace}]
The first display follows directly from the classical Laplace method
(\citealp{hwang1980laplace}; also see \citealp[Proof of Thm.~2]{leminh2025natural}).
The second display follows from the same argument with a continuous function $h$ as amplitude. After summing over the modes and normalizing, it follows that
\[
\int_{U_i}h(\xb)\,g_\omega(d\xb)
=
c_i h(\xb_i^\star)+o(1),
\qquad
c_i=
\frac{(\det\Hb_i)^{-1/2}}
{\sum_{j=1}^I(\det\Hb_j)^{-1/2}},
\]
and hence $g_\omega \rightarrow \sum_{i=1}^I c_i \delta_{\xb_i^\star}$.

Finally, since $\mathcal S$ is compact and $U_0$ contains no global
mode, there exists $\eta_0>0$ such that
\[
\sup_{\xb\in U_0}f(\xb)\le f^\star-\eta_0.
\]
This gives, for some $C,\eta>0$,
\[
g_\omega(U_0)\le Ce^{-\eta/\omega}.
\]
\end{proof}

\begin{proof}[Proof of Proposition~\ref{prop:laplace}]
The first display follows directly from the classical Laplace method
(\citealp{hwang1980laplace}; also see
\citealp[Proof of Thm.~2]{leminh2025natural}).

Since $\mathcal S$ is compact and $U_0$ contains no global mode,
there exists $\eta_0>0$ such that
\[
\sup_{\xb\in U_0}f(\xb)\le f^\star-\eta_0.
\]
Hence
\[
\int_{U_0}e^{f(\xb)/\omega}\,d\xb
\le
|\mathcal S|e^{(f^\star-\eta_0)/\omega}
=
o\!\left(e^{f^\star/\omega}\omega^{d/2}\right).
\]
Summing the local Laplace expansions over $i\in[I]$ therefore gives
\[
Z_\omega(f)
=
e^{f^\star/\omega}(2\pi\omega)^{d/2}
\left\{\sum_{j=1}^I(\det\Hb_j)^{-1/2}\right\}
\bigl(1+o(1)\bigr).
\]

The same Laplace argument with a bounded continuous function $h$ as
amplitude gives
\[
\int_{U_i}h(\xb)e^{f(\xb)/\omega}\,d\xb
=
e^{f^\star/\omega}(2\pi\omega)^{d/2}
(\det\Hb_i)^{-1/2}
\{h(\xb_i^\star)+o(1)\}.
\]
Dividing by the preceding expansion of $Z_\omega(f)$ yields
\[
\int_{U_i}h(\xb)\,g_\omega(d\xb)
=
c_i h(\xb_i^\star)+o(1).
\]
Using $g_\omega(U_0)\to0$, this implies $g_\omega\rightarrow\sum_{i=1}^Ic_i\delta_{\xb_i^\star}$ weakly.

Finally, the preceding inequalities give
\[
g_\omega(U_0)
\le
C_0\omega^{-d/2}e^{-\eta_0/\omega}.
\]
which means $g_\omega(U_0)\le Ce^{-\eta/\omega}$ for some $C,\eta>0$.
\end{proof}

\begin{proposition}
\label{prop:gibbs_mean_cov}
For $i \in [I]$ and $\psi_{i,\omega}$ defined in
equation~\eqref{eq:varphi_psi}, as $\omega \rightarrow 0$,
\begin{equation*}
    \mathbb{E}_{\psi_{i,\omega}}[\Xb]
    \rightarrow \xb_i^\star,
    \qquad
    \mathrm{Cov}_{\psi_{i,\omega}}(\Xb)
    \rightarrow 0.
\end{equation*}
\end{proposition}

\begin{proof}
By Proposition~\ref{prop:laplace}, applied with $h\equiv1$, we have $\gamma_{i,\omega}:=g_\omega(U_i)=c_i+o(1)$, with $c_i>0$. For each $j\in[d]$, applying Proposition~\ref{prop:laplace} to
$h_j(\xb)=x_j$ gives
\[
\int_{U_i}x_j\,g_\omega(d\xb)
=
c_i x_{i,j}^\star+o(1).
\]
Therefore,
\[
\mathbb{E}_{\psi_{i,\omega}}[X_j]
=
\frac{\int_{U_i}x_j\,g_\omega(d\xb)}
{\gamma_{i,\omega}}
\longrightarrow x_{i,j}^\star,
\]
and thus,
\[
\mathbb{E}_{\psi_{i,\omega}}[\Xb]\longrightarrow\xb_i^\star.
\]

Next, for $j,k\in[d]$, let $h_{jk}(\xb) = (x_j-x_{i,j}^\star)(x_k-x_{i,k}^\star)$. Since $h_{jk}(\xb_i^\star)=0$, Proposition~\ref{prop:laplace} gives
\[
\int_{U_i}h_{jk}(\xb)\,g_\omega(d\xb)=o(1).
\]
Dividing by $\gamma_{i,\omega}\to c_i>0$ gives
\[
\mathbb{E}_{\psi_{i,\omega}}
\!\left[
(X_j-x_{i,j}^\star)(X_k-x_{i,k}^\star)
\right]
\longrightarrow0.
\]
Using
$\mathbb{E}_{\psi_{i,\omega}}[\Xb]\to\xb_i^\star$, this implies
\[
\mathrm{Cov}_{\psi_{i,\omega}}(\Xb)\longrightarrow0.
\]
\end{proof}

\subsubsection*{Step 2: Competitor mixture and KL convergence}

For $\omega > 0$, define the Gaussian distributions, for $i \in [I]$: 
\begin{equation*}
    \widetilde{q}_{i,\omega}(\xb) = \mathcal{N}(\xb; \xb_i^\star, \omega \Hb_i^{-1}).
\end{equation*}
These Gaussian dsitributions approximate locally the Gibbs measure at each mode. 
We define the associated mixture with Laplace weights
\begin{equation*}
    \widetilde{q}_\omega(\xb) = \sum_{i=1}^I
 c_i \widetilde{q}_{i,\omega}(\xb),
\end{equation*}
where $c_i \propto (\det \Hb_i)^{-1/2}$ and $\sum_i c_i = 1$.
We use the truncated mixture $\widetilde{q}^{\mathcal{S}}_\omega := \widetilde{q}_\omega \mathds{1}_\mathcal{S} / \int_{\mathcal{S}} \widetilde{q}_\omega$ as a competitor for $q_\omega^\star$ in $\mathcal{Q}$. In this step, we prove that $\widetilde{q}^\mathcal{S}_\omega$ converges in KL to $g_\omega$. 

\begin{proposition}
\label{prop:kl_rate}
Let $\widetilde{q}^{\mathcal{S}}_\omega := \widetilde{q}_\omega \mathds{1}_\mathcal{S} / \int_{\mathcal{S}} \widetilde{q}_\omega$ be the truncation of $\widetilde{q}_\omega$ on $\mathcal{S}$. Then, $\widetilde{q}^{\mathcal{S}}_\omega \in \mathcal{Q}$ and, as $\omega \to 0$,
\begin{equation*}
    \KL(\widetilde{q}^\mathcal{S}_\omega  \Vert  g_\omega) \longrightarrow 0.
\end{equation*}
\end{proposition}

To prove Proposition~\ref{prop:kl_rate}, we derive an intermediate lemma.

\begin{lemma}
\label{lem:kl_rate_untruncated} Let $\KL_{\mathcal{S}}(\widetilde{q}_\omega  \Vert  g_\omega) := \int_{\mathcal{S}} \widetilde{q}_\omega \log \widetilde{q}_\omega/g_\omega$. As $\omega \to 0$,
\begin{equation*}
    \KL_{\mathcal{S}}(\widetilde{q}_\omega  \Vert  g_\omega) \longrightarrow 0.
\end{equation*}
\end{lemma}

\begin{proof}[Proof of Lemma~\ref{lem:kl_rate_untruncated}]
For $j\in[I]$, let $a_j(\xb)
:=
\tfrac12(\xb-\xb_j^\star)^T\Hb_j(\xb-\xb_j^\star)$ and $A:=\sum_{j=1}^I(\det\Hb_j)^{-1/2}$.
On $U_i$, expand $f$ as
\[
f(\xb)=f^\star-a_i(\xb)+\Delta_i(\xb),
\qquad
|\Delta_i(\xb)|
\le C\|\xb-\xb_i^\star\|^3.
\]
By the definition of $c_j$,
\begin{equation}
\label{eq:q_tilde_A}
\widetilde q_\omega(\xb)
=
(2\pi\omega)^{-d/2}A^{-1}
\sum_{j=1}^I e^{-a_j(\xb)/\omega},
\end{equation}
and, from Proposition~\ref{prop:laplace}, we also have 
\begin{equation}
\label{eq:z_omega_A}
    Z_\omega(f) = e^{f^\star/\omega}(2\pi\omega)^{d/2}A\{1+o(1)\}.
\end{equation}
This implies, for $\xb\in U_i$,
\begin{equation}
\label{eq:mixture_logratio}
\log\frac{\widetilde q_\omega(\xb)}{g_\omega(\xb)}
=
-\frac{\Delta_i(\xb)}{\omega}
+
\log\left(
\sum_{j=1}^I
e^{-(a_j(\xb)-a_i(\xb))/\omega}
\right)
+r_\omega,
\end{equation}
where $r_\omega=o(1)$ does not depend on $\xb$.

Write $\xb=\xb_i^\star+\sqrt{\omega}\,\yb$. Since
\[
0\le
\log\left(
\sum_{j=1}^I
e^{-(a_j(\xb)-a_i(\xb))/\omega}
\right)
\le
\log I+\frac{a_i(\xb)}{\omega}
\le C(1+\|\yb\|^2),
\]
and, because $U_i$ is bounded,
\[
\frac{|\Delta_i(\xb)|}{\omega}
\le C\sqrt{\omega}\|\yb\|^3
\le C'\|\yb\|^2,
\]
we obtain, for sufficiently small $\omega$,
\begin{equation}
\label{eq:uniform_logratio}
\left|
\log\frac{\widetilde q_\omega(\xb)}{g_\omega(\xb)}
\right|
\le C(1+\|\yb\|^2),
\qquad
\xb\in U_i.
\end{equation}

Fix $R>0$ and let
$B_i(R,\omega)
=\{\xb:\|\xb-\xb_i^\star\|\le R\sqrt{\omega}\}$.
For $\xb\in B_i(R,\omega)$, separation of the modes implies
$a_j(\xb)-a_i(\xb)\ge b_i>0$ for $j\neq i$ and all sufficiently
small $\omega$. Hence, by~\eqref{eq:mixture_logratio},
\[
\sup_{\xb\in B_i(R,\omega)}
\left|
\log\frac{\widetilde q_\omega(\xb)}{g_\omega(\xb)}
\right|
\longrightarrow0.
\]
Therefore,
\begin{equation}
\label{eq:logratio_bound_1}
\int_{B_i(R,\omega)}
\widetilde q_{i,\omega}(\xb)
\log\frac{\widetilde q_\omega(\xb)}{g_\omega(\xb)}
\,d\xb
\longrightarrow0.
\end{equation}

Next, if $\Yb\sim\mathcal N(0,\Hb_i^{-1})$, then
\eqref{eq:uniform_logratio} gives
\begin{align}
\label{eq:logratio_bound_2}
&\int_{B_i(R,\omega)^c\cap U_i}
\widetilde q_{i,\omega}(\xb)
\left|
\log\frac{\widetilde q_\omega(\xb)}{g_\omega(\xb)}
\right|d\xb \\ \nonumber
&\qquad\le
C\,\mathbb E\left[
(1+\|\Yb\|^2)\mathbf 1_{\{\|\Yb\|>R\}}
\right]
=: \rho_i(R),
\end{align}
where $\rho_i(R)\to0$ as $R\to\infty$.

Finally, since $\xb_i^\star\in U_i$, Gaussian tail bounds give $\widetilde q_{i,\omega}(\mathcal S\setminus U_i)
\le C_1e^{-b_i/\omega}$, $b_i > 0$. Moreover, compactness of $\mathcal S$, boundedness of $f$, and equations~\eqref{eq:q_tilde_A} and~\eqref{eq:z_omega_A} imply
\[
\sup_{\xb\in\mathcal S}
\left|
\log\frac{\widetilde q_\omega(\xb)}{g_\omega(\xb)}
\right|
\le \frac{C_2}{\omega}.
\]
Thus
\begin{equation}
\int_{\mathcal S\setminus U_i}
\widetilde q_{i,\omega}(\xb)
\left|
\log\frac{\widetilde q_\omega(\xb)}{g_\omega(\xb)}
\right|d\xb
\longrightarrow0.
\label{eq:logratio_bound_3}
\end{equation}

Since
$\widetilde q_\omega=\sum_{i=1}^Ic_i\widetilde q_{i,\omega}$,
\eqref{eq:logratio_bound_1}, \eqref{eq:logratio_bound_2} and \eqref{eq:logratio_bound_3} yield, for every fixed $R$,
\[
\limsup_{\omega\to0}
\left|
\KL_{\mathcal S}(\widetilde q_\omega\Vert g_\omega)
\right|
\le
\sum_{i=1}^I c_i\rho_i(R).
\]
Letting $R\to\infty$ ensures
\[
\KL_{\mathcal S}(\widetilde q_\omega\Vert g_\omega)
\longrightarrow0.
\]
\end{proof}

\begin{proof}[Proof of Proposition~\ref{prop:kl_rate}]
Note that bounding the Gaussian tails in $\widetilde{q}_\omega$ gives $\beta_\omega := \int_{\mathbb{R}^d \setminus \mathcal{S}} \widetilde{q}_\omega (\xb) d\xb = O(e^{-c/\omega})$, for some $c > 0$. From Lemma~\ref{lem:kl_rate_untruncated}, $\KL_{\mathcal{S}}(\widetilde{q}_\omega || g_\omega) = o(1)$. Since $\widetilde{q}^\mathcal{S}_\omega = \widetilde{q}_\omega/(1-\beta)$ on $\mathcal{S}$, we have
\begin{align*}
    \KL(\widetilde{q}^\mathcal{S}_\omega || g_\omega) - \KL_{\mathcal{S}}(\widetilde{q}_\omega || g_\omega) &= \left(\frac{1}{1 - \beta_\omega} - 1\right) \KL_{\mathcal{S}}(\widetilde{q}_\omega || g_\omega) - \log (1 - \beta_\omega) \\
    &= O(\beta_\omega) \KL_{\mathcal{S}}(\widetilde{q}_\omega || g_\omega) + O(\beta_\omega) \\
    &= O(e^{-c/\omega}).
\end{align*}
Therefore, we can conclude that $\left\lvert \KL(\widetilde{q}^\mathcal{S}_\omega || g_\omega) - \KL_{\mathcal{S}}(\widetilde{q}_\omega || g_\omega) \right\rvert \rightarrow 0$, which proves the proposition.
\end{proof}

\subsubsection*{Step 3: Conclusion, proof of Theorem~\ref{thm:annealing_main}}

We are now going to prove all the statements of Theorem~\ref{thm:annealing_main}. Recall the conditionals of $g_\omega$ and $q_\omega^\star$ on $U_i$, $i = 0, \dots, I$:
\begin{equation*}
\varphi_{i,\omega}(\xb)
= \frac{q_\omega^\star(\xb)\mathds 1_{U_i}(\xb)}{\alpha_{i,\omega}},
\qquad
\psi_{i,\omega}(\xb)
= \frac{g_\omega(\xb)\mathds 1_{U_i}(\xb)}{\gamma_{i,\omega}},
\end{equation*}
where
\begin{equation*}
\alpha_{i,\omega} = \int_{U_i} q_\omega^\star(\xb) d\xb,
\qquad
\gamma_{i,\omega} = \int_{U_i} g_\omega(\xb) d\xb,
\end{equation*}
and the competitor mixture $\widetilde{q}_{\omega}^\mathcal{S} = \widetilde{q}_\omega \mathds{1}_\mathcal{S} / \int_{\mathcal{S}} \widetilde{q}_\omega(\xb) d\xb$, with 
\begin{equation*}
    \widetilde{q}_\omega(\xb) = \sum_{i=1}^I
 c_i \mathcal{N}(\xb; \xb_i^\star, \omega \Hb_i^{-1}).
\end{equation*}

\begin{proof}[Proof of Theorem~\ref{thm:annealing_main}]
The principle of the proof is as follows: because the competitor $\widetilde{q}_{\omega}^\mathcal{S}$ of
Proposition~\ref{prop:kl_rate} converges in KL to $g_\omega$, the same holds for $q_\omega^\star$, and therefore $q_\omega^\star$ inherits the asymptotics of the Gibbs measure $g_\omega$. Steps~1-2 prove Theorem~\ref{thm:annealing_main}(i), and Step~3 proceeds for Theorem~\ref{thm:annealing_main}(ii).

\noindent\textit{\underline{Convergence of $q_\omega^\star$ in KL.}}
For any density $q$ on $\mathcal{S}$, we have the variational identity,
\begin{equation*}
\mathcal L_\omega(q)
= \omega \log Z_\omega(f) - \omega \KL(q \Vert g_\omega).
\end{equation*}
The competitor mixture $\widetilde q_\omega^\mathcal{S} \in \mathcal{Q}$. Since $q_\omega^\star$ maximizes $\mathcal L_\omega$ over $\mathcal Q$, we have, as $\omega\to0$,
\begin{equation*}
\KL(q_\omega^\star || g_\omega)
\le \KL(\widetilde q_\omega^\mathcal{S} || g_\omega)
\xrightarrow[]{}0,
\end{equation*}
where the limit follows from Proposition~\ref{prop:kl_rate}. Thus, the optimal truncated mixture satisfies
\begin{equation}
\KL(q_\omega^\star \Vert g_\omega)\to0.
\label{eq:mixture_KL_goes_to_zero}
\end{equation}

\noindent\textit{\underline{Asymptotic mass in mode neighborhoods.}}
For $i = 0, \dots, I$, consider $\alpha_{i,\omega}$ and $\gamma_{i,\omega}$ of~\eqref{eq:alpha_gamma}. Let $T:\mathcal S\to\{0,1,\dots,I\}$ send $x\in U_j$ to $j$.
Let $\alphab_\omega$ and $\gammab_\omega$ be the induced discrete laws.
Theorem 4.1 of \citealp{kullback1951information} gives
\begin{equation*}
\KL(\alphab_\omega \Vert \gammab_\omega)
\le
\KL(q_\omega^\star \Vert g_\omega)\xrightarrow[]{}0.
\end{equation*}
From Proposition~\ref{prop:laplace},  
$\gamma_{i,\omega}\to c_i$ for $i \in [I]$, and $\gamma_{0,\omega}\to0$.
Hence~\eqref{eq:mixture_KL_goes_to_zero} implies
\begin{equation*}
\alpha_{i,\omega}\to c_i,\qquad i \in [I],
\qquad
\alpha_{0,\omega}\to0.
\end{equation*}

\noindent\textit{\underline{Local conditional densities.}}
For $i \in [I]$, consider the conditional densities $\varphi_{i,\omega}$ and $\psi_{i,\omega}$ from~\eqref{eq:varphi_psi}. The chain rule for the KL divergence over the partition
$\{U_0,U_1,\dots,U_I\}$ gives
\begin{equation*}
\KL(q_\omega^\star \Vert g_\omega)
=
\KL(\alphab_\omega\Vert \gammab_\omega)
+
\sum_{i=0}^I
\alpha_{i,\omega} \KL(\varphi_{i,\omega}\Vert \psi_{i,\omega}).
\end{equation*}
Using \eqref{eq:mixture_KL_goes_to_zero} and
$\alpha_{i,\omega}\to c_i>0$, we obtain for each $i \in [I]$,
\begin{equation*}
\KL(\varphi_{i,\omega}\Vert \psi_{i,\omega})
 \xrightarrow[]{} 0.
\end{equation*}

Furthermore, since all densities live on the bounded set $U_i$, Pinsker's inequality implies convergence in total variation.  
From Proposition~\ref{prop:gibbs_mean_cov}, the conditional law $\psi_{i,\omega}$ has
mean $\xb_i^\star+o(1)$ and covariance $o(1)$.
Hence, the KL convergence forces the same limits for the mixture’s local conditional moments:
\begin{equation*}
\mathbb E_{\varphi_{i,\omega}}[\Xb]
= \xb_i^\star + o(1),
\qquad
\mathrm{Cov}_{\varphi_{i,\omega}}(\Xb)
= o(1).
\end{equation*}
\end{proof}

\subsection{Proof of Theorem~\ref{thm:consistency_main}}
\label{app:thm_consistency}

In this section, we prove Theorem~\ref{thm:consistency_main}, establishing fixed-temperature asymptotics for the empirical maximizers when the sample size grows to infinity. 

For convenience, we recall the assumptions of Theorem~\ref{thm:consistency_main}.
\begin{assumption} \label{ass:consistency}
Define the following set of assumptions:
    \begin{enumerate}
        \item[(B1)] (\textbf{Data}) $X_1,\dots,X_N \stackrel{i.i.d.}{\sim} p$ supported on a bounded $\mathcal S\subset\mathbb R^d$, with $p$ bounded and continuous on $\mathcal S$.
        \item[(B2)] (\textbf{Unicity of the population maximizer}) $\mathcal{L}_{\omega_0}$ admits a unique maximizer $\Lambdab^\star_{\omega_0} \in \operatorname{int}(\Theta_K)$.
        \item[(B3)] (\textbf{Nondegeneracy of the population maximizer}) At the maximizer $\Lambdab^\star_{\omega_0}$, $\Hb^\star_{\omega_0} := \Hb_{\omega_0}(\Lambdab^\star_{\omega_0})$ is positive definite.
    \end{enumerate}
\end{assumption}

Each of the following subsections proves one of the three statements in    Theorem~\ref{thm:consistency_main}:
    \begin{enumerate}
        \item[(i)] Under (B1), the population objective $\mathcal{L}_{\omega_0}$ attains its maximum on $\Theta_K$.
        \item[(ii)] Under (B1)-(B2), any empirical maximizer $\widehat{\Lambdab}_{\omega_0,N}\in\argmax \widehat{\mathcal{L}}_{\omega_0,N}$ satisfies
\begin{equation*}
\widehat{\Lambdab}_{\omega_0,N} \xrightarrow[]{P} \Lambdab^\star_{\omega_0}.
\end{equation*}
        \item[(iii)] Under (B1)-(B2)-(B3), any empirical maximizer $\widehat{\Lambdab}_{\omega_0,N}\in\argmax \widehat{\mathcal{L}}_{\omega_0,N}$ satisfies \begin{equation*}
\sqrt{N} \big(\widehat\Lambdab_{\omega_0,N}-\Lambdab^\star_{\omega_0}\big)
 \xrightarrow[]{\mathcal{D}} 
\mathcal{N} (0,\Wb_{\omega_0}),
\end{equation*}
where $\Wb_{\omega_0} := \big(\Hb_{\omega_0}^\star\big)^{-1} \Vb_{\omega_0} \big(\Hb_{\omega_0}^\star\big)^{-1}$, and $\Vb_{\omega_0}:=\operatorname{Var} \big(\nabla_{\Lambdab} q_{\Lambdab}(\Xb)\big)\big|_{\Lambdab=\Lambdab^\star_{\omega_0}}$.
    \end{enumerate}

\subsubsection{Existence of a population maximizer}

\begin{proof}[Proof of (i)]
On the compact $\Theta_K$, $(\Lambdab, \xb) \mapsto q_\Lambdab(\xb)$ is continuous and is uniformly bounded on the bounded set $\mathcal{S}$. For each $x\in\mathcal S$, $\Lambdab\mapsto q_\Lambdab(\xb)$ is continuous. By dominated convergence (bounded envelope on $\mathcal S$),
$\Lambdab\mapsto \mathbb E[q_\Lambdab(\Xb)]$ is continuous. The entropy term
$\Lambdab\mapsto \mathcal H_{\mathcal S}(q_\Lambdab)$ is continuous on $\Theta_K$ as well.
Thus, by Weierstrass, $\mathcal L_{\omega_\ast}$ is continuous on compact $\Theta_K$ and attains its maximum.
\end{proof}

\subsubsection{Consistency of the empirical maximizers}

The proof of this result consists of two steps: the derivation of a uniform law of large numbers, due to the Glivenko--Cantelli (GC) property of the Gaussian mixture family in $\Theta_K$ (see App.~\ref{app:gc_property}, Lem.~\ref{lem:gc-donsker-mixture}), and the application of the Argmax Theorem.

\begin{proof}[Proof of (ii)]
The class $\{q_\Lambdab:\Lambdab\in\Theta_K\}$ is $P$-Glivenko--Cantelli with an integrable envelope on $\mathcal S$, so as $N \to \infty$:
\begin{equation*}
\sup_{\Lambdab\in\Theta_K} \left\lvert \frac{1}{N} \sum_{i=1}^N q_\Lambdab(\Xb_i) - \mathbb{ E}_p[q_\Lambdab(\Xb)]\right\rvert \xrightarrow[]{P} 0.
\end{equation*}
Adding the deterministic term $\omega_\ast \mathcal{H}_\mathcal{S}(q_\Lambdab)$ yields 
\begin{equation*}
\sup_{\Lambdab\in\Theta_K}\big|\widehat{\mathcal L}_{\omega_\ast,N}(\Lambdab)-\mathcal L_{\omega_\ast}(\Lambdab)\big|\xrightarrow[]{P} 0.
\end{equation*}
Since $\mathcal L_{\omega_\ast}$ has a unique maximizer $\Lambdab^\star_{\omega_\ast}$ (Assumption (B2)), the Argmax Theorem \citep[Thm.~5.7]{van1998asymptotic} yields
$\widehat\Lambdab_{\omega_\ast,N}\xrightarrow{P}\Lambdab^\star_{\omega_\ast}$, as $N \to \infty$.
\end{proof}

\subsubsection{Asymptotic normality of the empirical maximizers}

The asymptotic normality result is derived from a Taylor expansion and the Donsker property on the Gaussian mixture family but also on the gradients. Lemma~\ref{lem:gc-donsker-mixture} gives these properties.

\begin{proof}[Proof of (iii)]
By Lemma~\ref{lem:gc-donsker-mixture}, $\mathcal{Q}$ is GC, so
$\sup_{\Lambdab\in\Theta_K}\big|\widehat{\mathcal{L}}_{\omega_\ast,N}(\Lambdab)-\mathcal{L}_{\omega_\ast}(\Lambdab)\big|\xrightarrow{P}0$,
hence $\widehat\Lambdab_{\omega_\ast,N}\xrightarrow{P}\Lambdab^\star_{\omega_\ast}$ by the standard M-estimator consistency theorem \citep[Thm.~5.7]{van1998asymptotic}.
Write
\begin{align*}
&\nabla_{\Lambdab} \widehat{\mathcal{L}}_{\omega_\ast,N}(\Lambdab)
=\frac{1}{N}\sum_{i=1}^N \nabla_{\Lambdab} q_{\Lambdab}(X_i)
+ \omega_\ast \nabla_{\Lambdab} \mathcal{H}_{\mathcal{S}}(q_{\Lambdab}),
\\
&\nabla_{\Lambdab} \mathcal{L}_{\omega_\ast}(\Lambdab)
=\mathbb{E} [\nabla_{\Lambdab} q_{\Lambdab}(X)]+\omega_\ast \nabla_{\Lambdab} \mathcal{H}_{\mathcal{S}}(q_{\Lambdab}).
\end{align*}
Thus 
$\nabla_{\Lambdab} \widehat{\mathcal{L}}_{\omega_\ast,N}(\Lambdab^\star_{\omega_\ast})-\nabla_{\Lambdab} \mathcal{L}_{\omega_\ast}(\Lambdab^\star_{\omega_\ast})
=(P_N-P)\nabla_{\Lambdab} q_{\Lambdab}(X)\big|_{\Lambdab=\Lambdab^\star_{\omega_\ast}}$
has mean zero and variance $\Vb_{\omega_\ast}/N$.

Furthermore, let $\Delta_N := \widehat\Lambdab_{\omega_\ast,N}
-\Lambdab^\star_{\omega_\ast}$.
By the integral form of Taylor's theorem,
\[
0
=
\nabla_{\Lambdab}
\widehat{\mathcal L}_{\omega_\ast,N}
(\Lambdab^\star_{\omega_\ast})
+
\left\{
\int_0^1
\nabla_{\Lambdab}^2
\widehat{\mathcal L}_{\omega_\ast,N}
(\Lambdab^\star_{\omega_\ast}+t\Delta_N)\,dt
\right\}\Delta_N.
\]
Since $\nabla_{\Lambdab}^2
\mathcal L_{\omega_\ast}
(\Lambdab^\star_{\omega_\ast})
=
-\Hb_{\omega_\ast}^\star$, 
we may write
\[
0
=
\nabla_{\Lambdab}
\widehat{\mathcal L}_{\omega_\ast,N}
(\Lambdab^\star_{\omega_\ast})
-
\Hb_{\omega_\ast}^\star\Delta_N
+
\Rb_N\Delta_N,
\]
where
\[
\Rb_N
:=
\int_0^1
\left[
\nabla_{\Lambdab}^2
\widehat{\mathcal L}_{\omega_\ast,N}
(\Lambdab^\star_{\omega_\ast}+t\Delta_N)
+
\Hb_{\omega_\ast}^\star
\right]dt.
\]
In addition, by consistency, continuity of the population Hessian, and the uniform
convergence of the empirical Hessian, we have $\|\Rb_N\|=o_P(1)$.

Therefore, by the CLT for the $P$-Donsker gradient class
(Lemma~\ref{lem:gc-donsker-mixture}), $\sqrt{N}
\nabla_{\Lambdab}
\widehat{\mathcal L}_{\omega_\ast,N}
(\Lambdab^\star_{\omega_\ast})
\xrightarrow{\mathcal D}
\mathcal N(0,\Vb_{\omega_\ast})$. Since $\Hb_{\omega_\ast}^\star$ is positive definite and
$\|\Rb_N\|=o_P(1)$,
\[
\sqrt N\,\Delta_N
=
(\Hb_{\omega_\ast}^\star-\Rb_N)^{-1}
\sqrt N\,
\nabla_{\Lambdab}
\widehat{\mathcal L}_{\omega_\ast,N}
(\Lambdab^\star_{\omega_\ast}),
\]
and therefore
\[
\sqrt N\,\Delta_N
=
(\Hb_{\omega_\ast}^\star)^{-1}
\sqrt N\,
\nabla_{\Lambdab}
\widehat{\mathcal L}_{\omega_\ast,N}
(\Lambdab^\star_{\omega_\ast})
+o_P(1).
\]
\end{proof}

\subsection{Proofs for Sections~\ref{sec:uq} and~\ref{app:uq}}
\label{app:proofs_uq}

\subsubsection{Preliminaries: Bounded-Lipschitz metric}

Let $(\mathcal S,d)$ be a metric space. For $f:\mathcal S\to\mathbb R$ define
\begin{equation*}
\|f\|_{\infty}=\sup_{\xb\in\mathcal S}|f(\xb)|,\qquad
\mathrm{Lip}(f)=\sup_{\xb\neq \yb}\frac{|f(\xb)-f(\yb)|}{d(\xb,\yb)}.
\end{equation*}
For two measures $\nu, \nu'$ on $\mathcal{S}$, the bounded-Lipschitz (BL) metric is
\begin{equation*}
d_{\mathrm{BL}}(\nu,\nu') := \sup_{\|f\|_{\mathrm{BL}}\le 1} \Big|\textstyle\int f d\nu - \int f d\nu'\Big|,
\qquad 
\|f\|_{\mathrm{BL}}:=\|f\|_\infty+\mathrm{Lip}(f).
\end{equation*}
Since $\mathcal S$ is a compact metric space, $d_{\mathrm{BL}}$
metrizes weak convergence of probability measures on $\mathcal S$ \citep{dudley2002real}, denoted by $\xrightarrow[]{w}$. We use it to state bootstrap results compactly.

\subsubsection{Technical lemmas}

We state these lemmas, with proofs in Section~\ref{app:technical_lemmas_proofs}.

\begin{lemma}[Regularity for Gaussian mixtures]\label{lem:mixture-regularity}
For a bounded space $\mathcal{S}$ and on the compact $\Theta_K$,
there exists a neighborhood $\mathcal V\subset \Theta_K$ of $\Lambdab^\star_{\omega_0}$ such that
\begin{enumerate}
    \item[(i)] $\Lambdab\mapsto q_\Lambdab(\xb)$ is $\mathcal{C}^2$ on $\mathcal V$ for every $\xb\in\mathcal S$;
    \item[(ii)] $\sup_{\Lambdab\in\mathcal V}\|\nabla_\Lambdab q_\Lambdab(\xb)\|\le F_1(\xb)$ for a bounded-Lipschitz function $F_1$ on $\mathcal S$;
    \item[(iii)] $\sup_{\Lambdab\in\mathcal V}\|\nabla^2_\Lambdab q_\Lambdab(\xb)\|\le F_2(\xb)$ for a bounded-Lipschitz function $F_2$ on $\mathcal S$.
\end{enumerate}
\end{lemma}

Let $\mathcal{F}$ be the function class of Lemma~\ref{lem:boot-clt} and $\mathcal M_0(\mathcal S)$ the set of finite signed measures on $\mathcal S$ with total mass zero. Let
\begin{equation*}
    \ell_{\mathrm{lin}}^\infty(\mathcal F)
    :=
    \overline{
    \left\{
    h_\nu\in\ell^\infty(\mathcal F):
    h_\nu(f)=\int f\,d\nu,\ 
    \nu\in\mathcal M_0(\mathcal S)
    \right\}
    }^{\ell^\infty(\mathcal F)}.
\end{equation*}

\begin{lemma}[Empirical-process CLT and bootstrap]\label{lem:boot-clt}
Let $\mathcal Q^{(u)}_\nabla := \bigl\{ \xb\mapsto [\nabla_\Lambdab q_\Lambdab(\xb)]_u: \Lambdab\in\Theta_K \bigr\}$ and $\mathcal Q^{(u,v)}_{\nabla^2} := \bigl\{ \xb\mapsto [\nabla^2_\Lambdab q_\Lambdab(\xb)]_{uv}: \Lambdab\in\Theta_K \bigr\}$ for all $u, v = 1, \dots, \dim(\Theta_K)$, and define the function class
\begin{equation*}
    \mathcal{F} = \bigcup_{u=1}^{\dim(\Theta_K)} Q^{(u)}_\nabla \cup \bigcup_{u,v=1}^{\dim(\Theta_K)} Q^{(u,v)}_{\nabla^2}.
\end{equation*}
We have, as $N \to \infty$:
\begin{equation*}
\sqrt{N} (P_N-P) \xrightarrow[]{w} \mathbb{G}_P\quad\text{in }\ell^\infty(\mathcal F), 
\qquad 
\sqrt{N} (P_N^\ast-P_N) \xrightarrow[]{w} \mathbb{G}_P\quad\text{in }\ell^\infty(\mathcal F),
\end{equation*}
conditionally on $\Xb_{1:N}$, in $P$-probability, where $\mathbb{G}_P$ denote the $P$-Brownian bridge indexed by $\mathcal F$ and
$\mathbb G_P\in\ell_{\mathrm{lin}}^\infty(\mathcal F)$ almost surely.
\end{lemma}

\begin{lemma}[Hadamard differentiability of the argmax map]\label{lem:argmax-hadamard}
For $h \in \ell_{\mathrm{lin}}^\infty(\mathcal F)$, let us write $h(\nabla_\Lambdab q_\Lambdab) := \left(h([\nabla_\Lambdab q_\Lambdab]_u)\right)_{1\le u\le \dim(\Theta_K)}$ and $h(\nabla^2_\Lambdab q_\Lambdab) := \left(h([\nabla^2_\Lambdab q_\Lambdab]_{u,v})\right)_{1\le u,v \le \dim(\Theta_K)}$. Define $S$ for finite signed measures $G$, with $G(f):=\int f\,dG$, by
\begin{equation*}
S(G,\Lambdab)= G(\nabla_{\Lambdab}q_{\Lambdab}) +\omega_0 \nabla_{\Lambdab}\mathcal H_{\mathcal S}(q_{\Lambdab}).
\end{equation*}
Under Assumption~\ref{ass:consistency}, we have the following
\begin{itemize}
    \item[(i)] There exists a neighborhood $\mathcal{U}$ of $P$ under $\|\cdot\|_{\mathcal F}$ such that for every finite signed measure $G \in \mathcal{U}$, the equation over $\Lambdab$ 
    \begin{equation*}
        S(G, \Lambdab) = 0
    \end{equation*}
    admits a locally unique solution $\Lambdab(G)$ in a neighborhood of $\Lambdab^\star_{\omega_0}$.
    \item[(ii)] The map $\Phi : G \mapsto \Lambdab(G)$ is Hadamard differentiable at $G=P$, tangentially to 
$\ell_{\mathrm{lin}}^\infty(\mathcal{F})$, with
\begin{equation*}
D\Phi_P(h)=(\Hb_{\omega_0}^\ast)^{-1} h(\nabla_{\Lambdab}q_{\Lambdab^\ast_{\omega_0}}).
\end{equation*}
\end{itemize}
\end{lemma}

\begin{lemma}[Measurability of the matching functional]\label{lem:measurability}
$\Lambdab\mapsto \mathcal{M}(\Lambdab)$ is Borel-measurable on $\Theta_K$. In fact, $\mathcal{M}$ is continuous on a finite Borel partition of $\Theta_K$.
\end{lemma}

\begin{lemma}[Continuity of the stability indicator $F_j$]\label{lem:indicator-continuity}
Fix $j\in [K_0]$ and let $U_j=B(\ub_j,r_j)$ be an open ball.
Assume (by shrinking $r_j$ if needed) that $\Eb_j \mathcal{M}(\Lambdab^\star_{\omega_0})$ lies at positive distance from $\partial U_j$.
Then, under Assumption~\ref{ass:consistency} and~\ref{as:sep-match_main}, for any deterministic measurable selector $\Phi$ from the argmax,
$F_j(G)=\mathds{1}\{\Eb_j \mathcal{M}(\Phi(G))\in U_j\}$ is continuous at $G=P$ with respect to weak convergence: for probability measures $G_n$, if $G_n\xrightarrow{w}P$, then $F_j(G_n)\to F_j(P)$.
\end{lemma}

\subsubsection{Proofs of lemmas}
\label{app:technical_lemmas_proofs}

\begin{proof}[Sketch of proof of Lemma~\ref{lem:mixture-regularity}]
  Gaussian mixture densities are $\mathcal{C}^\infty$ in the parameters $\Lambdab$ as long as the covariances remain positive definite and in a compact subset, which is guaranteed in $\Theta_K$. This gives $(i)$.

  The first and second order derivatives of $q_\Lambdab(\xb)$ w.r.t. $\Lambdab$ are products of polynomials of $\xb$ and Gaussian densities, which are smooth in $\xb$, hence Lipschitz on the bounded set $\mathcal{S}$. Since $\mathcal{S}$ and $\Theta_K$ are compact, they are also uniformly bounded. Typically, there exists $C_1, C_2 > 0$ such that
  \begin{equation*}
      F_1(\xb) = C_1(1 + \lVert \xb \rVert + \lVert \xb \rVert^2), \qquad F_2(\xb) = C_2(1 + \lVert \xb \rVert + \lVert \xb \rVert^2 + \lVert \xb \rVert^3 + \lVert \xb \rVert^4),
  \end{equation*}
  which are bounded-Lipschitz on $\mathcal{S}$, and $(ii)$ and $(iii)$ hold.
\end{proof}

\begin{proof}[Proof of Lemma~\ref{lem:boot-clt}]
By Lemma~\ref{lem:gc-donsker-mixture}, each class 
$\mathcal{Q}^{(u)}_{\nabla}$ and $\mathcal{Q}^{(u,v)}_{\nabla^2}$
is $P$-Donsker and admits a measurable envelope in $L_2(P)$. Since $\mathcal{F}$ is a finite union of $P$-Donsker classes, it is also $P$-Donsker (\citealp[Example 2.10.7]{van1996weak}). Hence, it is also Glivenko--Cantelli and
\begin{equation*}
\sqrt{N} (P_N-P) \xrightarrow[]{w} \mathbb{G}_P\quad\text{in }\ell^\infty(\mathcal F),
\end{equation*}
where $\mathbb{G}_P$ is a $P$-Brownian bridge indexed by $\mathcal{F}$. 
Furthermore, the nonparametric bootstrap empirical process satisfies the conditional functional CLT \citep[Thm.~3.6.3]{van1996weak}:
\begin{equation*}
\sqrt{N} (P_N^\ast-P_N) \xrightarrow[]{w} \mathbb{G}_P
\quad\text{in } \ell^\infty(\mathcal{F}),
\end{equation*}
conditionally on $\Xb_{1:N}$, in $P$-probability. Moreover, for every $N$, $\sqrt{N}(P_N-P)$ is induced by a finite signed measure with total mass zero, and thus belongs to $\ell_{\mathrm{lin}}^\infty(\mathcal F)$. Since $\ell_{\mathrm{lin}}^\infty(\mathcal F)$ is closed in $\ell^\infty(\mathcal F)$, the Portmanteau theorem \citep[Lem.~2.2]{van1998asymptotic} implies that $\mathbb G_P\in\ell_{\mathrm{lin}}^\infty(\mathcal F)$ almost surely.
\end{proof}

\begin{proof}[Proof of Lemma~\ref{lem:argmax-hadamard}]
\textit{(i)} Let $\mathcal V \subset \Theta_K$ be the set in which Lemma~\ref{lem:mixture-regularity} holds. By Lemma~\ref{lem:mixture-regularity}, for all $\xb \in \mathcal{S}$, the map $\Lambdab \mapsto \nabla_\Lambdab q_\Lambdab(\xb)$ is $\mathcal{C}^1$ on $\mathcal{V}$, $\sup_{\Lambdab\in\mathcal V}\|\nabla_\Lambdab q_\Lambdab(\xb)\|\le F_1(\xb)$ and $\sup_{\Lambdab\in\mathcal V}\|\nabla^2_\Lambdab q_\Lambdab(\xb)\|\le F_2(\xb)$ for some bounded functions $F_1$ and $F_2$. 
Therefore, for any finite signed measure $G$, the maps
$\Lambdab \mapsto G([\nabla_\Lambdab q_\Lambdab]_u)$,
for $u=1,\dots,\dim(\Theta_K)$, are $\mathcal C^1$ on $\mathcal V$,
with derivatives
$(G([\nabla^2_\Lambdab q_\Lambdab]_{u1}),\dots,
G([\nabla^2_\Lambdab q_\Lambdab]_{u\dim(\Theta_K)}))$,
by differentiation under the integral. This implies that $S(G,\Lambdab)$ is $\mathcal C^1$ in $\Lambdab$, and 
\begin{equation*}
S(P,\Lambdab^\star_{\omega_0})=0,\qquad 
\partial_{\Lambdab} S(P,\Lambdab^\star_{\omega_0})
=\nabla_{\Lambdab}^2\mathcal L_{\omega_0}(\Lambdab^\star_{\omega_0})
= \Ab \prec 0,
\end{equation*}
where $\Ab := - \Hb_{\omega_0}^\star$ is invertible by Assumption~\ref{ass:consistency}. 

Furthermore, since $\partial_\Lambdab S(P,\Lambdab)$ is continuous in $\Lambdab$, there is $r > 0$ such that on the closed ball $\overline{B}(\Lambdab^\star_{\omega_0}, r) \subset \mathcal{V}$,
\begin{equation}
    \sup_{\Lambdab \in \overline{B}(\Lambdab^\star_{\omega_0}, r)} \lVert \partial_\Lambdab S(P,\Lambdab) - \Ab \rVert \le \frac{1}{3 \lVert \Ab^{-1} \rVert}.
    \label{eq:bound_grad_score}
\end{equation}

Now, define for any finite signed measure $G$, $\Lambdab \in \overline{B}(\Lambdab^\star_{\omega_0}, r)$,
\begin{equation*}
    T_G(\Lambdab) := \Lambdab - \Ab^{-1} S(G, \Lambdab).
\end{equation*}
To prove (i), we would like to apply the Banach fixed-point theorem (\citealp[Th. 1.A]{zeidler1995applied108}) to $T_G$, for all $G$ in some neighborhood of $P$, in the set $\overline{B}(\Lambdab^\star_{\omega_0}, r)$. For this, we need to verify a self-mapping property and a contraction property, which hold for a sufficiently small neighborhood of $P$: 
we consider the neighborhood $\mathcal{T}$ formed by finite signed measures $G$ such that
\begin{equation*}
    \|G-P\|_{\mathcal{F}} \le \min \left(\frac{1}{3\dim(\Theta_K)\|\Ab^{-1}\|}, \frac{r}{3\sqrt{\dim(\Theta_K)}\|\Ab^{-1}\|} \right).
\end{equation*}

\noindent\textit{\underline{Contraction property.}} $S(G, \Lambdab)$ is $\mathcal{C}^1$ on $\mathcal{V}$, so for any $\Lambdab_1$ and $\Lambdab_2$ in $\overline{B}(\Lambdab^\star_{\omega_0}, r)$, the integral form of the mean-value theorem gives
\begin{equation*}
    S(G, \Lambdab_1) - S(G, \Lambdab_2)
    =
    \left[
    \int_0^1
    \partial_\Lambdab S\bigl(G,\Lambdab_2+t(\Lambdab_1-\Lambdab_2)\bigr)
    dt
    \right]
    (\Lambdab_1-\Lambdab_2).
\end{equation*}
Therefore,
\begin{equation*}
    T_G(\Lambdab_1) - T_G(\Lambdab_2)
    =
    \left[
    \Ib-\Ab^{-1}
    \int_0^1
    \partial_\Lambdab S\bigl(G,\Lambdab_2+t(\Lambdab_1-\Lambdab_2)\bigr)
    dt
    \right]
    (\Lambdab_1-\Lambdab_2).
\end{equation*}
This means
\begin{equation}
    \lVert T_G(\Lambdab_1) - T_G(\Lambdab_2)\rVert \le \sup_{\Lambdab \in \overline{B}(\Lambdab^\star_{\omega_0}, r)} \lVert \Ib - \Ab^{-1} \partial_\Lambdab S(G, \Lambdab) \rVert \times \lVert \Lambdab_1 - \Lambdab_2 \rVert.
    \label{eq:contraction_pre}
\end{equation}
Furthermore,
\begin{equation*}
    \lVert \Ib - \Ab^{-1} \partial_\Lambdab S(G, \Lambdab) \rVert \le \lVert \Ab^{-1} \rVert (\lVert \Ab - \partial_\Lambdab S(P, \Lambdab) \rVert + \lVert \partial_\Lambdab S(G, \Lambdab) - \partial_\Lambdab S(P, \Lambdab) \rVert).
\end{equation*}
Then, by definition of $\lVert \cdot \rVert_{\mathcal{F}}$, for any $\Lambdab\in \overline B(\Lambdab^\star_{\omega_0},r)$,
\[
\|\partial_\Lambdab S(G,\Lambdab)-\partial_\Lambdab S(P,\Lambdab)\|
=
\Big\|(G-P)(\nabla_\Lambdab^2 q_\Lambdab)\Big\|
\le \dim(\Theta_K)\|G-P\|_{\mathcal{F}}.
\]
For all $G \in \mathcal{T}$, we have $\lVert G - P \rVert_{\mathcal{F}} \le 1/(3\dim(\Theta_K)\lVert \Ab^{-1} \rVert)$, so
\begin{equation*}
    \lVert \Ab^{-1} \rVert \lVert \partial_\Lambdab S(G, \Lambdab) - \partial_\Lambdab S(P, \Lambdab) \rVert \le \frac{1}{3}.
\end{equation*}
Combining this with~\eqref{eq:bound_grad_score}, we have
\begin{equation*}
    \sup_{\Lambdab \in \overline{B}(\Lambdab^\star_{\omega_0}, r)} \lVert \Ib - \Ab^{-1} \partial_\Lambdab S(G, \Lambdab) \rVert \le \frac{2}{3},
\end{equation*}
so in~\eqref{eq:contraction_pre}, we have
\begin{equation}
    \lVert T_G(\Lambdab_1) - T_G(\Lambdab_2)\rVert \le \frac{2}{3} \lVert \Lambdab_1 - \Lambdab_2 \rVert,
    \label{eq:contraction_tg}
\end{equation}
which proves contraction for $G \in \mathcal{T}$.

\noindent\textit{\underline{Self-mapping property.}} Since $S(P,\Lambdab^\star_{\omega_0})=0$,
\begin{equation*}
    T_G(\Lambdab^\star_{\omega_0}) - \Lambdab^\star_{\omega_0} = -\Ab^{-1} S(G, \Lambdab^\star_{\omega_0}) = -\Ab^{-1} (G-P)(\nabla_\Lambdab q_{\Lambdab^\star_{\omega_0}}).
\end{equation*}
By definition of $\lVert \cdot \rVert_{\mathcal{F}}$,
\[
\Big\|(G-P)(\nabla_\Lambdab q_{\Lambdab^\star_{\omega_0}})\Big\|
\le \sqrt{\dim(\Theta_K)}\|G-P\|_{\mathcal{F}}.
\]
Hence, for $G \in \mathcal{T}$,
\[
\|T_G(\Lambdab^\star_{\omega_0})-\Lambdab^\star_{\omega_0}\|
\le \sqrt{\dim(\Theta_K)}\|\Ab^{-1}\| \|G-P\|_{\mathcal{F}} \le \frac{r}{3}.
\]
Thus, using the contraction property~\eqref{eq:contraction_tg}, for any $\Lambdab\in \overline B(\Lambdab^\star_{\omega_0},r)$,
\[
\|T_G(\Lambdab)-\Lambdab^\star_{\omega_0}\|
\le
\|T_G(\Lambdab)-T_G(\Lambdab^\star_{\omega_0})\| + \|T_G(\Lambdab^\star_{\omega_0})-\Lambdab^\star_{\omega_0}\| \le \frac{2r}{3} + \frac{r}{3} \le r,
\]
and this ensures that $T_G(\overline B(\Lambdab^\star_{\omega_0}, r))\subset \overline B(\Lambdab^\star_{\omega_0}, r)$, the self-mapping property for $G \in \mathcal{T}$.

\noindent\textit{\underline{Conclusion.}} Finally, by the Banach fixed-point theorem for $T_G$, for every $G$ in the neighborhood $\mathcal{T}$ of $P$, there exists a unique fixed point $\Lambdab(G) \in \overline{B}(\Lambdab^\star_{\omega_0}, r) \subset \mathcal{V}$ such that $\Lambdab(G) = T_G(\Lambdab(G))$ is equivalent to $S(G, \Lambdab(G)) = 0$. This gives (i).

\medskip
\textit{(ii)} Let $t_n \to 0$ and $h_n \to h$ in $\ell^\infty(\mathcal F)$ under
$\lVert\cdot\rVert_{\mathcal F}$, for
$h\in\ell_{\mathrm{lin}}^\infty(\mathcal F)$, such that
$G_n:=P+t_nh_n$ is induced by a finite signed measure for every $n$.
Set $\Lambdab_n:=\Phi(G_n)=\Lambdab(G_n)$. Our objective is to show that 
\begin{equation*}
    D\Phi_P(h):= \lim_{n \to \infty} \frac{\Phi(G_n) - \Phi(P)}{t_n}=(\Hb_{\omega_0}^\ast)^{-1} h(\nabla_{\Lambdab}q_{\Lambdab^\ast_{\omega_0}}).
\end{equation*}

First, the sequence $(h_n)$ is bounded in $\lVert \cdot \rVert_{\mathcal{F}}$, so $\lVert G_n - P \rVert_{\mathcal{F}} = \lvert t_n \rvert \lVert h_n \rVert_{\mathcal{F}} \to 0$, and $G_n \in \mathcal{T}$ and $\Lambdab_n \in \mathcal{V}$, so by (i), we have $0 = S(G_n, \Lambdab_n)$. Furthermore, since we also have $S(P, \Lambdab^\star_{\omega_0}) = 0$, then
\begin{equation*}
    0 = \left(S(G_n, \Lambdab_n) - S(G_n, \Lambdab^\star_{\omega_0})\right) + \left(S(G_n, \Lambdab^\star_{\omega_0}) - S(P, \Lambdab^\star_{\omega_0})\right).
\end{equation*}
By linearity, the second difference is
\begin{equation*}
    S(G_n, \Lambdab^\star_{\omega_0}) - S(P, \Lambdab^\star_{\omega_0}) = t_n h_n(\nabla_\Lambdab q_{ \Lambdab^\star_{\omega_0}}).
\end{equation*}
For the first difference, since $S(G, \Lambdab)$ is $\mathcal{C}^1$ in $\Lambdab$, the integral form of the mean-value theorem gives
\begin{equation*}
    S(G_n, \Lambdab_n) - S(G_n, \Lambdab^\star_{\omega_0})
    =
    \left[
    \int_0^1
    \partial_\Lambdab S\bigl(G_n,
    \Lambdab^\star_{\omega_0}
    +t(\Lambdab_n-\Lambdab^\star_{\omega_0})\bigr)
    dt
    \right]
    (\Lambdab_n-\Lambdab^\star_{\omega_0}).
\end{equation*}
Thus, we have the relation
\begin{equation}
    0 =
    \left[
    \int_0^1
    \partial_\Lambdab S\bigl(G_n,
    \Lambdab^\star_{\omega_0}
    +t(\Lambdab_n-\Lambdab^\star_{\omega_0})\bigr)
    dt
    \right]
    (\Lambdab_n-\Lambdab^\star_{\omega_0})
    +t_n h_n(\nabla_\Lambdab q_{\Lambdab^\star_{\omega_0}}).
    \label{eq:relation_for_derivative}
\end{equation}

On the other hand, the fixed-point and the contraction properties of $T_G$ for all $G \in \mathcal{T}$, proven in (i), imply that $\Lambdab_n \rightarrow \Lambdab^\star_{\omega_0}$. Indeed, we have
\begin{equation*}
\begin{split}
    \lVert \Lambdab_n - \Lambdab^\star_{\omega_0} \rVert &= \rVert T_{G_n}(\Lambdab_n) - T_P(\Lambdab^\star_{\omega_0}) \rVert \\
    &\le \rVert T_{G_n}(\Lambdab_n) - T_{G_n}(\Lambdab^\star_{\omega_0}) \rVert + \rVert T_{G_n}(\Lambdab^\star_{\omega_0}) - T_P(\Lambdab^\star_{\omega_0}) \rVert \\
    &\le \frac{2}{3} \lVert \Lambdab_n - \Lambdab^\star_{\omega_0} \rVert + \rVert T_{G_n}(\Lambdab^\star_{\omega_0}) - T_P(\Lambdab^\star_{\omega_0}) \rVert,
\end{split}
\end{equation*}
hence
\begin{equation*}
\begin{split}
    \lVert \Lambdab_n - \Lambdab^\star_{\omega_0} \rVert &\le 3 \rVert T_{G_n}(\Lambdab^\star_{\omega_0}) - T_P(\Lambdab^\star_{\omega_0}) \rVert \\
    &\le 3 \lVert \Ab^{-1} \rVert \left\lVert (G_n-P)(\nabla_\Lambdab q_{\Lambdab^\star_{\omega_0}}) \right\rVert \\
    &\le 3 \sqrt{\dim(\Theta_K)} \lVert \Ab^{-1} \rVert
\lVert G_n - P \rVert_{\mathcal{F}} \\
    &\rightarrow 0.
\end{split}
\end{equation*}
Define
\begin{equation*}
    \Ab_n
    :=
    \int_0^1
    \partial_\Lambdab S\bigl(G_n,
    \Lambdab^\star_{\omega_0}
    +t(\Lambdab_n-\Lambdab^\star_{\omega_0})\bigr)
    dt.
\end{equation*}
Since $\Lambdab_n\to\Lambdab^\star_{\omega_0}$,
\begin{equation*}
\begin{split}
    \|\Ab_n-\Ab\|
    &\le
    \sup_{\Lambdab\in\mathcal V}
    \|\partial_\Lambdab S(G_n,\Lambdab)
    -\partial_\Lambdab S(P,\Lambdab)\|\\
    &\quad+
    \sup_{t\in[0,1]}
    \left\|
    \partial_\Lambdab S\bigl(P,
    \Lambdab^\star_{\omega_0}
    +t(\Lambdab_n-\Lambdab^\star_{\omega_0})\bigr)
    -\Ab
    \right\|\\
    &\le
    \dim(\Theta_K)\|G_n-P\|_{\mathcal F}+o(1)
    \to0.
\end{split}
\end{equation*}
Thus,
\begin{equation}
    \Ab_n\to\Ab=-\Hb^\star_{\omega_0}.
    \label{eq:limit_d_score}
\end{equation}
Hence $\Ab_n$ is invertible for all sufficiently large $n$, and
\eqref{eq:relation_for_derivative} yields
\begin{equation}
    \frac{\Lambdab_n-\Lambdab^\star_{\omega_0}}{t_n}
    =
    -\Ab_n^{-1}
    h_n(\nabla_\Lambdab q_{\Lambdab^\star_{\omega_0}}).
    \label{eq:relation_for_derivative_inverse}
\end{equation}

Finally, as $n \to \infty$, $h_n \to h$ in $\lVert \cdot \rVert_{\mathcal{F}}$, so
\begin{equation}
    \left\|
    h_n(\nabla_\Lambdab q_{\Lambdab^\star_{\omega_0}})
    -
    h(\nabla_\Lambdab q_{\Lambdab^\star_{\omega_0}})
    \right\|
    \le
    \sqrt{\dim(\Theta_K)}
    \|h_n-h\|_{\mathcal F}
    \to0.
\end{equation}
Combining with~\eqref{eq:limit_d_score} and~\eqref{eq:relation_for_derivative_inverse}, we obtain
\begin{equation*}
    \frac{\Phi(G_n) - \Phi(P)}{t_n}
    = \frac{\Lambdab_n -\Lambdab^\star_{\omega_0}}{t_n}
    \to D\Phi_P(h)
    =(\Hb_{\omega_0}^\ast)^{-1}
    h(\nabla_{\Lambdab}q_{\Lambdab^\ast_{\omega_0}}),
\end{equation*}
which proves (ii).
\end{proof}

\begin{proof}[Proof of Lemma~\ref{lem:measurability}]
Let $\Pi$ be the finite set of injections $\pi:\{1,\dots,K_0\}\hookrightarrow\{1,\dots,K\}$ and define
\begin{equation*}
\mathsf{cost}_\pi(\Lambdab):=\sum_{j=1}^{K_0}\|\mub_{\pi(j)}(\Lambdab)-\ub_j\|.
\end{equation*}
Each $\mathsf{cost}_\pi$ is continuous in $\Lambdab$ (composition of continuous maps).
For $\pi\in\Pi$, set the open “uniqueness” region
\begin{equation*}
\mathcal U_\pi := \Big\{\Lambdab\in\Theta_K: \mathsf{cost}_\pi(\Lambdab)<\mathsf{cost}_\rho(\Lambdab)  \forall \rho\neq\pi\Big\}.
\end{equation*}
On $\mathcal U_\pi$ the minimizer is uniquely $\pi$, hence
\begin{equation*}
\mathcal{M}(\Lambdab)=\big(\mub_{\pi(1)}(\Lambdab)^T,\dots,\mub_{\pi(K_0)}(\Lambdab)^T\big)^T,
\end{equation*}
which is continuous there because each \(\mub_k\) is continuous.

Equip $\Pi$ with the lexicographic order $\prec_{\mathrm{lex}}$: for $\pi,\rho\in\Pi$, write $\pi\prec_{\mathrm{lex}}\rho$ if there exists $j^\star$ with
$\pi(j)=\rho(j)$ for all $j<j^\star$ and $\pi(j^\star)<\rho(j^\star)$. For each $\pi\in\Pi$, define the tie region
\begin{equation*}
\mathcal T_\pi := \Big\{\Lambdab\in\Theta_K: 
\mathsf{cost}_\pi(\Lambdab)\le \mathsf{cost}_\eta(\Lambdab) \forall\eta\in\Pi
 \text{ and } 
\mathsf{cost}_\pi(\Lambdab)<\mathsf{cost}_\rho(\Lambdab) \forall\rho\in\Pi\text{ with }\rho\prec_{\mathsf{lex}}\pi
\Big\}.
\end{equation*}
Here, each set $\{\mathsf{cost}_\pi\le \mathsf{cost}_\eta\}$ is closed and each
$\{\mathsf{cost}_\pi<\mathsf{cost}_\rho\}$ is open, so $\mathcal T_\pi$ is Borel.
On $\mathcal T_\pi$, the tie-break selects $\pi$, hence the same formula for $\mathcal{M}$ holds and $\mathcal{M}$ is continuous there.

Finally, $\Theta_K = \bigcup_{\pi\in\Pi}\mathcal{T}_\pi$
is a finite Borel partition on whose pieces $\mathcal{M}$ is continuous.
Therefore $\mathcal{M}$ is Borel measurable on $\Theta_K$.
\end{proof}

\begin{proof}[Proof of Lemma~\ref{lem:indicator-continuity}]
By compactness of $\mathcal S\times\Theta_K$ and smoothness of the
Gaussian-mixture densities in $\xb$, we have
\[
C_q:=\sup_{\Lambdab\in\Theta_K}\|q_\Lambdab\|_{\mathrm{BL}}<\infty.
\] 
If $G_n\overset{w}{\to}P$, we have $d_{\mathrm{BL}}(G_n,P)\to0$, and therefore
\begin{equation*}
    \sup_{\Lambdab\in\Theta_K}
    |G_n(q_\Lambdab)-P(q_\Lambdab)|
    \le C_q\,d_{\mathrm{BL}}(G_n,P)
    \to 0,
\end{equation*}
i.e. the corresponding objectives converge uniformly on $\Theta_K$. Since $\Lambdab^\star_{\omega_0}$ is the unique population maximizer, the argmax theorem yields
$\Phi(G_n)\to\Lambdab^\star_{\omega_0}$.
By Proposition~\ref{prop:mode-matching-smoothness}, $\mathcal{M}$ is $\mathcal C^1$ (hence continuous) at $\Lambdab^\star_{\omega_0}$. Furthermore, $\Eb_j$ is linear and
the ball-membership map is continuous away from the boundary, e.g. via
\begin{equation*}
\psi_j(\yb):=r_j-\|\yb-\ub_j\| \quad\text{(continuous),}\qquad
\mathds{1}\{\yb\in U_j\}=\mathds{1}\{\psi_j(\yb)>0\}.
\end{equation*}
Therefore
\begin{equation*}
\psi_j \big(\Eb_j \mathcal{M}(\Phi(G_n))\big) \longrightarrow \psi_j \big(\Eb_j \mathcal{M}(\Lambdab^\star_{\omega_0})\big).
\end{equation*}
By the positive-margin assumption, the limit is nonzero, and since the map
$x\mapsto\mathds{1}\{x>0\}$ is continuous at every nonzero point, we have
$F_j(G_n)\to F_j(P)$.
\end{proof}

\subsubsection{Proof of Theorem~\ref{thm:bootstrap-params-max-cor_main}} 

Theorem~\ref{thm:bootstrap-params-max-cor_main} is a direct consequence of Theorem~\ref{thm:consistency_main}(iii) and the following result (Thm.~\ref{thm:bootstrap-params-max}), which we prove below.

\begin{theorem}[Bootstrap validity for parameters]\label{thm:bootstrap-params-max}
Under Assumption~\ref{ass:consistency},
\begin{equation*}
d_{\mathrm{BL}} \left(
\operatorname{Law}^\ast \big(\sqrt{N}(\widehat\Lambdab^\ast_{\omega_0,N}-\widehat\Lambdab_{\omega_0,N})\big),\
\operatorname{Law} \big(\sqrt{N}(\widehat\Lambdab_{\omega_0,N}-\Lambdab^\star_{\omega_0})\big)
\right) \xrightarrow{P} 0,
\end{equation*}
where 
$\operatorname{Law}(\cdot)$ denotes law under $P$ and $\operatorname{Law}^\ast(\cdot)$ law under $P^\ast$.
\end{theorem}

\begin{proof}[Proof of Theorem~\ref{thm:bootstrap-params-max}]
Let $\Phi$ denote the locally unique solution map of
Lemma~\ref{lem:argmax-hadamard}. By Theorem~\ref{thm:consistency_main},
\begin{equation*}
    \widehat\Lambdab_{\omega_0,N}
    \xrightarrow{P}\Lambdab^\star_{\omega_0}.
\end{equation*}
Moreover, since $\mathcal Q:=\{q_\Lambdab:\Lambdab\in\Theta_K\}$ is
$P$-Donsker by Lemma~\ref{lem:gc-donsker-mixture},
\[
\sup_{\Lambdab\in\Theta_K}
|(P_N^\ast-P_N)(q_\Lambdab)|
=O_{P^\ast}(N^{-1/2})
\]
in $P$-probability. Together with the Glivenko--Cantelli property of
$\mathcal Q$ and uniqueness of $\Lambdab^\star_{\omega_0}$, the argmax
theorem gives
\[
\widehat\Lambdab^\ast_{\omega_0,N}
\xrightarrow{P^\ast}\Lambdab^\star_{\omega_0}
\]
in $P$-probability.

By Lemma~\ref{lem:boot-clt} and the triangle inequality,
\begin{equation*}
    \|P_N-P\|_{\mathcal F}\xrightarrow{P}0,
    \qquad
    \|P_N^\ast-P\|_{\mathcal F}\xrightarrow{P^\ast}0
\end{equation*}
in $P$-probability. Hence, with probability tending to one for the empirical maximizer,
and with conditional probability tending to one in $P$-probability
for the bootstrap maximizer, both lie in the neighborhood of
$\Lambdab^\star_{\omega_0}$ where the score equation has the unique
solution given by $\Phi$. Since $\Lambdab^\star_{\omega_0}\in\operatorname{int}(\Theta_K)$, this
neighborhood can be chosen inside the interior of $\Theta_K$. Hence
\[
    \widehat\Lambdab_{\omega_0,N}=\Phi(P_N)
\]
with probability tending to one, and
\[
    \widehat\Lambdab^\ast_{\omega_0,N}=\Phi(P_N^\ast)
\]
with conditional probability tending to one, in $P$-probability.

Now, by Lemma~\ref{lem:boot-clt},
\begin{equation*}
\sqrt{N}(P_N^\ast-P_N)
\xrightarrow[]{w}_\ast \mathbb G_P
\quad\text{in }\ell^\infty(\mathcal F),
\end{equation*}
conditionally on the data, in $P$-probability, with
$\mathbb G_P\in\ell_{\mathrm{lin}}^\infty(\mathcal F)$ almost surely.
Therefore, Lemma~\ref{lem:argmax-hadamard} and the bootstrap functional
delta method \citep[Thm.~3.9.11]{van1996weak} yield
\begin{equation*}
\sqrt{N}
(\widehat\Lambdab^\ast_{\omega_0,N}
-\widehat\Lambdab_{\omega_0,N})
\xrightarrow[]{\mathcal D}
D\Phi_P(\mathbb G_P)
\sim\mathcal N(0,\Wb^\star_{\omega_0}),
\end{equation*}
conditionally on the data, in $P$-probability.
Together with Theorem~\ref{thm:consistency_main}(iii), which gives the
same limiting distribution for
$\sqrt N(\widehat\Lambdab_{\omega_0,N}-\Lambdab^\star_{\omega_0})$,
the triangle inequality for $d_{\mathrm{BL}}$ gives
\begin{equation*}
d_{\mathrm{BL}} \left(
\operatorname{Law}^\ast \big(\sqrt{N}(\widehat\Lambdab^\ast_{\omega_0,N}-\widehat\Lambdab_{\omega_0,N})\big),\
\operatorname{Law} \big(\sqrt{N}(\widehat\Lambdab_{\omega_0,N}-\Lambdab^\star_{\omega_0})\big)
\right) \xrightarrow{P} 0.
\end{equation*}
\end{proof}

\subsubsection{Proof of Proposition~\ref{prop:inexact_main}}

\begin{proof}[Proof of Proposition~\ref{prop:inexact_main}]
Let $\widehat\Hb_{\omega_0,N}(\Lambdab) :=-\nabla_{\Lambdab}^2 \widehat{\mathcal L}_{\omega_0,N}(\Lambdab)$.
By Theorem~\ref{thm:consistency_main} and the assumed consistency of
$\widetilde\Lambdab_{\omega_0,N}$, both
$\widehat\Lambdab_{\omega_0,N}$ and
$\widetilde\Lambdab_{\omega_0,N}$ converge in probability to
$\Lambdab^\star_{\omega_0}$. Moreover, by Assumption~\ref{ass:consistency}~(B3),
continuity of the population Hessian, and uniform convergence of the empirical Hessian,
there exist a ball neighborhood $\mathcal V$ of $\Lambdab^\star_{\omega_0}$ and $c>0$ such that $\inf_{\Lambdab\in\mathcal V} \mathrm{eig}_{\min} \bigl(\widehat\Hb_{\omega_0,N}(\Lambdab)\bigr) \ge c $
with probability tending to one.

Hence, with probability tending to one, the segment joining
$\widehat\Lambdab_{\omega_0,N}$ and
$\widetilde\Lambdab_{\omega_0,N}$ lies in $\mathcal V$. Since
$\nabla_{\Lambdab}\widehat{\mathcal L}_{\omega_0,N}
(\widehat\Lambdab_{\omega_0,N})=0$, the integral form of the mean-value theorem gives
\begin{equation*}
\nabla_{\Lambdab}\widehat{\mathcal L}_{\omega_0,N}
(\widetilde\Lambdab_{\omega_0,N})
=
-\overline\Hb_{\omega_0,N}
(\widetilde\Lambdab_{\omega_0,N}
-\widehat\Lambdab_{\omega_0,N}),
\end{equation*}
where
\begin{equation*}
\overline\Hb_{\omega_0,N}
:=
\int_0^1
\widehat\Hb_{\omega_0,N}
\left(
\widehat\Lambdab_{\omega_0,N}
+t(\widetilde\Lambdab_{\omega_0,N}
-\widehat\Lambdab_{\omega_0,N})
\right)dt.
\end{equation*}
Therefore,
$\mathrm{eig}_{\min}(\overline\Hb_{\omega_0,N})\ge c$ with probability tending to one, and
\begin{equation*}
\|\widetilde\Lambdab_{\omega_0,N}-\widehat\Lambdab_{\omega_0,N}\|
\le c^{-1}
\big\|
\nabla_{\Lambdab}\widehat{\mathcal L}_{\omega_0,N}
(\widetilde\Lambdab_{\omega_0,N})
\big\|
=o_P(N^{-1/2}).
\end{equation*}
Thus
$\sqrt N(\widetilde\Lambdab_{\omega_0,N}
-\widehat\Lambdab_{\omega_0,N})=o_P(1)$.

The same argument holds conditionally for the bootstrap quantities, using the consistency of
$\widetilde\Lambdab^\ast_{\omega_0,N}$ and
$\widehat\Lambdab^\ast_{\omega_0,N}$, and gives
\begin{equation*}
\|\widetilde\Lambdab^\ast_{\omega_0,N}
-\widehat\Lambdab^\ast_{\omega_0,N}\|
=o_{P^\ast}(N^{-1/2})
\end{equation*}
in $P$-probability. The conclusions of
Theorem~\ref{thm:bootstrap-params-max-cor_main} therefore follow by Slutsky's theorem.
Under Assumption~\ref{as:sep-match_main}, the local $\mathcal C^1$ property of
$\mathcal M$ gives the corresponding $o(N^{-1/2})$ equivalence after matching, and the conclusions of
Theorem~\ref{thm:modes-CLT-bootstrap_main} also follow by Slutsky's theorem.
\end{proof}

\subsubsection{Proof of Theorem~\ref{thm:modes-CLT-bootstrap_main}}

Theorem~\ref{thm:modes-CLT-bootstrap_main} is obtained by combining Proposition~\ref{prop:mode-matching-smoothness}, Theorem~\ref{thm:clt-matched-modes}, and Corollary~\ref{cor:per-mode-projections}, proven below.

\paragraph{Mode-matching map smoothness}

\begin{proposition}[Mode-matching map smoothness]\label{prop:mode-matching-smoothness}
Under Assumption~\ref{as:sep-match_main}, there exists a neighborhood $\mathcal{V} \subset \Theta_K\cap \mathcal G_{\delta/2}$ of $\Lambdab^\star_{\omega_0}$ such that $\mathcal{M}$ is $\mathcal{C}^1$ on $\mathcal{V}$. Its Jacobian at $\Lambdab^\star_{\omega_0}$ is $\Jb=\nabla_\Lambdab \mathcal{M}(\Lambdab^\star_{\omega_0})$.
\end{proposition}

\begin{proof}[Proof of Proposition~\ref{prop:mode-matching-smoothness}]
By Assumption~\ref{as:sep-match_main}, there exists $\delta>0$ such that, at
$\Lambdab^\star_{\omega_0}$, each target $\ub_j$ has a unique nearest component
mean, with gap at least $\delta$, and the corresponding nearest-component
indices are pairwise distinct. By continuity of
$\Lambdab\mapsto\|\mub_k(\Lambdab)-\ub_j\|$, these nearest-component indices
remain unchanged in a neighborhood $\mathcal V\subset\mathcal G_{\delta/2}$
of $\Lambdab^\star_{\omega_0}$.

Write $\pi^\star(j)$ for the nearest-component index associated with $\ub_j$.
Since the indices $\pi^\star(1),\dots,\pi^\star(K_0)$ are pairwise distinct,
$\pi^\star$ defines an admissible injection. Moreover, since each
$\mub_{\pi^\star(j)}(\Lambdab)$ uniquely minimizes the distance to $\ub_j$,
$\pi^\star$ is the unique minimizer of the total matching cost on
$\mathcal V$. Hence the matching is fixed on $\mathcal V$, and for all
$\Lambdab\in\mathcal V$,
\begin{equation*}
\mathcal{M}(\Lambdab)
=
\big(
\mub_{\pi^\star(1)}(\Lambdab)^T,\dots,
\mub_{\pi^\star(K_0)}(\Lambdab)^T
\big)^T.
\end{equation*}
Since
$\mb(\Lambdab)
=(\mub_1(\Lambdab)^T,\dots,\mub_K(\Lambdab)^T)^T$
is $\mathcal C^1$ in a neighborhood of
$\Lambdab^\star_{\omega_0}$, the representation above shows that
$\mathcal M$ is $\mathcal C^1$ on $\mathcal V$. In particular,
\begin{equation*}
\Jb
=
\nabla_{\Lambdab}\mathcal M(\Lambdab^\star_{\omega_0})
=
\begin{bmatrix}
\nabla_{\Lambdab}\mub_{\pi^\star(1)}(\Lambdab^\star_{\omega_0})\\[-1mm]
\vdots\\[-1mm]
\nabla_{\Lambdab}\mub_{\pi^\star(K_0)}(\Lambdab^\star_{\omega_0})
\end{bmatrix}.
\end{equation*}
\end{proof}

\paragraph{CLT for matched modes}

\begin{theorem}[CLT and bootstrap for $\mathcal{M}$]\label{thm:clt-matched-modes}
Under Assumptions~\ref{ass:consistency} and \ref{as:sep-match_main}, as $N \to \infty$:
\begin{equation*}
\sqrt{N}\big(\mathcal{M}(\widehat\Lambdab_{\omega_0,N})-\mathcal{M}(\Lambdab^\star_{\omega_0})\big)\xrightarrow[]{\mathcal{D}} \mathcal N(0,\Cb_{\mathcal{M}}),
\quad
\sqrt{N}\big(\mathcal{M}(\widehat\Lambdab^{\ast}_{\omega_0,N})-\mathcal{M}(\widehat\Lambdab_{\omega_0,N})\big)\xrightarrow[]{\mathcal{D}}\mathcal N(0,\Cb_{\mathcal{M}}),
\end{equation*}
conditionally on $\Xb_{1:N}$, in $P$-probability, with $\Cb_{\mathcal{M}} = \Jb \Wb_{\omega_0} \Jb^T$.
\end{theorem}

\begin{proof}[Proof of Theorem~\ref{thm:clt-matched-modes}]
By Theorem~\ref{thm:consistency_main}(iii),
\(
\sqrt{N} (\widehat\Lambdab_{\omega_0,N}-\Lambdab^\star_{\omega_0})
\xrightarrow[]{\mathcal{D}}\mathcal{N}(0,\Wb_{\omega_0}).
\)
Furthermore, by Proposition~\ref{prop:mode-matching-smoothness}, $\mathcal{M}$ is $\mathcal C^1$ at $\Lambdab^\star_{\omega_0}$ with Jacobian $\Jb$, so the functional delta method
\citep[Thm.~20.8]{van1998asymptotic} yields
\begin{equation*}
\sqrt{N} \big(\mathcal{M}(\widehat\Lambdab_{\omega_0,N})-\mathcal{M}(\Lambdab^\star_{\omega_0})\big)
\xrightarrow[]{\mathcal{D}}\mathcal N \big(0, \Jb \Wb_{\omega_0}\Jb^T\big),
\end{equation*}
which proves the first claim.

Finally, by the bootstrap functional delta method \citep[Thm~3.9.11]{van1996weak},
\begin{equation*}
\sqrt N\big(\mathcal{M}(\widehat\Lambdab^\ast_{\omega_0,N})-\mathcal{M}(\widehat\Lambdab_{\omega_0,N})\big)
\xrightarrow[]{\mathcal{D}}
\mathcal N(0, \Jb \Wb_{\omega_0}\Jb^T)
\end{equation*}
conditionally on $\Xb_{1:N}$, in $P$-probability, which gives the second claim.
\end{proof}

\begin{corollary}[Setwise convergence]\label{cor:setwise-matched-modes}
As a consequence of Theorem~\ref{thm:clt-matched-modes}, for any Borel set
$B\subset\mathbb R^{dK_0}$ whose boundary has zero probability under
$\mathcal{N}(0,\Cb_{\mathcal{M}})$, as $N \to \infty$:
\begin{equation*}
P^\ast\left(
\sqrt N\left\{
\mathcal{M}(\widehat\Lambdab^\ast_{\omega_0,N})
-\mathcal{M}(\widehat\Lambdab_{\omega_0,N})
\right\}\in B
\right)
-
P\left(
\sqrt N\left\{
\mathcal{M}(\widehat\Lambdab_{\omega_0,N})
-\mathcal{M}(\Lambdab^\star_{\omega_0})
\right\}\in B
\right)
\xrightarrow{P} 0.
\end{equation*}
In particular, bootstrap confidence regions for each matched mode coordinate are asymptotically valid.
\end{corollary}

\begin{proof}[Proof of Corollary~\ref{cor:setwise-matched-modes}]
This is directly implied by Theorem~\ref{thm:clt-matched-modes} and the Portmanteau theorem \citep[Lem.~2.2]{van1998asymptotic}.
\end{proof}

\paragraph{Per-mode projections}

\begin{corollary}[Per-mode projections]\label{cor:per-mode-projections}
As a consequence of Theorem~\ref{thm:clt-matched-modes}, for the $j$-th matched mode $\mathcal{M}_j := \Eb_j \mathcal{M}$ with $\Eb_j$ extracting the $j$-th block, as $N \to \infty$:
\begin{equation*}
\sqrt N\big(\mathcal{M}_j(\widehat\Lambdab_{\omega_0,N}) - \mathcal{M}_j(\Lambdab^\star_{\omega_0})\big)\xrightarrow[]{\mathcal{D}} \mathcal N(0,\Cb_j),
\quad
\sqrt{N}\big(\mathcal{M}_j(\widehat\Lambdab^{\ast}_{\omega_0,N})-\mathcal{M}_j(\widehat\Lambdab_{\omega_0,N})\big)\xrightarrow[]{\mathcal{D}}\mathcal N(0,\Cb_j),
\end{equation*}
and bootstrap confidence regions for $\mathcal{M}_j$ are asymptotically valid.
\end{corollary}

\begin{proof}[Proof of Corollary~\ref{cor:per-mode-projections}]
The asymptotic normality results follow directly from
Theorem~\ref{thm:clt-matched-modes} and the linearity of $\Eb_j$.
The validity of the corresponding bootstrap confidence regions follows
from Corollary~\ref{cor:setwise-matched-modes}.
\end{proof}

\subsubsection{Proof of Proposition~\ref{prop:stability-max}}
\label{app:proof_stability}

\begin{proof}[Proof of Proposition~\ref{prop:stability-max}]
By Lemma~\ref{lem:indicator-continuity}, $F_j$ is continuous at $P$ with respect to weak convergence.

\emph{(i) Conditional LLN.} 
Given the data, $\{F_j(P_N^{\ast(\ell)})\}_{\ell=1}^L$ are i.i.d. Bernoulli with mean $\mathbb{E}^\ast[F_j(P_N^\ast)]=\tau_{j,N}$ and finite variance. Hence, by the strong law,
$s_j\to \tau_{j,N}$ almost surely under the bootstrap (a.s.$^\ast$), and
$\operatorname{Var}^\ast(s_j)=\operatorname{Var}^\ast(F_j(P_N^\ast))/L\to 0$. Hoeffding’s inequality yields the tail bound.

\emph{(ii) Large-$N$ limits.} 
Since $P_N \xrightarrow[]{w} P$ and $F_j$ is continuous at $P$, we have $F_j(P_N)\to F_j(P)$ in probability, so
$\pi_{j,N}\to \pi_j:=F_j(P)$. For the bootstrap mean, standard nonparametric bootstrap consistency at continuity points gives
$F_j(P_N^\ast) \xrightarrow{P} F_j(P)$ in $P$-probability. Because $F_j\in[0,1]$ is bounded, conditional bounded convergence implies
$\mathbb{E}^\ast[F_j(P_N^\ast)]\xrightarrow{P} F_j(P)=\pi_j$.

Finally, for any $L=L_N\to\infty$ as $N \to \infty$,
\begin{equation*}
|s_j-\pi_j| \le |s_j-\tau_{j,N}| + |\tau_{j,N}-\pi_j| \xrightarrow{P} 0.
\end{equation*}
By~(i) and~(ii), we obtain the last claim $s_j\xrightarrow{P}\pi_j$.
\end{proof}

\subsection{Proofs for Section~\ref{app:truncation_approximation} }
\label{app:proof_truncation}

\subsubsection{Proof of Proposition~\ref{prop:trunc_gap}}
\label{app:proof_truncation_prop}

\begin{proof}[Proof of Proposition~\ref{prop:trunc_gap}]
    Write $q^\mathcal{S} = q/\psi$, $\varepsilon = 1 - \psi$. On the one hand, we have
    \begin{equation*}
        \left\lvert \mathbb{E}_{q^{\mathcal S}}[p] - \mathbb{E}_q[p] \right\rvert = \left\lvert \frac{1}{\psi} - 1 \right\rvert \int_{\mathcal{S}} p q \le M \varepsilon.
    \end{equation*}
    On the other hand, we have $\mathcal{H}_\mathcal{S}(q^\mathcal{S}) = \frac{1}{\psi} \mathcal{H}_\mathcal{S}(q) + \log \psi$. Hence,
    \begin{equation*}
        \mathcal{H}_\mathcal{S}(q^\mathcal{S}) - \mathcal{H}_\mathcal{S}(q) = \left( \frac{1}{\psi} - 1 \right) \mathcal{H}_\mathcal{S}(q) + \log \psi = \frac{\varepsilon}{\psi} \mathcal{H}_\mathcal{S}(q) + \log (1-\varepsilon).
    \end{equation*}
    On a compact parameter set and bounded $\mathcal{S}$, $C_q=\sup_{\lambdab}\sup_{\xb\in\mathcal S}\big|\log q_\lambdab(\xb)\big|<\infty$ is well-defined and $\left\lvert \mathcal{H}_\mathcal{S}(q) \right\rvert \le C_q \int_\mathcal{S} q = C_q \psi$. Therefore,
    \begin{equation*}
        \left\lvert \mathcal{H}_\mathcal{S}(q^\mathcal{S}) - \mathcal{H}_\mathcal{S}(q) \right\rvert \le C_q \varepsilon + \lvert \log (1 - \varepsilon) \rvert.
    \end{equation*}
    Combining the two inequalities, we obtain the pointwise bound. The uniform bound follows by taking suprema over $\Theta$.
\end{proof}

\subsubsection{Proof of Theorem~\ref{thm:trunc_equiv}}
\label{app:proof_truncation_thm}

Before proving Theorem~\ref{thm:trunc_equiv}, we derive three supporting lemmas.

\begin{lemma}[Argmax stability]
\label{lem:argmax_stability}
Let $\Theta$ be compact and let $F,G:\Theta\to\mathbb{R}$ be continuous functions. 
Assume $F$ has a maximizer $\lambdab^\star$ and that for some $\rho>0$ the value gap
\begin{equation*}
m(\rho):=F(\lambdab^\star)-  \sup_{\|\lambdab-\lambdab^\star\|\ge \rho}F(\lambdab)>0.
\end{equation*}
If $\|G-F\|_\infty:=\sup_{\lambdab\in\Theta}|G(\lambdab)-F(\lambdab)|\le m(\rho)/3$, then every maximizer of $G$ lies in $B(\lambdab^\star,\rho)$. 

\end{lemma}

\begin{proof}[Proof of Lemma~\ref{lem:argmax_stability}]
Fix $\rho>0$ with $m(\rho)>0$ and set $\delta:=\|G-F\|_\infty$. For any $\lambdab$ with $\|\lambdab-\lambdab^\star\|\ge\rho$,
\begin{equation*}
G(\lambdab)\le F(\lambdab)+\delta \le F(\lambdab^\star)-m(\rho)+\delta,
\end{equation*}
while $G(\lambdab^\star)\ge F(\lambdab^\star)-\delta$. If $\delta<m(\rho)/2$ (hence, under $\delta\le m(\rho)/3$), then
$G(\lambdab^\star)>G(\lambdab)$ for all $\|\lambdab-\lambdab^\star\|\ge\rho$, so every maximizer of $G$ lies in $B(\lambdab^\star,\rho)$. 

\end{proof}

\begin{lemma}[Uniform tail]
\label{lem:margin_uniform_tail}
Fix $\omega>0$. 
Assume
\begin{equation*}
\mub(\lambdab)\in\operatorname{int}(\mathcal S), \qquad
\mathrm{dist} \big(\mub(\lambdab),\partial\mathcal S\big) \ge 
M_0 \sqrt{\mathrm{eig}_{\max}(\Sigma(\lambdab))},  \qquad \forall \lambdab\in\Theta.
\end{equation*}
Then the uniform tail bound
\begin{equation*}
\varepsilon_{\max}:=\sup_{\lambdab\in\Theta}\Big(1- \int_{\mathcal S} q_\lambdab(\xb) d\xb\Big) \le 2^{d/2} e^{-M_0^2/4}
\end{equation*}
holds for some constant $C > 0$.
\end{lemma}

\begin{proof}[Proof of Lemma~\ref{lem:margin_uniform_tail}]
Let $\lambdab \in \Theta$, $\Xb \sim \mathcal{N}(\mub, \Sigmab)$, and $r := M_0 \sqrt{\mathrm{eig}_{\max}(\Sigmab)}$. The margin condition $\text{dist}(\mub(\lambdab), \partial \mathcal{S}) \geq r$ implies $\overline{B(\mub, r)} \subset \mathcal{S}$, so $\varepsilon(\mub, \Sigmab) = \mathbb{P}(\Xb \notin \mathcal{S}) \leq \mathbb{P}(\|\Xb - \mub\| \geq r)$. Since $\Sigmab \preceq \mathrm{eig}_{\max}(\Sigmab) \Ib$, for $\Zb \sim \mathcal{N}(0, \Ib)$, there exists $C > 0$ such that
\begin{equation*}
    \mathbb{P}(\|\Xb - \mub\| \geq r) \leq \mathbb{P}(\|\Zb\| \geq r / \sqrt{\mathrm{eig}_{\max}(\Sigmab)}) = \mathbb{P}(\|\Zb\| \geq M_0) \leq 2^{d/2} e^{-M_0^2/4},
\end{equation*}
by chi-squared tail bounds for $\|\Zb\|^2$. Thus, $\varepsilon_{\max} = \sup_{\lambdab \in \Theta} \varepsilon(\lambdab) \leq 2^{d/2} e^{-M_0^2/4}$.
\end{proof}

\begin{proof}[Proof of Theorem~\ref{thm:trunc_equiv}]
By Lemma~\ref{lem:margin_uniform_tail}, $\varepsilon_{\max}\le 2^{d/2} e^{-M_0^2/4}$. Proposition~\ref{prop:trunc_gap} then yields the uniform bound
\begin{equation*}
\sup_{\lambdab\in\Theta}\big|G_\omega(\lambdab)-F_\omega(\lambdab)\big|
 \le \delta_\omega(\varepsilon_{\max}).
\end{equation*}
Fix $\rho>0$ and set $m(\rho):=F_\omega(\lambdab_\omega^\star)-\sup_{\|\lambdab-\lambdab_\omega^\star\|\ge\rho}F_\omega(\lambdab)>0$. Choose $M_0$ large enough that $\delta_\omega(\varepsilon_{\max})<m(\rho)/3$. By Lemma~\ref{lem:argmax_stability}, every maximizer of $G_\omega$ lies in $B(\lambdab_\omega^\star,\rho)$.
\end{proof}

\section{Algorithms and variants}
\label{app:gerve_derivations}

This section details the update rules for GERVE for different variational families and provides pseudo-codes for implementation. For each case, we:
\begin{itemize}
    \item Recall the parameterization and its link to natural gradients.
    \item Present the specific update rules and their mini-batch version.
    \item Give a pseudo-code for practical use.
    \item Provide derivation notes for readers interested in the intermediate steps.
\end{itemize}

\subsection{Natural-gradient update rule}

For distributions in an exponential family or MCEF, we can use (\ref{eq:natural_expectation_duality}) \citep{lin2019fast, khan2023bayesian} so that 
the natural-gradient update rule writes:
\begin{equation*}
\lambdab_{t + 1} = \lambdab_t + \rho_t \nabla_\Mb \widehat{\mathcal{L}}_{\omega_t, N}(\lambdab)|_{\lambdab=\lambdab_t},
\end{equation*}
which will be the starting point for all the following derivations.

\subsection{Fixed-covariance Gaussian case}
\label{app:fcg_derivation}

\paragraph{Parameterization.} We consider $q_\lambdab(\xb) = \mathcal{N}(\xb; \mub, s^{-1} \Ib)$ with fixed precision $s > 0$.
\begin{itemize}
    \item Natural parameter: $\lambdab = s \mub$;
    \item Sufficient statistic: $\xb$;
    \item Expectation parameter: $\mub = \mathbb{E}_{q_{\lambdab}}[\Xb]$.
\end{itemize}
Only $\mub$ is optimized. 

\paragraph{GERVE update rule.} Considering $\nabla_\mub \mathcal H_{\mathcal{S}}(q_\lambdab) \approx \nabla_\mub \mathcal H(q_\lambdab) =  0$, the fixed-covariance Gaussian GERVE update rule is:
\begin{empheq}[box=\fbox]{equation*} 
     \mub_{t+1} = \mub_{t} + \frac{\rho_t}{N} \sum_{i = 1}^N (\Xb_i - \mub_t) q_{\lambdab_t}(\Xb_i).
\end{empheq}
This is equivalent to gradient ascent on the kernel density estimate
\begin{equation*}
    p_h(\mub) = \frac{1}{N} \sum_{i=1}^N \varphi_h(\Xb_i - \mub),
\end{equation*}
where $\varphi_h$ is the Gaussian kernel with bandwidth $h = s^{-1}$. This reduces to the classic Gaussian mean-shift update when the step sizes are chosen adaptively 
$\rho_t = p_h(\mub_t)^{-1} $.

\paragraph{Mini-batch implementation.} For a batch $\Xb^{(t)}_{1:B}$ sampled uniformly with replacement from the dataset, i.e. $J^{(t)}_1,\dots,J^{(t)}_B \overset{\text{i.i.d.}}{\sim} \mathrm{Unif}([N]), \Xb^{(t)}_{1:B} := (\Xb_{J^{(t)}_1}, \dots, \Xb_{J^{(t)}_B})$, define:
\begin{equation*}
    \widehat{\gb}^{(t)}_B = \frac{s}{B} \sum_{i=1}^B (\Xb^{(t)}_i - \mub_t) q_{\lambdab_t}(\Xb^{(t)}_i).
\end{equation*}
Conditioned on $\Xb_{1:N}$, $\widehat{\gb}^{(t)}_B$ is an unbiased estimator of $\nabla_{\mub} \widehat{\mathcal{L}}_{\omega_t, N}(\lambdab)|_{\lambdab = \lambdab_t}$.

The mini-batch version of the update rule is:
\begin{empheq}[box=\fbox]{equation*} 
    \mub_{t+1} = \mub_{t} +  \rho_t s^{-1} \widehat{\gb}^{(t)}_B.
\end{empheq}

\paragraph{Algorithm.} Algorithm~\ref{alg:fcg_nva} presents the pseudo-code corresponding to this implementation of GERVE.

%

\LinesNumbered
\begin{algorithm}[H]
	\caption{Fixed-covariance Gaussian GERVE (with mini-batches)}  \label{alg:fcg_nva}
 \textsc{Given} samples $\Xb_{1:N}$. \\
		 \textsc{Set} $T$, $B$, $s$, $\lambdab_1$, $\rho_{1:T}$. \\
		 \For{$t=1\!:\!T$}{
        \textsc{Sample} $J^{(t)}_{1:B} \overset{\text{i.i.d.}}{\sim} \mathrm{Unif}([N])$, set $\Xb^{(t)}_i \gets \Xb_{J^{(t)}_i}$ for $i=1\!:\!B$. \\
      \textsc{Compute} $\widehat{\gb}^{(t)}_B$ from $\Xb^{(t)}_{1:B}, \lambdab_t$. \\
 \textsc{Update} $\mub_{t+1} = \mub_{t} + \rho_t s^{-1} \widehat{\gb}^{(t)}_B$.
}
		\Return $\lambdab_{T+1} (\text{or } \mub_{T+1}$). 
\end{algorithm}

\paragraph{Derivation details.}
From the natural-gradient rule for exponential families with $\lambdab = s \mub$, we have:
\begin{equation*}
     s \mub_{t+1} = s \mub_{t} + \rho_t \nabla_{\mub} \widehat{\mathcal{L}}_{\omega_t, N}(\lambdab)|_{\lambdab = \lambdab_t}.
\end{equation*} 
The gradient of the empirical objective is:
\begin{equation*}
    \nabla_{\mub} \widehat{\mathcal{L}}_{\omega_t, N}(\lambdab) = \frac{s}{N} \sum_{i = 1}^N (\Xb_i - \mub) q_\lambdab(\Xb_i) + \omega \nabla_\mub \mathcal H_{\mathcal{S}}(q_\lambdab),
\end{equation*}
which yields the update rule above.

\subsection{Gaussian mixture case}
\label{app:mix_derivation}

\paragraph{Parameterization.} For $K$ components:
\begin{equation*}
    q_\Lambdab(\xb) = \sum_{k=1}^K \pi_k q_{\lambdab_k}(\xb),
\end{equation*}
where the weights $\pi_1, \dots, \pi_K$ sum to 1, and each $q_{\lambdab_k}(\xb)$ is a Gaussian density with natural parameters $\lambdab_k$. 
\begin{itemize}
    \item MCEF natural parameters:
\begin{equation*}
    \Lambdab = (v_1, \dots, v_{K-1}, \lambdab_1, \dots, \lambdab_K),
\end{equation*}
where $v_k := \log(\pi_k/\pi_K)$;
    \item MCEF expectation parameters:
\begin{equation*}
    \Mb = (\pi_1, \dots, \pi_{K-1}, \Mb_1, \dots, \Mb_K),
\end{equation*}
where $\Mb_k := (\mb_k^{(1)}, \mb_k^{(2)}) = (\pi_k \mub_k, \pi_k (\Sb_k^{-1} + \mub_k \mub_k^T))$.
\end{itemize}

\paragraph{GERVE update rules.} The Gaussian mixture GERVE update rules are:
\begin{empheq}[box=\fbox]{align*} 
    \Sb_{k,t+1} &= \Sb_{k,t} - \frac{\rho_t}{N} \Sb_{k,t} \sum_{i=1}^N ( (\Xb_i - \mub_{k,t}) (\Xb_i - \mub_{k,t})^T \Sb_{k,t} - \Ib) q_{\lambdab_{k,t}}(\Xb_i) \\
    & \quad - \frac{2 \rho_t \omega_t}{\pi_{k,t}} \nabla_{\Sb^{-1}_k} \mathcal{H}_\mathcal{S}(q_\Lambdab)|_{\Lambdab=\Lambdab_t}, \\
    \mub_{k,t+1} &= \mub_{k,t} +  \frac{\rho_t}{N} \Sb^{-1}_{k,t+1} \Sb_{k,t}  \sum_{i=1}^N  (\Xb_i - \mub_{k,t}) q_{\lambdab_{k,t}}(\Xb_i) + \frac{\rho_t \omega_t}{\pi_{k,t}} \Sb^{-1}_{k,t+1} \nabla_{\mub_k} \mathcal{H}_\mathcal{S}(q_\Lambdab)|_{\Lambdab=\Lambdab_t}, \\
    v_{k,t+1} &= v_{k,t} +  \frac{\rho_t}{N} \sum_{i=1}^N (q_{\lambdab_{k,t}}(\Xb_i) -q_{\lambdab_{K,t}}(\Xb_i) ) + \rho_t \omega_t \nabla_{\pi_k} \mathcal{H}_\mathcal{S}(q_\Lambdab)|_{\Lambdab=\Lambdab_t}.
\end{empheq}

\paragraph{Mini-batch and Monte-Carlo implementation.} For a batch $\Xb^{(t)}_{1:B}$ sampled uniformly with replacement from the dataset, i.e. $J^{(t)}_1,\dots,J^{(t)}_B \overset{\text{i.i.d.}}{\sim} \mathrm{Unif}([N]), \Xb^{(t)}_{1:B} := (\Xb_{J^{(t)}_1}, \dots, \Xb_{J^{(t)}_B})$, define:
\begin{align*}
    \widehat{f}^{(t)}_{k, B} &= \frac{1}{B}\sum_{i = 1}^B (q_{\lambdab_{k,t}}(\Xb^{(t)}_i) - q_{\lambdab_{K,t}}(\Xb^{(t)}_i)), \\
    \widehat{\gb}^{(t)}_{k, B} &= \frac{\pi_{k,t}}{B} \Sb_{k,t} \sum_{i = 1}^B (\Xb^{(t)}_i - \mub_{k,t}) q_{\lambdab_{k,t}}(\Xb^{(t)}_i), \\
    \widehat{\Hb}^{(t)}_{k, B} &= \frac{\pi_{k,t}}{2 B} \Sb_{k,t} \sum_{i = 1}^B ((\Xb^{(t)}_i - \mub_{k,t}) (\Xb^{(t)}_i - \mub_{k,t})^T \Sb_{k,t} - \Ib) q_{\lambdab_{k,t}}(\Xb^{(t)}_i).
\end{align*}
Conditioned on $\Xb_{1:N}$, $\widehat{\gb}^{(t)}_{k,B}$, $\widehat{\Hb}^{(t)}_{k,B}$, and $\widehat{f}^{(t)}_{k,B}$ are unbiased estimators of their full-data counterparts.

Entropy derivatives must also be estimated. For some size $B_e$, draw additional Monte-Carlo samples $\Zb^{(k,t)}_{1:B_{e}} \sim q_{\lambdab_{k,t}}$ to compute $\widehat{\etab}^{(t)}_{k, B_{e}}$, $\widehat{\gammab}^{(t)}_{k, B_{e}}$, and $\widehat{\varphi}^{(t)}_{k, B_{e}}$ (see details below). 

We obtain the practical updates:
\begin{empheq}[box=\fbox]{align*} 
    \Sb_{k,t+1} &= \Sb_{k,t} - \frac{2 \rho_t}{\pi_{k,t}} (\widehat{\Hb}^{(t)}_{k, B} + \omega_t \widehat{\etab}^{(t)}_{k, B_e}), \\
    \mub_{k,t+1} &= \mub_{k,t} + \frac{\rho_t}{\pi_{k,t}} \Sb^{-1}_{k,t+1} (\widehat{\gb}^{(t)}_{k, B} + \omega_t \widehat{\gammab}^{(t)}_{k, B_e}), \\
    v_{k,t+1} &= v_{k,t} + \rho_t (\widehat{f}^{(t)}_{k, B} + \omega_t \widehat{\varphi}^{(t)}_{k, B_{e}}).
\end{empheq}

\paragraph{Algorithm.}  Algorithm~\ref{alg:nva_gm} depicts our proposed implementation of Gaussian mixture GERVE.


\LinesNumbered
\begin{algorithm}[H]
	\caption{Gaussian mixture GERVE (with mini-batches)}  \label{alg:nva_gm}
    \textsc{Given} samples $\Xb_{1:N}$. \\
	\textsc{Set} $T$, $B$, $B_e$, $K$, $\Lambdab_1$, $\rho_{1:T}$, $\omega_{1:T}$. \\
	\For{$t=1\!:\!T$}{
        \textsc{Sample} $J^{(t)}_{1:B} \overset{\text{i.i.d.}}{\sim} \mathrm{Unif}([N])$, set $\Xb^{(t)}_i \gets \Xb_{J^{(t)}_i}$ for $i=1\!:\!B$. \\
        \For {$k=1\!:\!K$}{
            \textsc{Compute} $\widehat{\gb}^{(t)}_{k, B}$, $\widehat{\Hb}^{(t)}_{k, B}$ from $\Xb^{(t)}_{1:B}, \Lambdab_t$. \\
			\textsc{Sample} $\Zb^{(k,t)}_{1:B_e} \overset{\text{i.i.d.}}{\sim} \mathcal N(\mub_{k,t},\Sb_{k,t}^{-1})$. \\
            \textsc{Compute} $\widehat{\gammab}^{(t)}_{k, B_e}$, $\widehat{\etab}^{(t)}_{k, B_e}$ from $\Zb^{(k,t)}_{1:B_e}, \Lambdab_t$. \\
            \textsc{Update} $\Sb_{k,t+1}, \mub_{k,t+1}$.
        } 
        \For {$k=1\!:\!(K\!-\!1)$}{
            \textsc{Compute} $\widehat{f}^{(t)}_{k, B}$, $\widehat{\varphi}^{(t)}_{k, B_{e}}$ from $\Xb^{(t)}_{1:B}, \Zb^{(k,t)}_{1:B_e}, \Zb^{(K,t)}_{1:B_e}, \Lambdab_t$.  \\
            \textsc{Update} $v_{k, t+1}$.
        }
    }
	\Return $\Lambdab_{T+1}$.
\end{algorithm}

\paragraph{Derivation details.} The natural-gradient update rules are:
\begin{align*}
    v_{k, t+1} &= v_{k, t} + \rho_t \nabla_{\pi_k} \widehat{\mathcal{L}}_{\omega_t, N}(\Lambdab)|_{\Lambdab = \Lambdab_{t}}, \\
    \Sb_{k, t+1} \mub_{k, t+1} &= \Sb_{k, t} \mub_{k, t} + \rho_t \nabla_{\mb_k^{(1)}} \widehat{\mathcal{L}}_{\omega_t, N}(\Lambdab)|_{\Lambdab = \Lambdab_{t}}, \\
    - \frac{1}{2}\Sb_{k, t+1} &= - \frac{1}{2}\Sb_{k, t} + \rho_t \nabla_{\mb_k^{(2)}} \widehat{\mathcal{L}}_{\omega_t, N}(\Lambdab)|_{\Lambdab = \Lambdab_{t}}.
\end{align*}
We rearrange them to obtain:
\begin{align*}
    \mub_{k,t+1} &= \mub_{k,t} + \rho_t \Sb_{k,t+1}^{-1}  ( \nabla_{\mb^{(1)}_k} \widehat{\mathcal{L}}_{\omega_t, N}(\Lambdab)|_{\Lambdab = \Lambdab_{t}} + 2 (\nabla_{\mb^{(2)}_k} \widehat{\mathcal{L}}_{\omega_t, N}(\Lambdab)|_{\Lambdab = \Lambdab_{t}}) \mub_{k,t}), \\
    \Sb_{k,t+1} &= \Sb_{k,t} - 2 \rho_t \nabla_{\mb^{(2)}_k} \widehat{\mathcal{L}}_{\omega_t, N}(\Lambdab)|_{\Lambdab = \Lambdab_{t}}, \\
    v_{k, t+1} &= v_{k, t} + \rho_t \nabla_{\pi_k} \widehat{\mathcal{L}}_{\omega_t, N}(\Lambdab)|_{\Lambdab = \Lambdab_{t}}.
\end{align*}

We apply the chain rule to express the gradients w.r.t. the expectation parameters as gradients w.r.t. $\mub_k$ and $\Sb_k^{-1}$:
\begin{align*}
\nabla_{\mb_k^{(1)}}\widehat{\mathcal{L}}_{\omega_t, N}(\Lambdab) &= \frac{1}{\pi_k} \left(\nabla_{\mub_k} \widehat{\mathcal{L}}_{\omega_t, N}(\Lambdab) - 2 (\nabla_{\Sb_k^{-1}} \widehat{\mathcal{L}}_{\omega_t, N}(\Lambdab)) \mub_k \right),  \\
\nabla_{\mb_k^{(2)}} \widehat{\mathcal{L}}_{\omega_t, N}(\Lambdab)  &= \frac{1}{\pi_k} \nabla_{\Sb_k^{-1}} \widehat{\mathcal{L}}_{\omega_t, N}(\Lambdab).
\end{align*}
These gradients are
\begin{align*}
    \nabla_{\mub_k} \widehat{\mathcal{L}}_{\omega_t, N}(\Lambdab) &= \nabla_{\mub_k} \left(\frac{1}{N} \sum_{i = 1}^N q_{\Lambdab}(\Xb_i)\right) + \omega_t \nabla_{\mub_k} \mathcal{H}_\mathcal{S}(q_\Lambdab) \\
    &= \frac{\pi_k}{N} \Sb_k \sum_{i = 1}^N (\Xb_i - \mub_k) q_{\lambdab_k}(\Xb_i) + \omega_t \nabla_{\mub_k} \mathcal{H}_\mathcal{S}(q_\Lambdab),
\end{align*}
and
\begin{align*}
    \nabla_{\Sb_k^{-1}} \widehat{\mathcal{L}}_{\omega_t, N}(\Lambdab) &= \nabla_{\Sb_k^{-1}} \left(\frac{1}{N} \sum_{i = 1}^N q_{\Lambdab}(\Xb_i)\right) + \omega_t \nabla_{\Sb_k^{-1}} \mathcal{H}_\mathcal{S}(q_\Lambdab) \\
    &= \frac{\pi_k}{2 N} \Sb_k \sum_{i = 1}^N ((\Xb_i - \mub_k) (\Xb_i - \mub_k)^T \Sb_k - \Ib) q_{\lambdab_k}(\Xb_i) + \omega_t \nabla_{\Sb_k^{-1}} \mathcal{H}_\mathcal{S}(q_\Lambdab).
\end{align*}
Given that $\pi_K = 1 - \sum_{k=1}^{K-1} \pi_k$, we also have
\begin{align*}
    \nabla_{\pi_k} \widehat{\mathcal{L}}_{\omega_t, N}(\Lambdab) &= \nabla_{\pi_k} \left(\frac{1}{N} \sum_{i = 1}^N q_{\Lambdab}(\Xb_i)\right) + \omega_t \nabla_{\pi_k} \mathcal{H}_\mathcal{S}(q_\Lambdab) \\
    &= \frac{1}{N}\sum_{i = 1}^N (q_{\lambdab_k}(\Xb_i) - q_{\lambdab_K}(\Xb_i)) + \omega_t \nabla_{\pi_k} \mathcal{H}_\mathcal{S}(q_\Lambdab).
\end{align*}
Combining these expressions, we find the Gaussian mixture GERVE update rules.

\paragraph{Entropy derivatives (for Monte-Carlo estimation).} The entropy derivatives do not admit a closed form, but they admit an integral representation. For the derivatives w.r.t. $\mub_k$ and $\Sb_k^{-1}$, the following expressions can be used:
\begin{align*}
    \nabla_{\Sb_k^{-1}} \mathcal{H}_\mathcal{S}(q_\Lambdab) &= - \frac{\pi_k}{2} \Sb_k \mathbb{E}_{q_{\lambdab_k}}[((\Xb - \mub_k) (\Xb - \mub_k)^T \Sb_k - \Ib) (1 + \log q_\Lambdab(\Xb)) \mathds{1}_\mathcal{S}(\Xb)]
\end{align*}
and
\begin{align*}
    \nabla_{\mub_k} \mathcal{H}_\mathcal{S}(q_\Lambdab) &= - \pi_k \Sb_k \mathbb{E}_{q_{\lambdab_k}}[(\Xb - \mub_k) (1 + \log q_\Lambdab(\Xb)) \mathds{1}_\mathcal{S}(\Xb)].
\end{align*}
Finally, for the derivative w.r.t. $\pi_k$, we use:
\begin{align*}
    \nabla_{\pi_k} \mathcal{H}_\mathcal{S}(q_\Lambdab) &=  \mathbb{E}_{q_{\lambdab_K}}[(1 + \log q_\Lambdab(\Xb)) \mathds{1}_\mathcal{S}(\Xb)] - \mathbb{E}_{q_{\lambdab_k}}[(1 + \log q_\Lambdab(\Xb)) \mathds{1}_\mathcal{S}(\Xb)].
\end{align*}

\subsection{Computational considerations}
\label{app:complexity}

GERVE's complexity for Gaussian mixtures with $K$ components is: 
\begin{itemize}
    \item Full covariances: $O\big((d^2 B + d^3) K T\big)$,
    \item Diagonal, isotropic or fixed covariances: $O(d B K T)$,
\end{itemize}
plus $O(d K^2 B_e T)$ for the Monte Carlo estimation of entropy gradients.

This is comparable to related methods:
\begin{itemize}
    \item \emph{Mean-shift mode-seeking.} A single run of classical mean-shift with fixed, spherical bandwidths and mini-batches costs $O(d B T)$. Running from $K$ initializations (to target $K$ modes) yields $O(d B K T)$, comparable to fixed-covariance GERVE. However, unlike GERVE, these independent trajectories lack a repulsion mechanism and may fail to explore the full modal structure.
    \item \emph{Mean-shift clustering.} Since mean-shift clustering launches one trajectory per data point, the cost is $O(d N B T)$ for spherical bandwidths. This can be substantially more expensive than GERVE, which only launches $K$ trajectories to find the cluster centroids.
    \item \emph{EM for Gaussian mixtures (GMM-EM).} The EM algorithm~\citep{dempster1977maximum, mclachlan2000finite} with full covariances has cost $O\big((d^2 B + d^3) K T\big)$
    per run. For diagonal covariances, it reduces to $O(d B K T)$. These are comparable to fixed-covariance GERVE. In contrast, GERVE uses natural-gradient steps rather than alternating E- and M-steps, and does not require likelihood maximization, making it applicable in likelihood-free settings.
\end{itemize}

\section{Simulations and benchmark}
\label{app:simulations}

\subsection{Details on the clustering example of Section \ref{sec:simu_clustering}}
\label{app:simulations_cluster}

\paragraph{Sample.} The dataset consists of $N=6000$ points from a 2D Gaussian mixture with three equally weighted components with means on the nodes of an equilateral triangle, at $(0,1)$, $(\cos(\pi/6),-0.5)$, $(-\cos(\pi/6),-0.5)$ and isotropic covariance $\sigma^2\Ib$ with $\sigma^2=0.25$. Figure~\ref{fig:samples_clustering} shows a sample and a kernel density estimate (KDE), revealing $J=3$ regions of high density.

\begin{figure}
    \centering
    \includegraphics[width=0.49\linewidth]{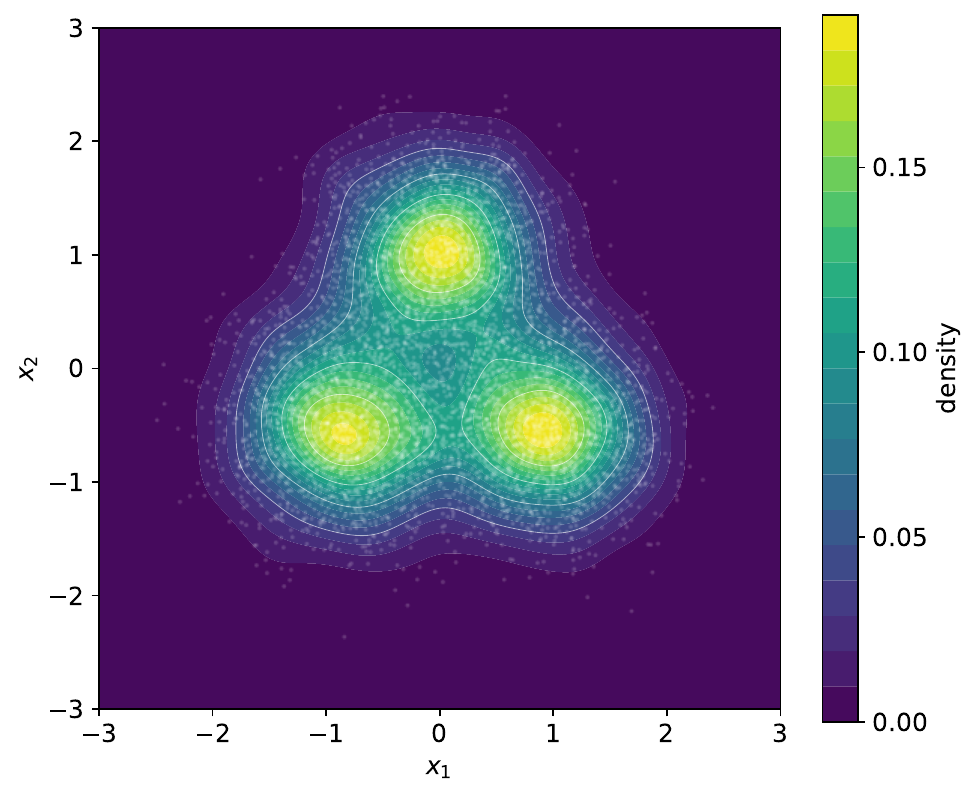}
    \caption{Sample of $N=6000$ points from a three-component Gaussian mixture whose means are located at the nodes of an equilateral triangle. Background shading: Gaussian KDE with bandwidth $h$ selected by Scott’s rule \citep{scott1992multivariate}, revealing three high-density regions.} 
    \label{fig:samples_clustering}
\end{figure}

\paragraph{GERVE hyperparameters.} 
The mixture parameters are initialized with $\mub_{1,1}, \dots, \mub_{K,1} \sim \mathrm{Uniform}([-2,0]^2)$, and for all $k \in [K]$, $\Sigmab_{k,1} = 2 \Ib$, $\pi_{k,1} = 1/K$. Table~\ref{tab:hyper_cluster} gives the other hyperparameters used for GERVE in this clustering task. 

\begin{table}
 \caption{Hyperparameters used for GERVE in the clustering simulations.}
\centering
\begin{tabular}{cc}
\toprule
 Hyperparameters & Values  \\ 
\midrule
  $T$ & $40000$ \\ 
  $B$ & $1000$ \\  
  $K$ & $\{3, 7\}$ \\   
  $\omega_t$ & $50/t^{1.1} + 0.004$ \\ 
  $\rho_t$ & $10^{-4}(50/\omega_t)^{0.7}$ \\ 
\bottomrule
\end{tabular}
\label{tab:hyper_cluster}
\end{table}

\paragraph{Additional figures.} Figure \ref{fig:clusters_k3} shows results for $K=3$ including an additional comparison with $k$-means for $K=7$. 

\begin{figure}
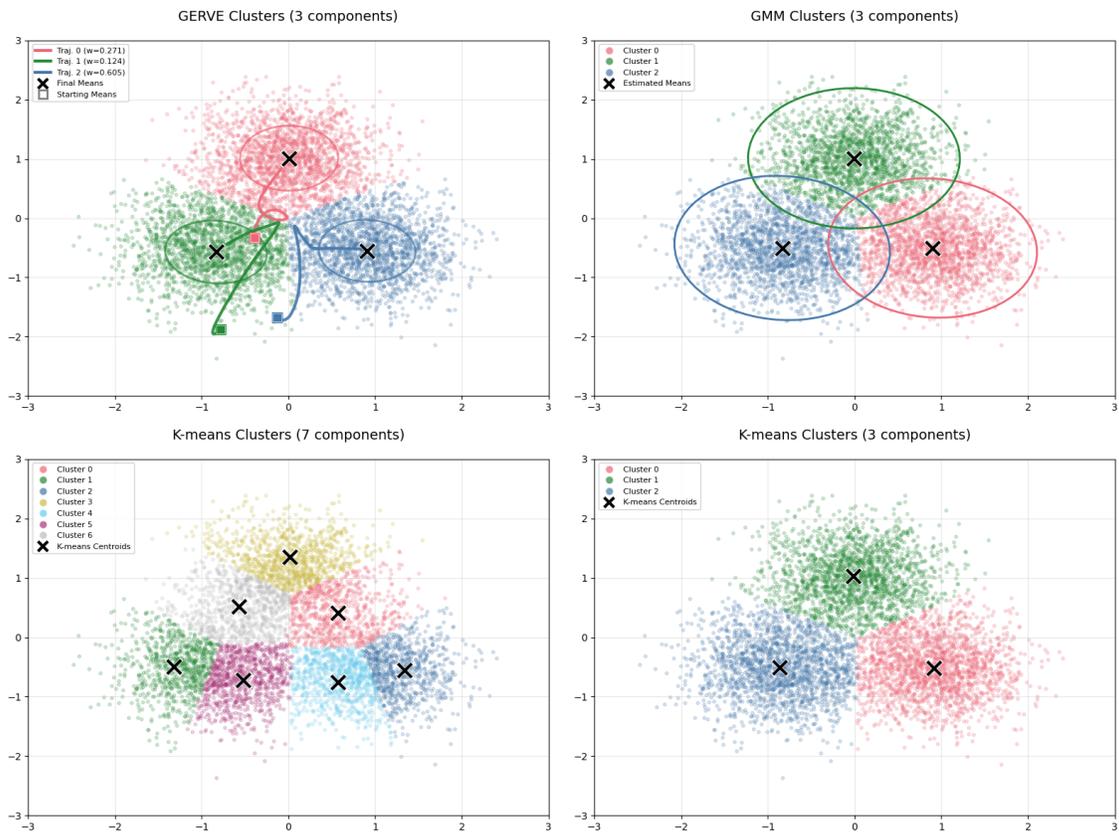

\centering
\includesvg[width=0.49\linewidth]{figures/cluster_gerve_3.svg}
\includesvg[width=0.49\linewidth]{figures/cluster_gmm_3.svg}\\
\includesvg[width=0.49\linewidth]{figures/cluster_km_7.svg} 
\includesvg[width=0.49\linewidth]{figures/cluster_km_3.svg}
\caption{Top row: Clustering with $K=3$ for GERVE (left) and  GMM-EM (right). Ellipses represent the covariance matrices of the Gaussian components. Bottom row: Clustering with $k$-means for $K=7$ (left) and $K=3$ (right).}
\label{fig:clusters_k3}
\end{figure}

\subsection{Mode-estimation}
\label{app:simulations_mode}

\paragraph{Models.} For each experiment replicate, we sample $N$ points. There are two settings, their respective densities are represented in Figure~\ref{fig:density_triangle_vs_skew}. 
\begin{itemize}
    \item[(A)] Gaussian triangle: We use the same triangle configuration with means at $(0,1)$, $(\cos(\pi/6),-0.5)$, $(-\cos(\pi/6),-0.5)$ and isotropic covariance $\sigma^2\Ib$, but with smaller covariance $\sigma^2=0.1$. 
    \item[(B)] Skew triangle: Each Gaussian component is replaced by a rotated skew-normal component 
$$X=\mu_k+\sigma R(\theta_k) (Z_1-\xi, Z_2)^T$$
with $Z_1\sim\mathrm{Skew}\mathcal{N}(0,1,4)$, $Z_2\sim\mathcal N(0,1)$, rotation angles $\theta_k\in\{0,\pi/3,-\pi/4\}$, and $\xi$ the mode of $Z_1$ (univariate skew-normal). Therefore, each component mode is exactly at $\mu_k$, and the geometry and scale match the Gaussian triangle but the density around each mode is skewed and anisotropic.
\end{itemize}  In both settings, the components are well separated, so the density global modes approximately coincide with the component modes: $\ub_1 \approx (0,1)$, $\ub_2 \approx (\cos(\pi/6),-0.5)$, $\ub_3 \approx (-\cos(\pi/6),-0.5)$.

\begin{figure}
\centering
\includegraphics[width=0.9\linewidth]{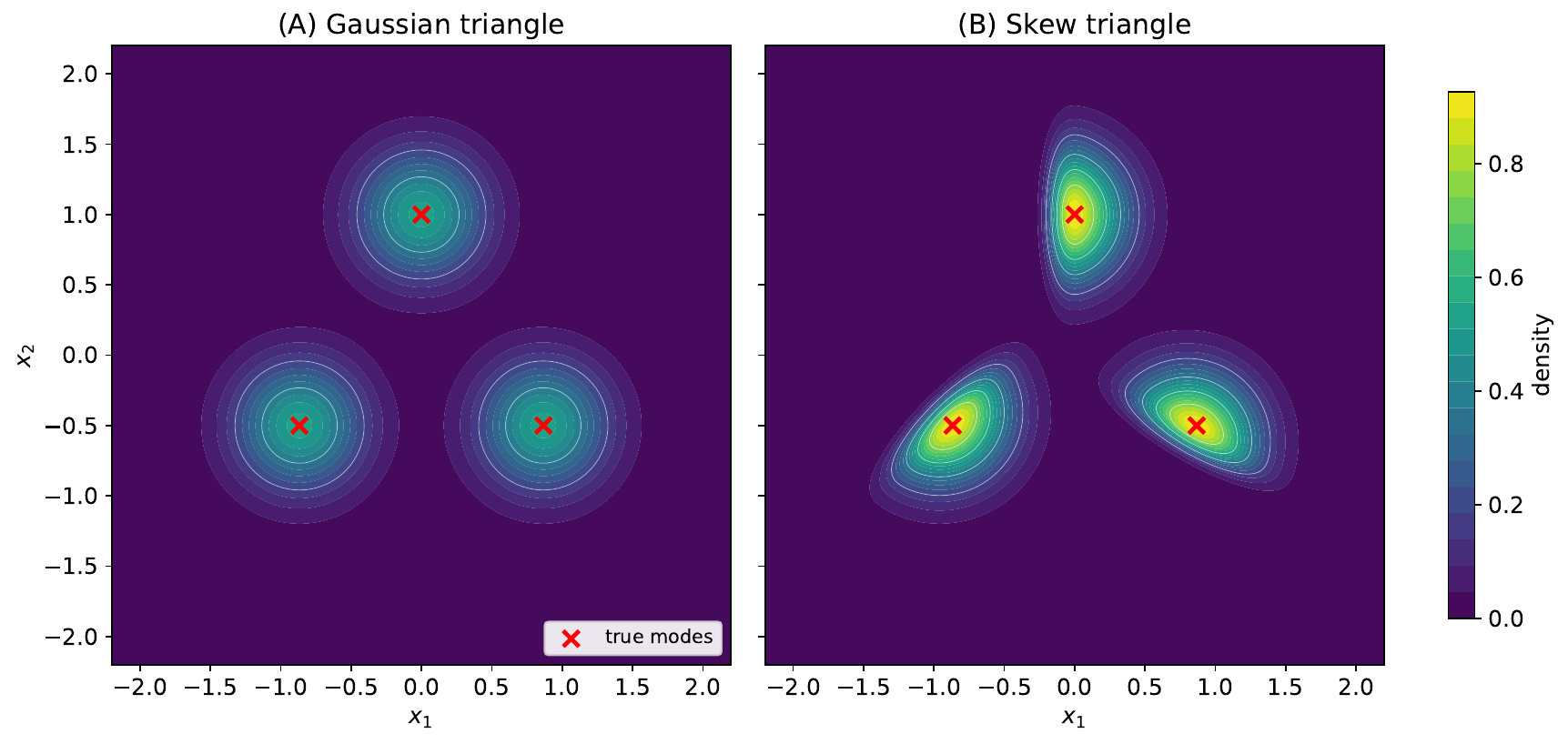}
\caption{Densities used for the mode-estimation experiment. Modes are marked by crosses. Left: Gaussian triangle density (A). Right: Skew triangle density (B).}
\label{fig:density_triangle_vs_skew}
\end{figure}

\paragraph{Details on the KDE+NVA baseline.}
To isolate the effect of GERVE's directly sample-based objective, we consider a two-stage KDE+NVA baseline. We first estimate the density from the full sample using a Gaussian KDE $\widehat p_h$, then apply the NVA Gaussian-mixture algorithm of \citet{leminh2025natural} to $\widehat p_h$ as the optimization target\footnote{code available at \texttt{github.com/tam-leminh/natural-variational-annealing}}. For each sample size, $K$, and evaluation metric, we use the bandwidth $h_m$ selected for Gaussian mean-shift and the corresponding best GERVE hyperparameters for the NVA optimization, including the initialization, annealing schedule, learning-rate schedule, and iteration budget. Thus, KDE+NVA differs from GERVE through the preliminary density-estimation step, while using the same variational-annealing machinery. Since evaluating the KDE repeatedly within NVA requires summing over the observations at each target evaluation, its computational cost increases rapidly with $N$. We therefore report KDE+NVA results up to $N=2^{14}$.

\paragraph{Details on the Local Mode Testing Algorithm baseline.}
We implement the local mode testing procedure of \citet{genovese2016non}\footnote{code available at \texttt{github.com/tam-leminh/lmta-python}}. For each replicate, the sample is randomly split into two halves. On the first half, candidate modes are obtained by Gaussian mean-shift applied to a KDE with scalar bandwidth $h_\ell$. On the second half, we evaluate the KDE Hessian at each candidate mode and apply the elementary-symmetric-polynomial bootstrap of \citet{genovese2016non} to test whether its largest Hessian eigenvalue is significantly negative. We use family-wise level $\alpha=0.05$ and $100$ bootstrap replicates. Following their bandwidth-selection proposal, $h_\ell$ is chosen to maximize the number of significant modes, with ties resolved in favor of the smallest bandwidth, over $h_\ell \in\{0.03,0.05,0.08,0.12,0.2,0.3\}$. Candidate modes are obtained from $100$ random mean-shift initializations and convergence points within $0.05$ are merged. To bound computation when very small bandwidths generate many KDE maxima, we retain at most the $20$ candidates with highest KDE values before significance testing. Since the method selects its number of significant modes internally, it does not depend on the value of $K$ used for GERVE and the other baselines. Owing to the computational cost of repeated KDE mode searches and bootstrap Hessian evaluations, we report this baseline up to $N=2^{16}$.

\paragraph{Metrics.} Each method is replicated $n_{\text{rep}}=100$ times and provides $\widehat{K}$ estimated modes $\widehat{\ub}=(\widehat{\ub}_1,\dots,\widehat{\ub}_{\widehat{K}})$, where $\widehat K=K$ for methods initialized from $K$ mode candidates and is data-dependent for LMTA.
We assess performance in terms of mode estimation using the following metrics:
\begin{itemize}
    \item \emph{Mode recovery} ($\mathsf{MR}_\epsilon$): number of true modes $\{\ub_i, i\in [I]\}$ recovered within a strict tolerance $\epsilon$:
\begin{equation*}
\mathsf{MR}_\epsilon(\widehat{\ub})=\sum_{i=1}^I \mathds{1} \left\{\min_{k\in[\widehat{K}]}\big\|\widehat{\ub}_k-\ub_i\big\|_2<\epsilon\right\}.
\end{equation*}
We set $\epsilon=10^{-2}$, a small value relative to $\sigma=\sqrt{0.1}$.
    \item \emph{Hungarian matching sum} (\textsf{HM}): minimum linear assignment cost between true and estimated modes,
\begin{equation*}
\mathsf{HM}(\widehat{\ub})
=\min_{\tau}
 \sum_{i=1}^{I}\big\|\widehat{\ub}_{\tau(i)}-\ub_i\big\|_2,
\end{equation*}
i.e., the minimum total distance over injective maps ($\widehat{K} \geq I$)  $\tau:\{1,\dots,I\}\hookrightarrow\{1,\dots,\widehat{K}\}$ (one distinct estimate per true mode). This can be computed via the Hungarian algorithm~\citep{kuhn1955hungarian}. When fewer than $I$ valid estimates remain (e.g., after discarding diverging GERVE or mean-shift estimates, or when LMTA returns fewer than $I$ significant modes), we optimally match the available estimates to distinct true modes and assign a penalty of $10/I$ to each unmatched true mode, so that a run returning no modes receives $\mathsf{HM}=10$.
    \item \emph{Nearest-neighbor sum} (\textsf{NN}): aggregate distance from each estimate to the closest true mode,
\begin{equation*}
\mathsf{NN}(\widehat{\ub})=\sum_{k=1}^{\widehat{K}}\min_{i\in[I]}\big\|\widehat{\ub}_k-\ub_i\big\|_2.
\end{equation*}
\end{itemize}

\paragraph{Hyperparameters and grid search.} For GERVE, the initial components have equal weights $1/K$ and covariances $\sigma_1^2 \Ib$, $\sigma_1 > 0$. The annealing schedule is $\omega_t = \omega_1/t^\beta$ and the step sizes are $\rho_t = \rho_1 (\omega_1/\omega_t)^\gamma$. We used $T = 4000$ iterations and mini-batches of size $B = 1000$. The mean-shift algorithm is run with a mini-batch size $B = 1000$ and uses a bandwidth $h_m$ for its Gaussian kernel. The feature significance method uses a bandwidth $h_f$ to identify the points of significant curvature or gradient. For each method, we performed a grid or line search to optimize hyperparameters with respect to each metric and sample size. Table~\ref{tab:hyper_benchmark} gives the values used in the grid search procedure.

For Gaussian mean-shift, we use a mini-batch fixed-covariance Gaussian GERVE (equivalent updates to Gaussian mean-shift with bandwidth $h$) with adaptive step size $\rho_t=B \Big(\sum_{b=1}^{B}q_{\lambdab_t}(\Xb^{(t)}_{b})\Big)^{-1}$, where $\Xb^{(t)}_{1:B}$ is the mini-batch sample. 

\begin{table}
 \caption{Grid of hyperparameter values used for GERVE, mean-shift, and the feature significance method in the mode-finding simulations.}
\centering
\begin{tabular}{cc}
\toprule
 Hyperparameters & Values  \\ 
\midrule
  $\sigma^2_1$ & $\{0.1, 0.2\}$ \\   
  $\omega_1$ & $\{10, 50\}$ \\ 
  $\beta$ & $\{1.1, 1.3, 1.5\}$ \\ 
  $\rho_1$ & $\{0.01, 0.03, 0.1\}$ \\ 
  $\gamma$ & $\{0.2, 0.3, 0.4\}$ \\ 
\midrule
  $h_m$ & $\{0.001, 0.002, 0.005, 0.01, 0.02, 0.05, 0.1, 0.2\}$ \\ 
  $h_f$ & $\{0.0001, 0.0003, 0.001, 0.003, 0.01, 0.03, 0.1, 0.3\}$ \\ 
\bottomrule
\end{tabular}
\label{tab:hyper_benchmark}
\end{table}

\section{Clustering performance on UCI datasets}
\label{app:real_datasets}

We assess robustness to misspecification on real data and compare to classic clustering methods. We also provide an ablation study and report runtimes.

\subsection{Benchmark}

\paragraph{Datasets.} The three datasets can be found in the UCI Machine Learning Repository (\textit{archive.ics.uci.edu}).
\begin{itemize}
    \item \textbf{Iris}: $N=150$, $d=4$ (sepal length/width, petal length/width), $J=3$ species (setosa, versicolor, virginica).
    \item \textbf{Wine}: $N=178$, $d=13$ chemical attributes from $J=3$ cultivars.
    \item \textbf{Pendigits}: $N\approx 11000$, $d=16$ (eight pen-tip $(x,y)$ pairs, normalized), $J=10$ classes (digits 0 to 9).
\end{itemize}

\paragraph{Clustering metrics.} We report the Adjusted Rand Index (\textsf{ARI}, pairwise agreement), the Adjusted Mutual Information (\textsf{AMI}, mutual information between estimated and true classes), Purity  (per-cluster majority proportion), and Macro-F1 (classwise F1, averaged).

\paragraph{GERVE default setup.} 
GERVE uses diagonal-covariance mixtures with weights fixed to $k$-means proportions. Means are initialized at $k$-means centroids, covariances at $\sigma^2_1 \Ib$. We include a brief fixed-covariance burn-in of $L_B$ iterations within a total budget of $T$ before learning covariances. We use overcompletion $K = J + 5$ to enable component merging. Other hyperparameters are listed in Table~\ref{tab:hyper_gerve}.

We enforce stability by clipping parameter updates ($\|\mub_{k,t+1}-\mub_{k,t}\|_2 \le D_{\mub}$, $\|\Sigmab_{k,t+1}-\Sigmab_{k,t}\|_F \le D_{\Sigmab}$) and constraining the parameter space via projections ($\sigma^2_{\min} I \preceq \Sigmab_k \preceq \sigma^2_{\max} I$). These constraints have little practical influence on results but improve stability in high dimensions. 

\begin{table}
 \caption{Dataset characteristics and hyperparameter values used for GERVE default setup.}
 \setlength{\tabcolsep}{10pt}
 \centering
 \begin{tabular*}{0.8 \textwidth}{@{\extracolsep{\fill}}l l c c c}
 \toprule
 & Dataset  & Iris & Wine & Pendigits \\
 \midrule
 \multirow{3}{*}{Dataset characteristics} & $N$ & $150$ & $178$ & $10992$ \\ 
  & $d$ & $4$ & $13$ & $16$  \\ 
  & $J$ & $3$ & $3$ & $10$  \\ 
 \midrule
 \multirow{12}{*}{GERVE default setup} & $T$ & $2000$ & $2000$ & $2000$ \\
  & $K$ & $8$ & $8$ & $15$ \\
  & $L_B$ & $100$ & $100$ & $100$ \\ 
  & $\sigma^2_1$ & $1$ & $1$ & $1$ \\   
  & $\omega_1$ & $0.1$ & $10^{-4}$ & $10^{-3}$ \\ 
  & $\omega_t$ & $\omega_1/t^2$ & $\omega_1/t^2$ & $\omega_1/t^2$ \\ 
  & $\rho_1$ & $10^6$ & $10^{12}$ & $10^{11}$ \\ 
  & $\rho_t$ & $\rho_1/t^2$ & $\rho_1/t^2$ & $\rho_1/t^2$ \\ 
  & $\sigma^2_{\min}$ & $0.01$ & $0.1$ & $0.01$ \\  
  & $\sigma^2_{\max}$ & $1$ & $1$ & $1$ \\  
  & $D_\mub$ & $0.1$ & $0.1$ & $0.2$ \\ 
  & $D_\Sigmab$ & $1$ & $1$ & $1$ \\ 
 \bottomrule
 \end{tabular*}
 \label{tab:hyper_gerve}
\end{table}

\paragraph{Baseline methods.}
We compare GERVE to classical clustering methods and to state-of-the-art modal clustering approaches \citep{chacon2020modal}.
\begin{itemize}
    \item \textbf{Flat mean-shift (MS-Scott):} flat kernel, bandwidth from Scott’s rule \citep{scott1992multivariate}. We set the maximum number of iterations to $300$ (scikit-learn \texttt{MeanShift}'s default).
    \item \textbf{GMM-BIC:} Gaussian mixture model with diagonal covariances, the number of components selected by BIC at each run. Initial mixture parameters are set using $k$-means++. We stop early when the lower bound average gain is below $10^{-3}$ and set the maximum number of iterations to $100$ (scikit-learn \texttt{GaussianMixture}'s default). 
    \item \textbf{$k$-means (KM-Elbow):} number of centroids selected by the elbow method. Initial centroid locations are set using greedy $k$-means++. We stop early when the lower bound average gain is below $10^{-4}$ and set the maximum number of iterations to $300$ (scikit-learn \texttt{KMeans}'s default).
    \item \textbf{Mixture Model Modal Merging (MMMM, \citealp{chacon2019mixture}):} starting from a fitted Gaussian mixture, component means are moved toward modes of the fitted mixture density, and mixture components converging to the same mode are merged into a single cluster. We use the R implementation provided by the author\footnote{available at \texttt{matematicas.unex.es/\string~jechacon}}. 
    \item \textbf{Mixture Model Modal Clustering (MMMC, \citealp{chacon2019mixture}):} starting from the same fitted Gaussian mixture, observations are assigned to modal basins according to the mode reached by gradient ascent. We use the R implementation provided by the author\footnote{available at \texttt{matematicas.unex.es/\string~jechacon}}. 
\end{itemize}

For reference, we also report oracle variants that use ground truth and are not available in practice: GERVE-K ($K = J$), MS-K (bandwidth chosen by line search so that the number of clusters is closest to $J$), MS-ARI (bandwidth maximizing \textsf{ARI}), GMM-K ($K = J$), GMM-ARI ($K$ maximizing \textsf{ARI}), KM-K ($K = J$), KM-ARI ($K$ maximizing \textsf{ARI}). 

GMM-BIC selects $K$ by BIC at each run. The -ARI oracle variants are chosen once per dataset by maximizing mean \textsf{ARI} across seeds and then held fixed across runs. 

\paragraph{Results.} 
Figure~\ref{fig:baseline_comparison} shows means and 95\% confidence intervals over $n_{\text{rep}}=10$ runs for the four metrics, and the median numbers of effective clusters and their [Q1-Q3] quantiles. GERVE is competitive with label-free baselines while adapting the number of clusters through merging. For reference, comparison against oracle baselines is provided in Figure~\ref{fig:baseline_comparison_all}.

\emph{Iris:} GERVE matches baselines on \textsf{ARI}/\textsf{AMI}/Macro-F1. It merges two close species, yielding two effective clusters versus $J=3$, with a small drop in Purity and a cleaner modal partition. MMMC and MMMM lead to similar conclusions, which indicates the limitations of modal clustering on this dataset. \emph{Wine:} GERVE attains strong scores across metrics. MMMC, MMMM, and $k$-means edge out the top values on this near-spherical dataset, while GERVE remains robust without tuning $K$ via BIC. \emph{Pendigits:} No method reaches $J=10$ in this overlapping, high-dimensional setting. MMMM and MMMC produce the effective cluster count closest to 
$J$, followed by GERVE. GERVE is competitive on all four metrics, with notably the best ARI and Purity scores, indicating consistent recovery of global structure even with class splits.

We emphasize \textsf{ARI}, \textsf{AMI}, and Macro-F1, since Purity can inflate with larger $K$. In practice, GERVE returns a smaller effective $K$ by merging redundant components. Compared to mean-shift, it scales better in dimension and avoids bandwidth tuning through adaptive covariance learning.

\begin{figure}
\centering
\includegraphics[width=0.95\linewidth]{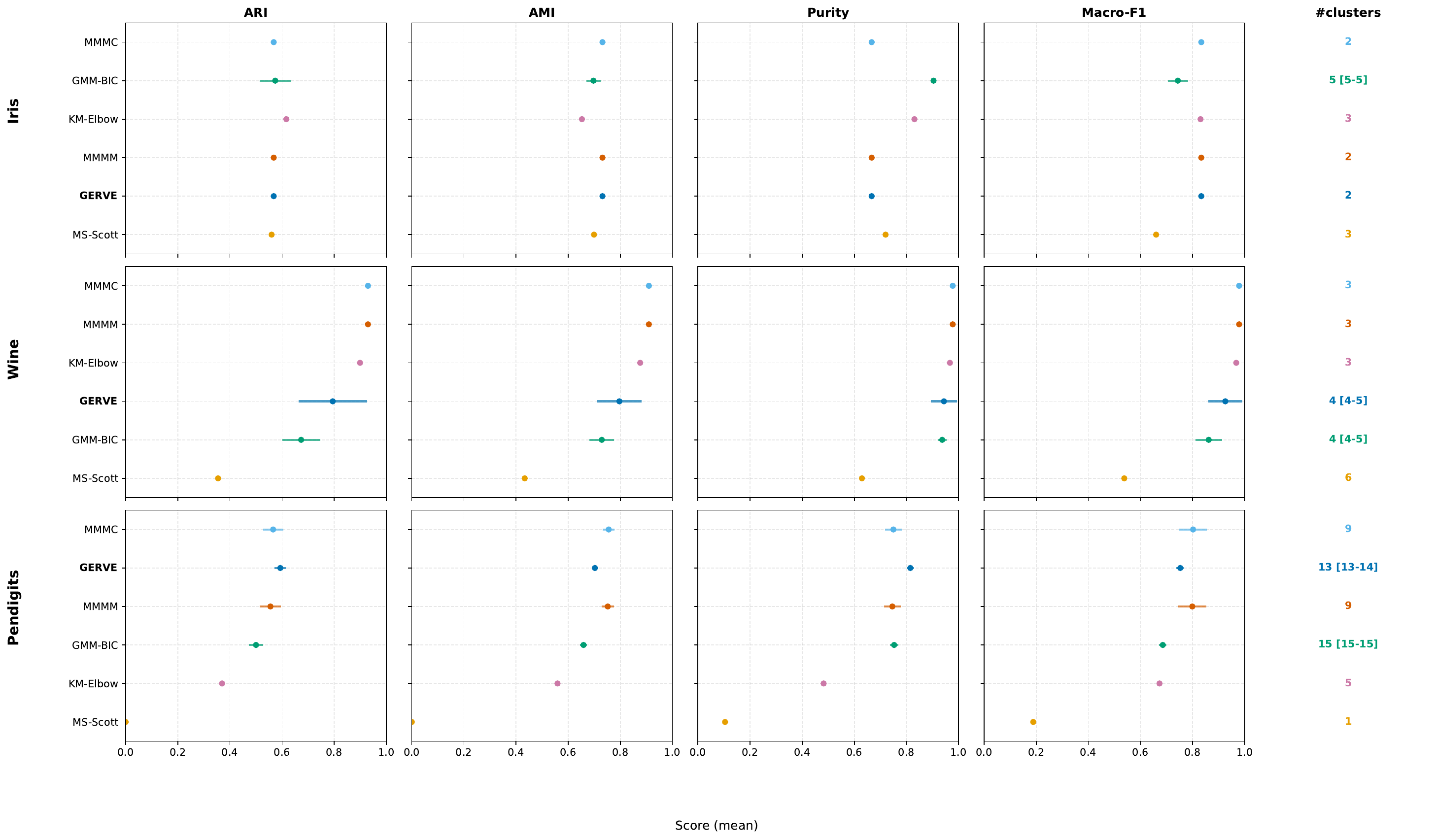}
\caption{Clustering metrics (means $\pm$95\% CI) and number of effective clusters (median and [Q1-Q3] quantiles when applicable, or constant across all replicates) for Iris, Wine, Pendigits over $n_{\text{rep}}=10$. Methods ordered by average rank within each dataset.}
\label{fig:baseline_comparison}
\end{figure}

\begin{figure}
    \centering
    \includegraphics[width=\linewidth]{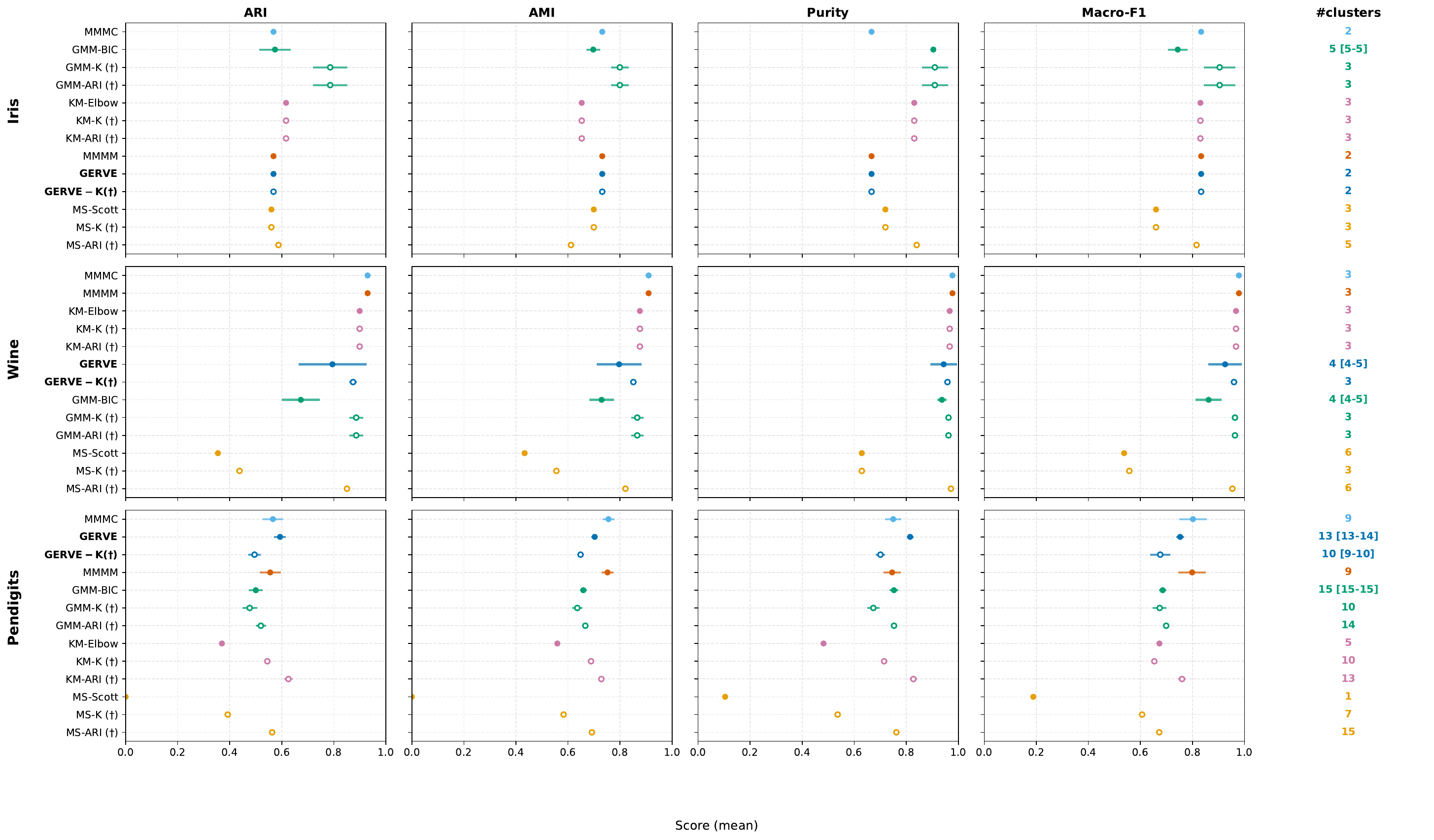}
    \caption{Clustering metrics (means $\pm$95\% CI) and effective clusters (median and [Q1-Q3] quantiles when applicable, or constant across all replicates) for Iris, Wine, Pendigits over $n_{\text{rep}}=10$.  Label-free methods (MS-Scott, GMM-BIC, KM-Elbow, MMMM, MMMC) are ordered by average rank within each dataset. Oracle variants are marked with the † symbol, and directly placed below the corresponding base method.}
    \label{fig:baseline_comparison_all}
\end{figure}

\subsection{Ablation study}

We now turn to an ablation study to examine how different configuration choices (weights, overcompleteness, annealing, covariance structure, burn-in) affect the behavior of GERVE.

We start with the same default GERVE dataset-specific configuration. For each dataset, we create new configurations by varying one factor at a time relative to the default:
(i) Weights: equal vs.\ proportional to $k$-means group sizes;
(ii) Overcompleteness: $K\in\{J, J{+}5\}$;
(iii) Annealing: $\omega_t \in\{0,\omega_1/t^\beta, \beta\in\{0.5,1,2,3\}\}$;
(iv) Covariance: fixed, diagonal, full;
(v) Burn-in: $L_B\in\{0,50,100,200,500\}$.

Figures~\ref{fig:ablation_components}-\ref{fig:ablation_covariance} report, for each configuration, the mean of each metric (with $95\%$ confidence intervals) and the effective number of clusters (median [Q1-Q3]). We note the main takeaways for each varying factor: 
(i) Weights proportional to $k$-means help on Wine (higher \textsf{ARI}/\textsf{AMI}/Purity/Macro-F1), without significant effect elsewhere; 
(ii) Starting overcomplete slightly increases the effective number of clusters but helps on Pendigits, where class overlap is substantial; 
(iii) Faster annealing improves end-state consolidation under fixed iteration budgets, $\omega_t \equiv 0$ underperforms on overlapping classes, illustrating the importance of the exploration induced by the entropy term; 
(iv) Diagonal covariances are a good speed vs. robustness compromise, as full covariances add cost and variance without systematic gains;
(v) A fixed-covariance short burn-in ($L_B \approx 100$) can help, but longer burn-ins steal iterations from covariance learning.

\begin{figure}
    \centering
    \includegraphics[width=\linewidth]{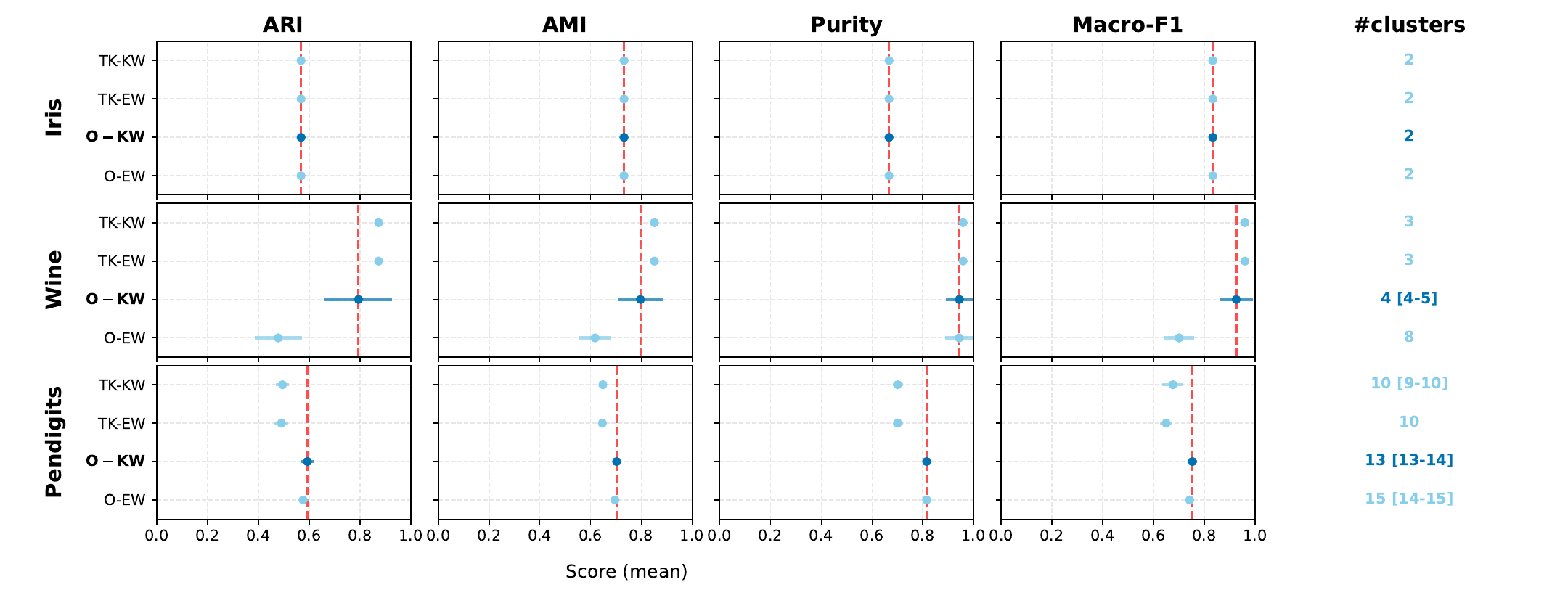}
    \includegraphics[width=\linewidth]{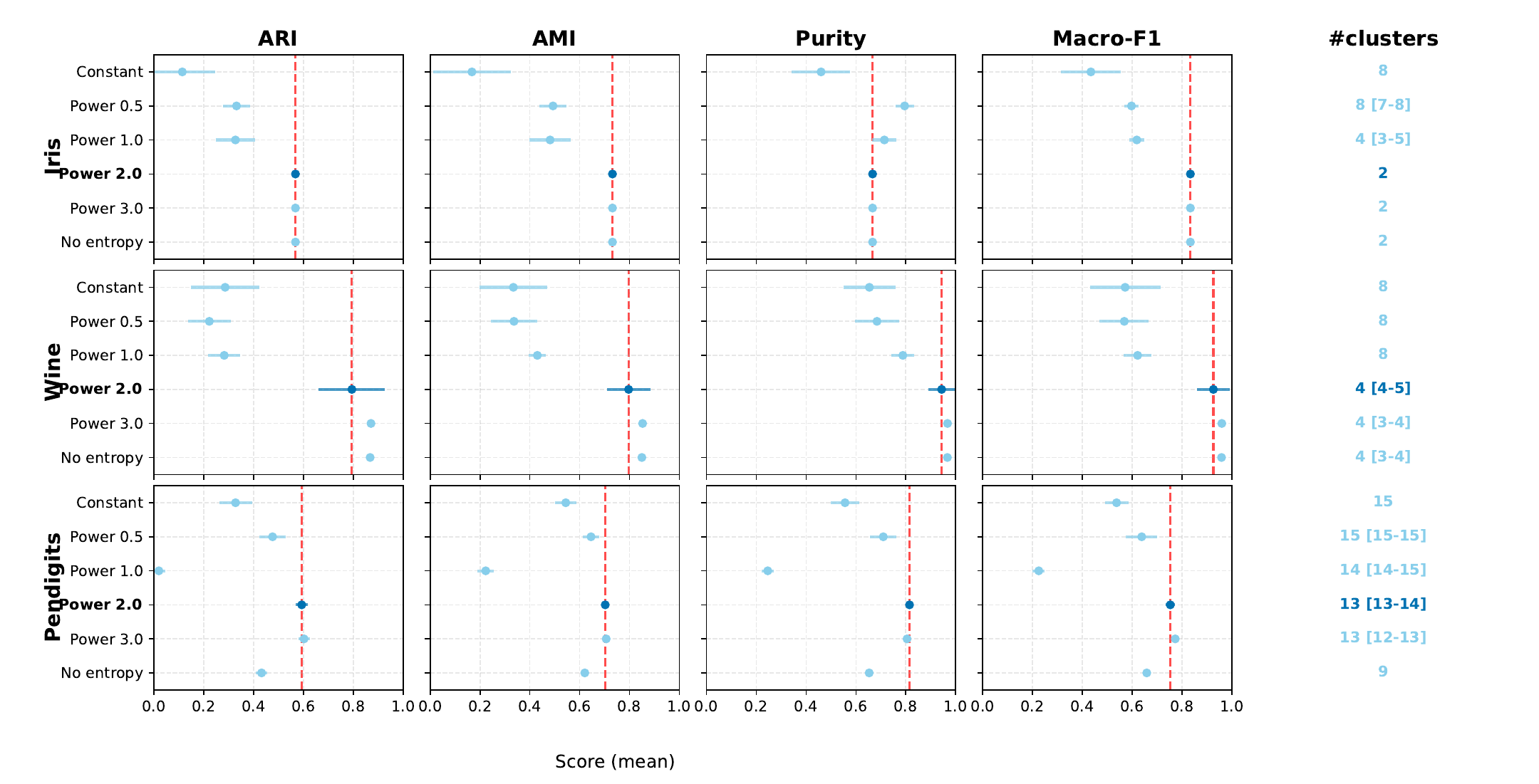}
    \caption{\textbf{Ablation -- Weights and overcompleteness (top) and Annealing (bottom).} For each dataset (rows), the first four columns show \textsf{ARI}/\textsf{AMI}/Purity/Macro-F1 (means ± 95\% $t$-CIs over $n_{\text{rep}}=10$ runs), the rightmost column shows the effective number of clusters (median [Q1--Q3], or constant across all runs). Bold marks the default GERVE configuration. The dashed reference line in each metric panel marks the default’s mean for that metric. 
    Dataset-specific $J$ values are listed in Table~\ref{tab:hyper_gerve}. \\
    Top. Prefixes: \textbf{T}- for $K=J$ (true groups), \textbf{O}- for $K=J{+}5$ (overcomplete). Suffixes: \textbf{-KW} for $k$-means-proportional weights, \textbf{-EW} for equal weights. \\
    Bottom. Schedules: \emph{Constant} ($\omega_t \equiv \omega_1$), \emph{Power-$x$} ($\omega_t=\omega_1 t^{-x}$), \emph{No entropy} ($\omega_t \equiv 0$).}
    \label{fig:ablation_components}
\end{figure}

\begin{figure}
    \centering
    \includegraphics[width=\linewidth]{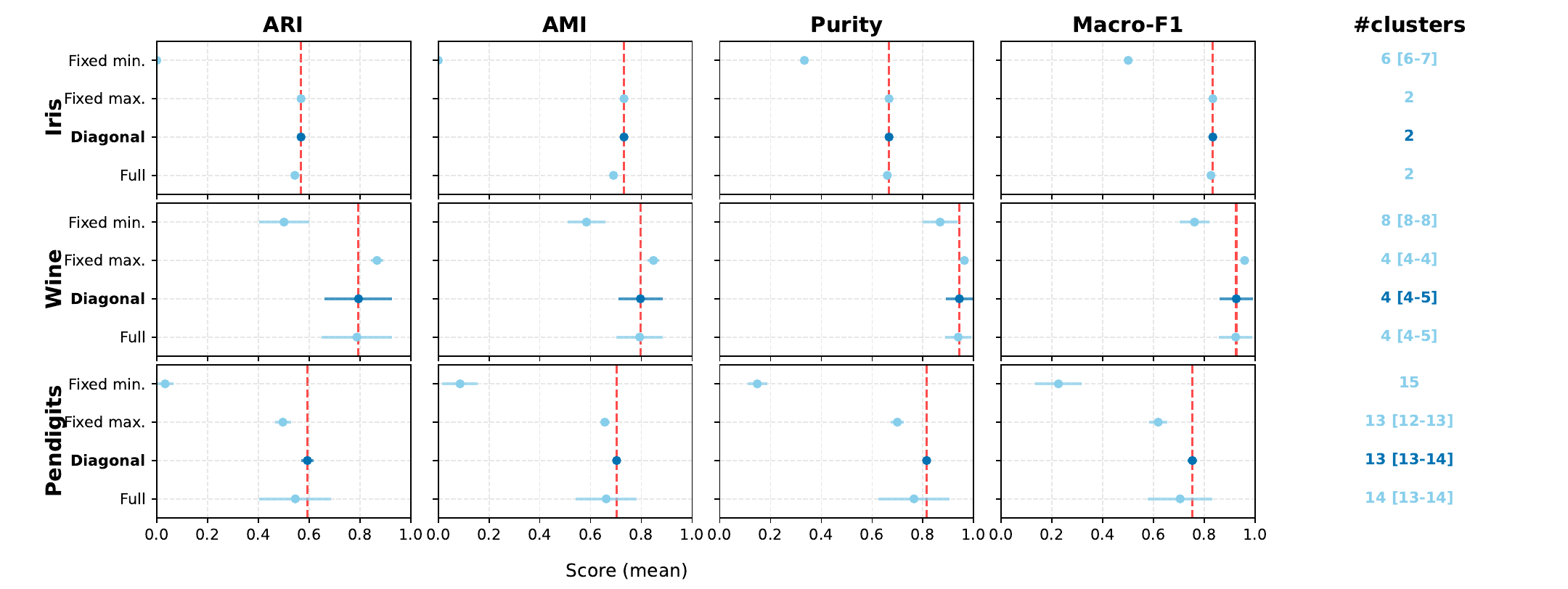}
    \includegraphics[width=\linewidth]{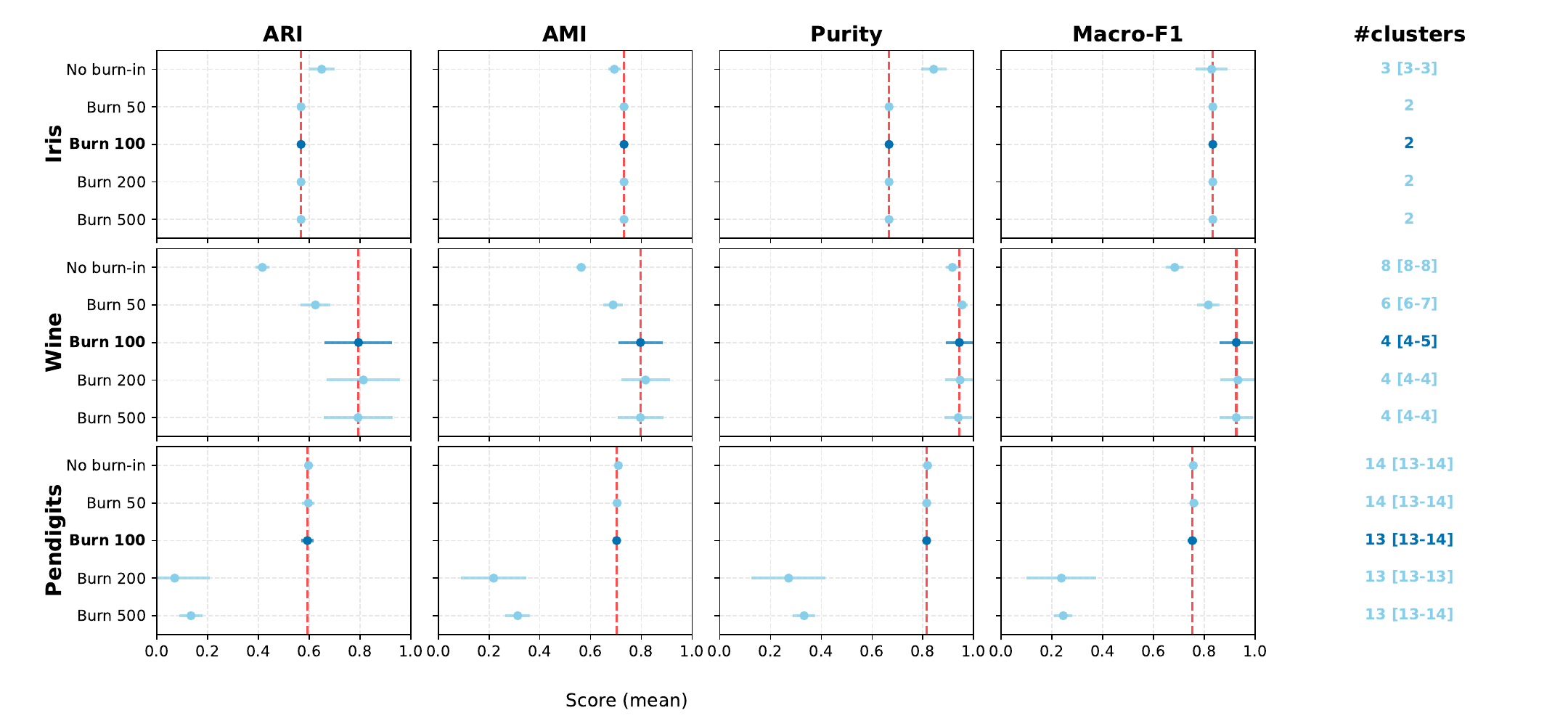}
    \caption{\textbf{Ablation -- Covariance (top) and Burn-in phase (bottom).} 
    For each dataset (rows), the first four columns show \textsf{ARI}/\textsf{AMI}/Purity/Macro-F1 (means ± 95\% $t$-CIs over $n_{\text{rep}}=10$ runs), the rightmost column shows the effective number of clusters (median [Q1--Q3], or constant across all runs). Bold marks the default GERVE configuration. The dashed reference line in each metric panel marks the default’s mean for that metric. \\
    Top. Covariance parameterizations: \emph{Fixed min.} ($\Sigmab_k \equiv \sigma_{\min}^2 \Ib$), \emph{Fixed max.} ($\Sigmab_k \equiv \sigma_{\max}^2 \Ib$), \emph{Diagonal} (learned diagonal $\Sigmab_k$ with bounds $\sigma_{\min}^2 \Ib \preceq \Sigmab_k \preceq \sigma_{\max}^2 \Ib$), and \emph{Full} (learned full $\Sigmab_k$ under the same bounds). Dataset-specific values for $\sigma_{\min}^2,\sigma_{\max}^2$ and the initialization $\sigma_0^2$ (for learned covariances) are listed in Table~\ref{tab:hyper_gerve}.\\
    Bottom. Burn-in variants: \emph{No burn-in} ($L=0$), \emph{Burn-$x$} (fixed-covariance phase with $L=x$ iterations). The total iteration budget $T$ is the same across variants and listed in Table~\ref{tab:hyper_gerve}.}
    \label{fig:ablation_covariance}
\end{figure}

\subsection{Runtime}

To illustrate how GERVE’s computational cost depends on data dimension, the number of mixture components, and covariance parameterization, we report mean runtimes across four datasets 
(Triangle mixture, Iris, Wine, Pendigits), under a common configuration. This is not intended as a scalability benchmark, since runtime is implementation-dependent and sensitive to engineering choices.

Table~\ref{tab:runtime} shows average wall-clock runtimes over ten replicates. As expected, cost increases with dimensionality and with the flexibility of the covariance model. Fixed covariances are the fastest to compute but also the least flexible, while full covariances are substantially slower, especially in higher dimensions. The diagonal setting offers a practical compromise: only slightly slower than the fixed covariance setting, yet avoiding the instability and high cost of full covariance updates. Overall, runtimes remain modest on medium-sized datasets, indicating that GERVE is practically usable in realistic clustering scenarios even when covariance learning is included.

\begin{table}[h]
\centering
\caption{Mean wall-clock runtime (in seconds) over 10 replicates for overcomplete GERVE ($K=J+5$) with weight updates, $T=500$, $B=150$, $L=0$, using identical hyperparameters across all datasets. Values are mean $\pm$ standard deviation. All runs on the same machine. This study illustrates cost trends by covariance parameterization, number of components, and dimensionality.}
\label{tab:runtime}
\begin{tabular}{lccc}
\toprule
Dataset & Fixed & Diagonal & Full \\
\midrule
Triangle ($d=2, J=3$) & 3.5 $\pm$ 0.6 & 3.9 $\pm$ 0.8 & 4.2 $\pm$ 0.8 \\
Iris ($d=4, J=3$)     & 3.5 $\pm$ 0.7 & 3.8 $\pm$ 0.7 & 4.6 $\pm$ 0.9 \\
Wine ($d=13, J=3$)    & 4.2 $\pm$ 0.8 & 4.8 $\pm$ 0.9 & 7.8 $\pm$ 1.2 \\
Pendigits ($d=16, J=10$) & 13.4 $\pm$ 2.1 & 15.2 $\pm$ 2.6 & 24.6 $\pm$ 3.3 \\
\bottomrule
\end{tabular}
\end{table}

We also plot execution time as a function of $K$ on the Pendigits dataset. The scaling is approximately linear in $K$, with diagonal covariances remaining a practical compromise between fixed and full parameterizations.

\begin{figure}
    \centering
    \includegraphics[width=0.5\linewidth]{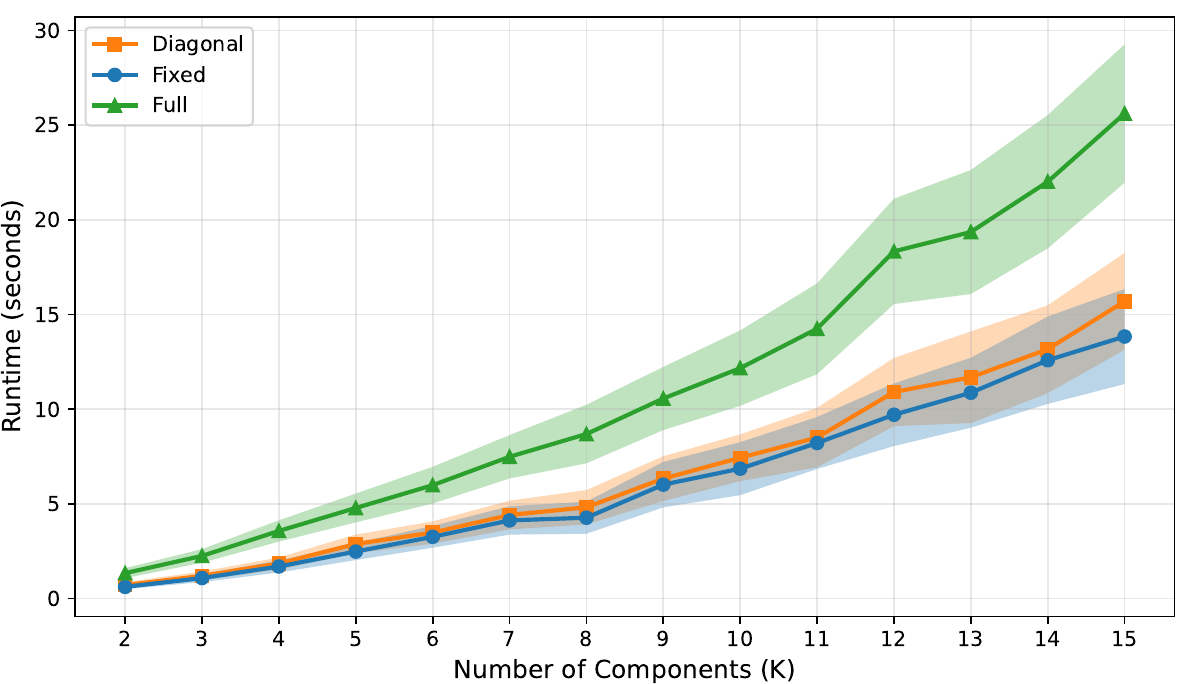}
    \caption{GERVE runtime with respect to the number of components $K$ on the Pendigits dataset, for diagonal, fixed and full covariance matrices. Curves show mean runtime over 10 replicates, shaded bands indicate $\pm 1$ standard deviation. 
    }
    \label{fig:runtime_k}
\end{figure}

\section{Details on the UK collision case study}
\label{app:uk}

\paragraph{Data normalization and processing.} We apply GERVE on a normalized space : the Greater London window $[-0.54, 0.33] \times [51.28, 51.70]$ (longitude-latitude) into a centred rectangle $[-0.7,0.7] \times [-0.35, 0.35]$ preserving the aspect ratio and with area $\approx 1$. This transformation is applied to the coordinates of the data points. To smoothen the underlying distribution $p$, we generate artificial data points uniformly in the region between the normalized rectangle and the $[-2, 2] \times [-2,2]$ square, with density $1000$ times smaller than the mean density of the data in the normalized rectangle.

\paragraph{GERVE hyperparameters.} 
We target the top $H=10$ hotspots. We fit an overcomplete mixture with $K=2H=20$.
\begin{itemize}
    \item \textit{Temperature.} We fix the stopping temperature at $\omega^\dagger=1$. This choice follows the elbow of the resolved-mode curve represented in Figure~\ref{fig:w0_component_counts_plot} and described in Sec.~\ref{sec:real_datasets}. The same $\omega^\dagger$ is used for the bootstrap refits.
    \item \textit{Step sizes.} The step-size schedule is $\rho_t = 0.2 t^{-0.1}$. 
    \item \textit{Initialization.} Component means use $k$-means++ centroids on the working coordinates. Initial covariances are $\Sigmab_{k,1}=\sigma_1^2\Ib$ with $\sigma_1^{2}=5\times 10^{-3}$.
    \item \textit{Covariance bounds.} We constrain $\sigma_{\min}^2=1\times 10^{-5}$ and $\sigma_{\max}^2=1\times 10^{-2}$. 
    \item \textit{Domain.} Optimization is performed on $\mathcal S=[-2.05,2.05]\times[-2.05,2.05]$. Mapping to British National Grid is used only for reporting distances and ellipses.
    \item \textit{Entropy term.} We use $B_e=100$ Monte Carlo samples per iteration to estimate the entropy contributions.
    \item \textit{Early stopping.} We cap iterations at $T=10{,}000$. We stop early if the following hold at three consecutive checks, evaluated every 10 iterations:
\begin{equation*}
\big\|\mub_{k,t+1}-\mub_{k,t}\big\|_2<10^{-2}
\quad\text{and}\quad
\frac{\big\|\Sb_{k,t+1}-\Sb_{k,t}\big\|_F}{\big\|\Sb_{k,t}\big\|_F}<10^{-1}
\quad\text{for all }k.
\end{equation*}
    \item \textit{Pruning and merging.} A component is pruned if
\begin{equation*}
\mathrm{eig}_{\max}(\Sigmab_k)>\sigma_{\min}^2+\beta\big(\sigma_1^2-\sigma_{\min}^2\big),
\end{equation*}
with \(\beta \approx 0.018\), so that the threshold equals \(10 \sigma_{\min}^2\).
Components means within $0.005$ ($\approx 200$ meters in the British National Grid EPSG:27700) are merged.
    \item \textit{Bootstrap settings.} We use $L=500$ resamples at $\omega^\dagger$. Bootstrap modes are matched to the baseline with a Hungarian assignment and adaptive gates around each baseline mode as described below.
\end{itemize}

\begin{figure}
\centering
\includegraphics[width=0.6\linewidth]{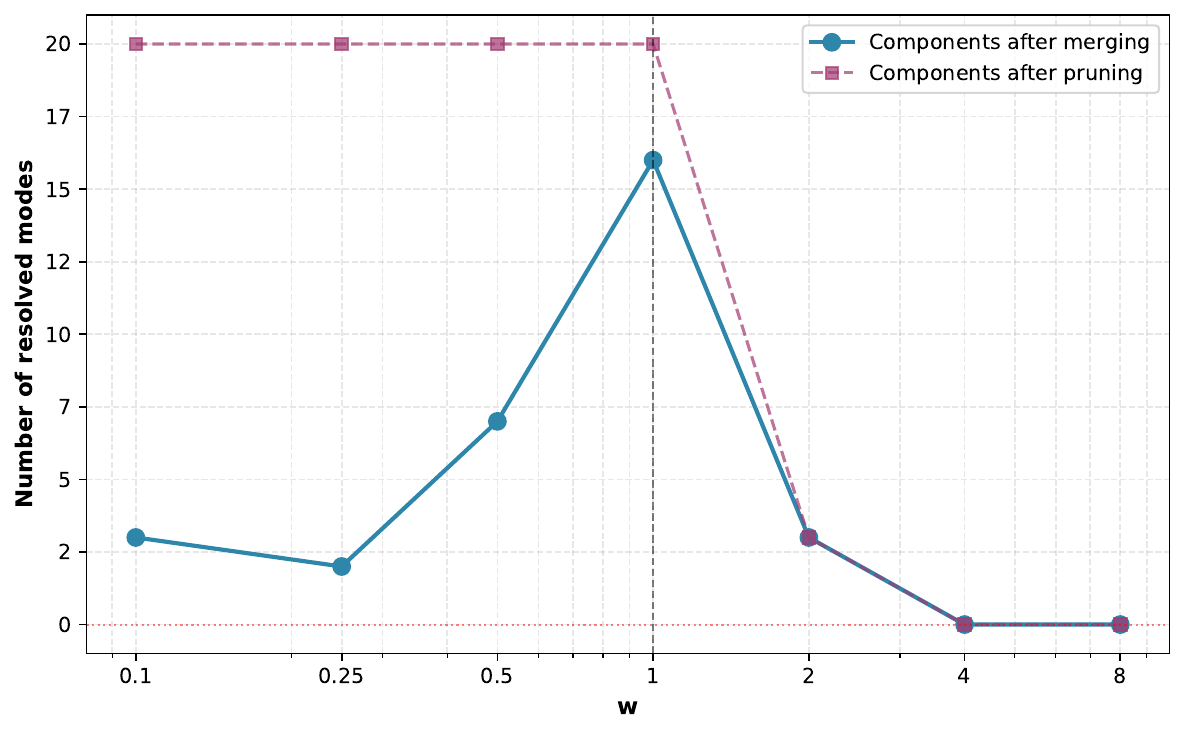}
\caption{Diagnostic for choosing \(\omega\): resolved modes after successive pruning (dashed) and merging (solid) for values of $\omega$ (in log-scale). The elbow at $\omega^\dagger = 1$ maximizes resolved modes before decline.}
\label{fig:w0_component_counts_plot}
\end{figure}

\paragraph{Mode matching and per-mode uncertainty.} We $z$-score coordinates using the baseline mode centres. For mode $k$, we compute $\mathrm{nn}_k$ as the distance to the nearest other baseline mode. The gate is $\tau_k=\mathrm{clip}(\eta  \mathrm{nn}_k,\tau_{\min},\tau_{\max})$ with defaults $\eta=0.35$, $\tau_{\min}=0.08$, $\tau_{\max}=0.22$, and then cap $\tau_k$ strictly below $nn_k/2$. Candidate pairs with distance greater than $\tau_k$ are excluded. Some baseline modes may remain unmatched. Among the feasible assignments, we first maximize the number of matches and then minimize the total matching distance. Accepted matches form the bootstrap cloud for mode $k$, from which we compute stability $s_k$ and a Gaussian approximation ellipse.

To express localization uncertainty in meters, we project coordinates to the British National Grid (EPSG:27700). For each resolved mode $k$, let $\{\widehat{\mub}^{(\ell)}\}_{\ell=1}^L$ denote the matched bootstrap centres and $\widehat{V}_k$ their empirical covariance: if $\mathrm{eig}_1 \ge \mathrm{eig}_2$ are eigenvalues of $\widehat{V}_k$ and $\eb_1$ is the principal eigenvector, then the ellipse semi-axes are $a=\sqrt{\chi^2_{2;1-\alpha} \mathrm{eig}_1}$ and $b=\sqrt{\chi^2_{2;1-\alpha} \mathrm{eig}_2}$, with angle $\mathrm{atan2}(\eb_{1,y}, \eb_{1,x})$ (in degrees) relative to Easting. These are Gaussian approximation ellipses conditional on successful gated matching.

Figure~\ref{fig:matched_ellipses_2024} shows, for each baseline mode, the cloud of matched bootstrap centres, the adaptive gate used for matching, and the resulting ellipse, illustrating mode stability and localization uncertainty.

\begin{figure}
    \centering
    \includegraphics[width=\linewidth]{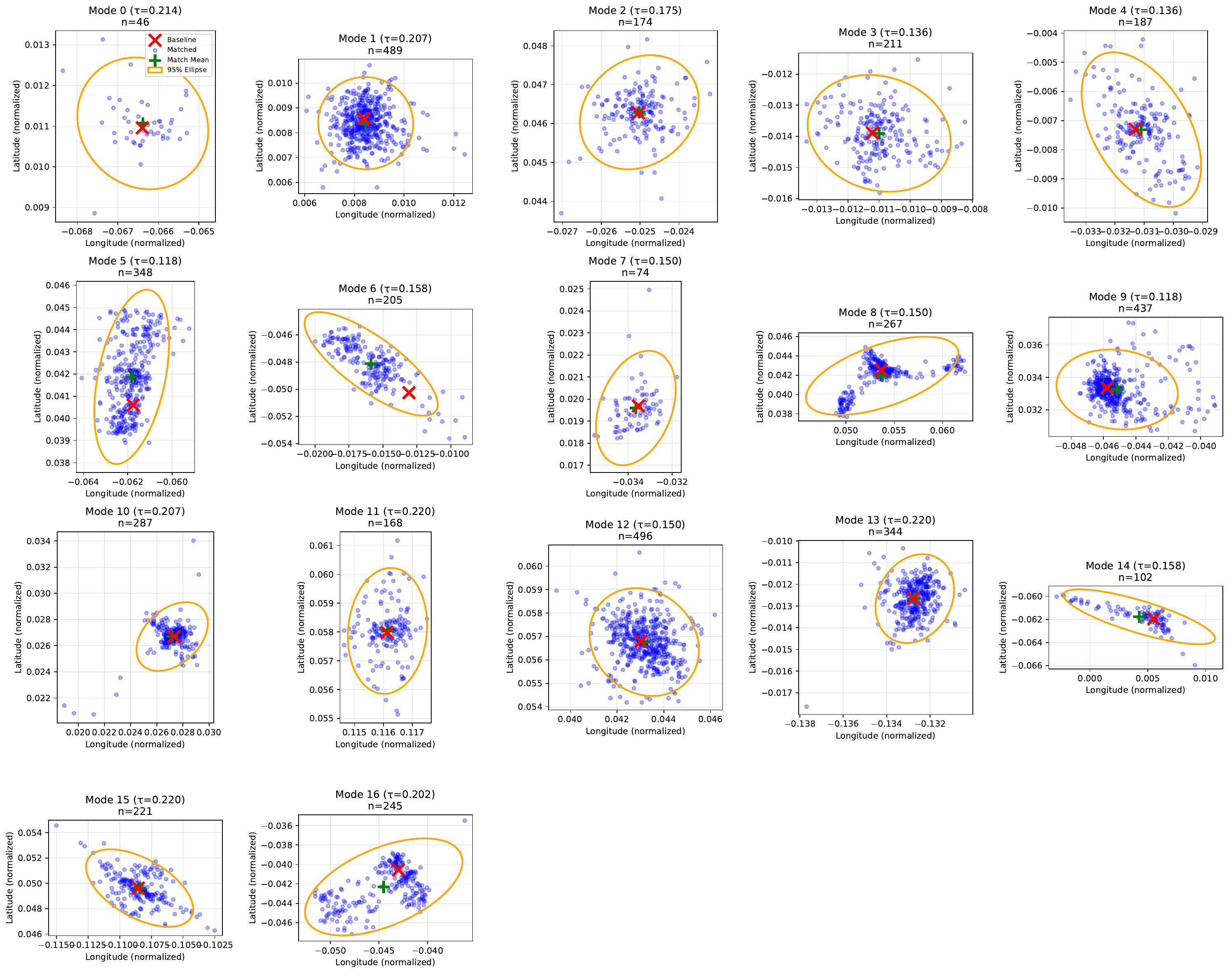}
    \caption{Distribution of the matched modes across the $L = 500$ bootstrap fits for each baseline hotspot from the 2020-2024 data, with the estimated ellipses and means.}
    \label{fig:matched_ellipses_2024}
\end{figure}

\paragraph{2020-2024 hotspots.} Table~\ref{tab:hotspots_2024_full} reports the $17$ hotspots found by GERVE and Figure~\ref{fig:collision_with_ellipses_2024} displays them in the complete search space.

\begin{table}[t]
\centering
\caption{Reported hotspots for years 2020-2024 with (un-normalized) coordinates, Gaussian approximation ellipse dimensions $(a,b)$, angles, stability, and local counts (radius 200 m). Hotspots are listed in decreasing stability order.}
\label{tab:hotspots_2024_full}
\small
\begin{tabular}{rrrrrrrrr}
\toprule
{ID} & {Name} & {Longitude ($^\circ$)} & {Latitude ($^\circ$)} & {$a$ (m)} & {$b$ (m)} & {Angle ($^\circ$)} & {Stability} & {Count} \\
\midrule
12 & Shoreditch & -0.078909 & 51.524210 & 136 & 122 & -54 & 0.99 & 32 \\
1 & Elephant \& Castle & -0.099878 & 51.494989 & 123 & 83 & 84 & 0.98 & 27 \\
9 & Piccadilly Circus & -0.132259 & 51.510012 & 163 & 159 & 87 & 0.87 & 28 \\
5 & Oxford Circus & -0.142326 & 51.515215 & 265 & 66 & 84 & 0.70 & 31 \\
13 & Gunter Grove & -0.185142 & 51.482277 & 133 & 83 & 72 & 0.69 & 22 \\
10 & London Bridge & -0.088568 & 51.506048 & 173 & 119 & 57 & 0.57 & 21 \\
8 & Aldgate & -0.072514 & 51.515251 & 375 & 208 & 30 & 0.53 & 13 \\
16 & Clapham HS & -0.131868 & 51.464360 & 404 & 256 & 36 & 0.49 & 25 \\
15 & Edgware Road & -0.170483 & 51.519900 & 230 & 144 & -41 & 0.44 & 22 \\
3 & Oval & -0.111627 & 51.481522 & 119 & 105 & 80 & 0.42 & 19 \\
6 & Brixton & -0.114567 & 51.460850 & 306 & 118 & -48 & 0.41 & 38 \\
4 & Vauxhall & -0.123765 & 51.485506 & 180 & 82 & -68 & 0.37 & 27 \\
2 & Holborn & -0.120095 & 51.517885 & 87 & 79 & 54 & 0.35 & 25 \\
11 & Mile End & -0.034828 & 51.524982 & 147 & 58 & 88 & 0.34 & 23 \\
14 & Herne Hill & -0.102411 & 51.452610 & 305 & 91 & -24 & 0.20 & 10 \\
7 & Westminster & -0.125308 & 51.501758 & 175 & 75 & 78 & 0.15 & 18 \\
0 & Victoria & -0.145082 & 51.496610 & 97 & 84 & -58 & 0.09 & 20 \\
\bottomrule
\end{tabular}
\end{table}

\begin{figure}
    \centering
    \includegraphics[width=\linewidth]{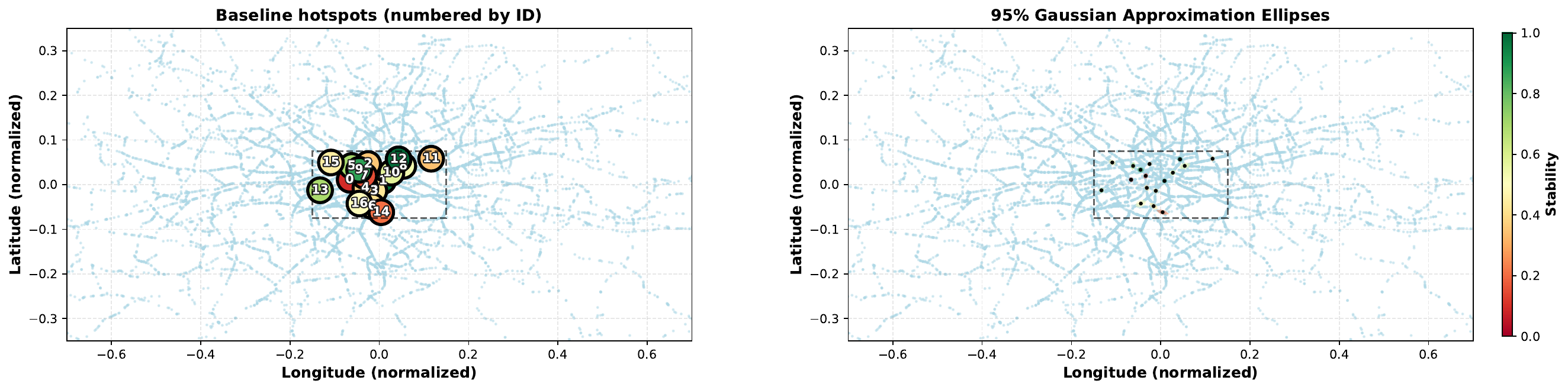}
    \caption{Baseline collision hotspots identified by GERVE for the Greater London Area between 2020 and 2024, with stability scores and 95\% Gaussian approximation ellipses, in normalized coordinates. The dashed rectangle is the zoomed-in window of Figure~\ref{fig:collision_with_ellipses_2024_zoomed}.} \label{fig:collision_with_ellipses_2024}
\end{figure}

\paragraph{Mean-shift details.} We use scikit-learn’s \texttt{MeanShift} with a flat (uniform) kernel on the same projected coordinates. The reference bandwidth $h_0$ is estimated as the 0.3-quantile of pairwise distances. We then sweep h on a grid from $0.02h_0$ to $0.2h_0$ and report the resulting
centres. Figure~\ref{fig:meanshift_specific_comparison} shows the centres across the sweep. Smaller $h$ recovers
many additional peripheral modes, while larger $h$ merges central structure. In Figure~\ref{fig:meanshift_vs_gerve_combined}, we
overlay the $0.05h_0$ and $0.1h_0$ centres in the zoomed-in window and
compare them with GERVE’s centres and 95\% Gaussian approximation ellipses. At $0.1h_0$, visual inspection against the road network some mean-shift
centres fall slightly off carriageways, so we select a smaller bandwidth ($0.05 h_0$) as a baseline comparator for GERVE.

\begin{figure}
    \centering
    \includegraphics[width=\linewidth]{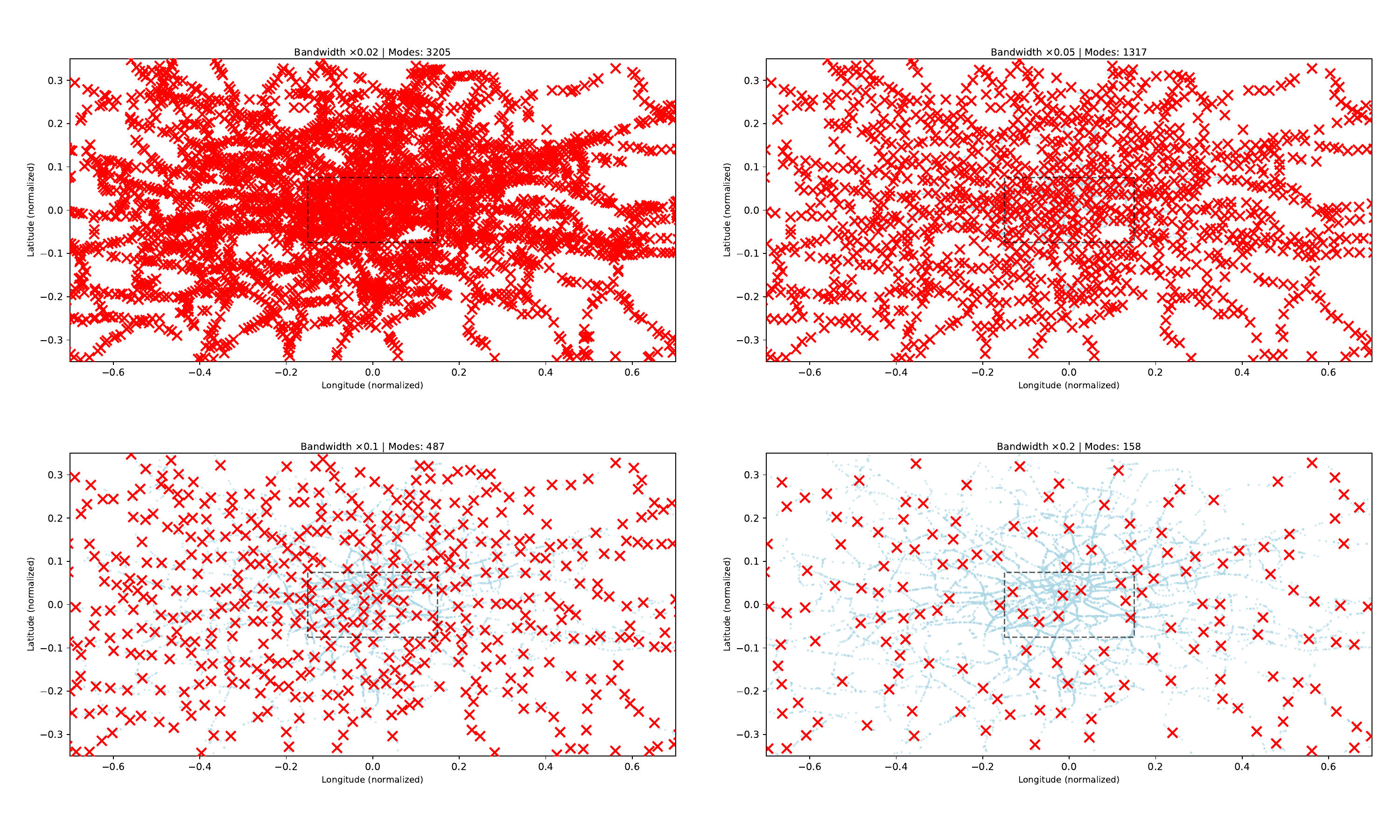}
    \caption{Mean-shift modes across different bandwidths ($h = 0.02 h_0,0.05 h_0, 0.1h_0, 0.2h_0$) in the normalized window. The dashed rectangle is the zoomed-in window.}
    \label{fig:meanshift_specific_comparison}
\end{figure}

\begin{figure}
\centering
\includegraphics[width=0.95\linewidth]{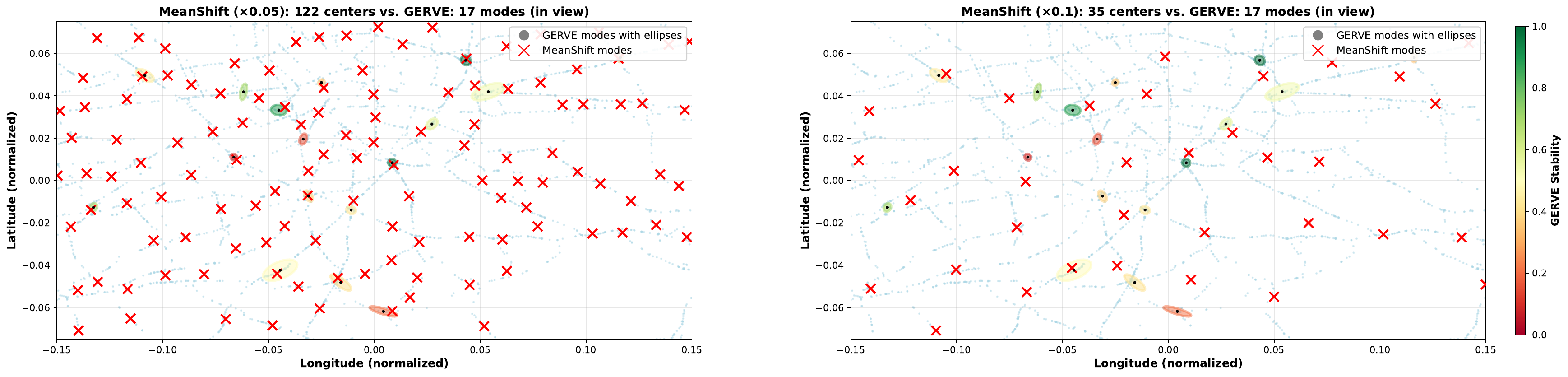}
\caption{Comparison of mean-shift ($h = 0.05 h_0, 0.1h_0$) and GERVE modes, with Gaussian approximation ellipses for GERVE, in the zoomed-in window.}
\label{fig:meanshift_vs_gerve_combined}
\end{figure}

\paragraph{2014-2019 hotspots and comparison to 2020-2024.} We fit GERVE to collision data from 2014 to 2019 with identical $\omega^\dagger$ and $K$. Ellipses are obtained after $L = 500$ bootstrap fits. Table~\ref{tab:hotspots_2019_full} reports the $16$ modes found by GERVE and Figure~\ref{fig:collision_with_ellipses_2019} displays them with their $95\%$ ellipses. Figure~\ref{fig:hotspot_comparison_2019_vs_2024} overlays the 2014-2019 hotspots on a map of the 2020-2024 ones in the zoomed-in window.

\begin{table}
\centering
\caption{Reported hotspots for years 2014-2019 with (un-normalized) coordinates, Gaussian approximation ellipse dimensions $(a,b)$, angles, stability, and local counts (radius 200 m). Hotspots are listed in decreasing stability order.}
\label{tab:hotspots_2019_full}
\small
\begin{tabular}{rrrrrrrrr}
\toprule
{ID} & {Name} & {Longitude ($^\circ$)} & {Latitude ($^\circ$)} & {$a$ (m)} & {$b$ (m)} & {Angle ($^\circ$)} & {Stability} & {Count} \\
\midrule
14 & Shoreditch & -0.078440 & 51.524044 & 167 & 135 & -81 & 1.00 & 34 \\
5 & Elephant \& Castle & -0.100224 & 51.494700 & 116 & 105 & -72 & 0.99 & 27 \\
0 & Ludgate Circus & -0.104357 & 51.513539 & 181 & 107 & -89 & 0.92 & 20 \\
4 & Charing Cross & -0.126952 & 51.507754 & 169 & 118 & -3 & 0.77 & 24 \\
1 & Vauxhall & -0.123467 & 51.486081 & 111 & 104 & -64 & 0.76 & 23 \\
15 & Sheperd's Bush & -0.223191 & 51.503715 & 179 & 133 & -3 & 0.70 & 25 \\
9 & Oxford Circus & -0.141671 & 51.515040 & 202 & 96 & 57 & 0.68 & 26 \\
2 & Aldgate & -0.071434 & 51.515463 & 307 & 120 & 13 & 0.67 & 15 \\
11 & Oval & -0.111774 & 51.481317 & 96 & 77 & -74 & 0.58 & 23 \\
6 & St Giles Circus & -0.130822 & 51.516010 & 171 & 122 & -21 & 0.56 & 15 \\
12 & Aldwych & -0.118405 & 51.511311 & 146 & 109 & 17 & 0.56 & 23 \\
3 & King's Cross & -0.122409 & 51.530580 & 266 & 155 & 24 & 0.50 & 21 \\
7 & Camden Town & -0.140959 & 51.540304 & 198 & 90 & 59 & 0.45 & 19 \\
8 & Edgware & -0.168207 & 51.518687 & 139 & 138 & -40 & 0.26 & 22 \\
13 & London Bridge & -0.088646 & 51.506248 & 328 & 133 & 64 & 0.24 & 14 \\
10 & Olympia & -0.206160 & 51.496955 & 153 & 71 & 7 & 0.07 & 14 \\
\bottomrule
\end{tabular}
\end{table}

\begin{figure}
    \centering
    \includegraphics[width=0.95\linewidth]{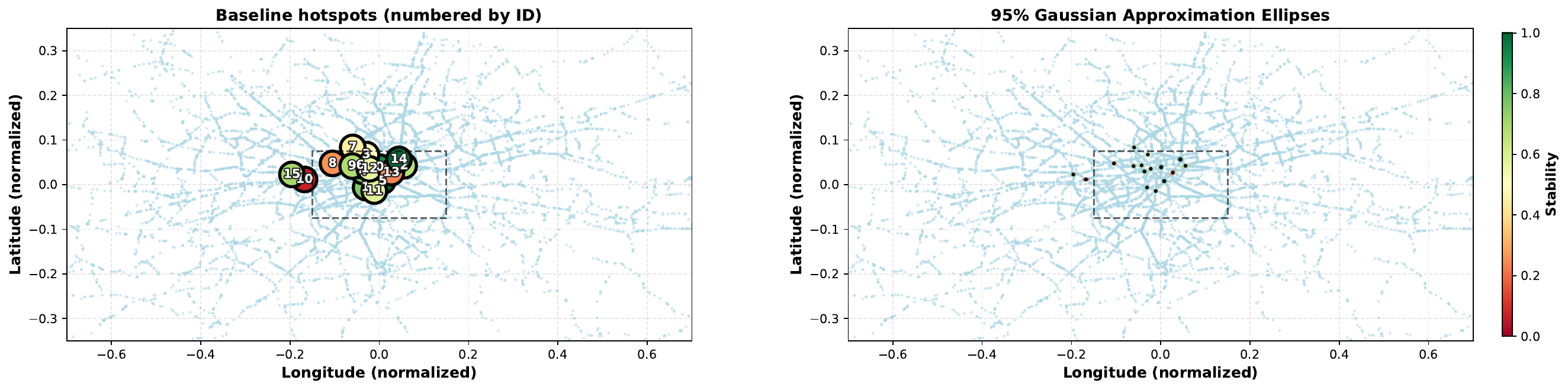}\\
    \includegraphics[width=0.95\linewidth]{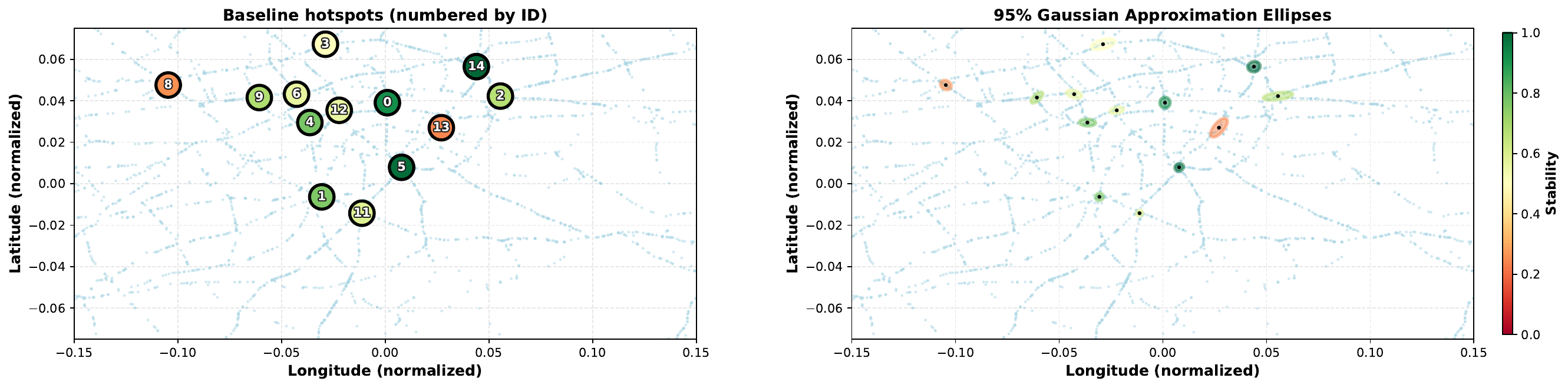}
    \caption{Baseline collision hotspots identified by GERVE for the Greater London Area between 2014 and 2019, with stability scores and 95\% Gaussian approximation ellipses, in normalized coordinates. Top: Shows the normalized window containing all data points. The dashed rectangle is the zoomed-in window. Bottom: Zoomed-in version for comparison to Figure~\ref{fig:collision_with_ellipses_2024_zoomed}.}
    \label{fig:collision_with_ellipses_2019}
\end{figure}

\begin{figure}
\centering
\includegraphics[width=0.8\linewidth]{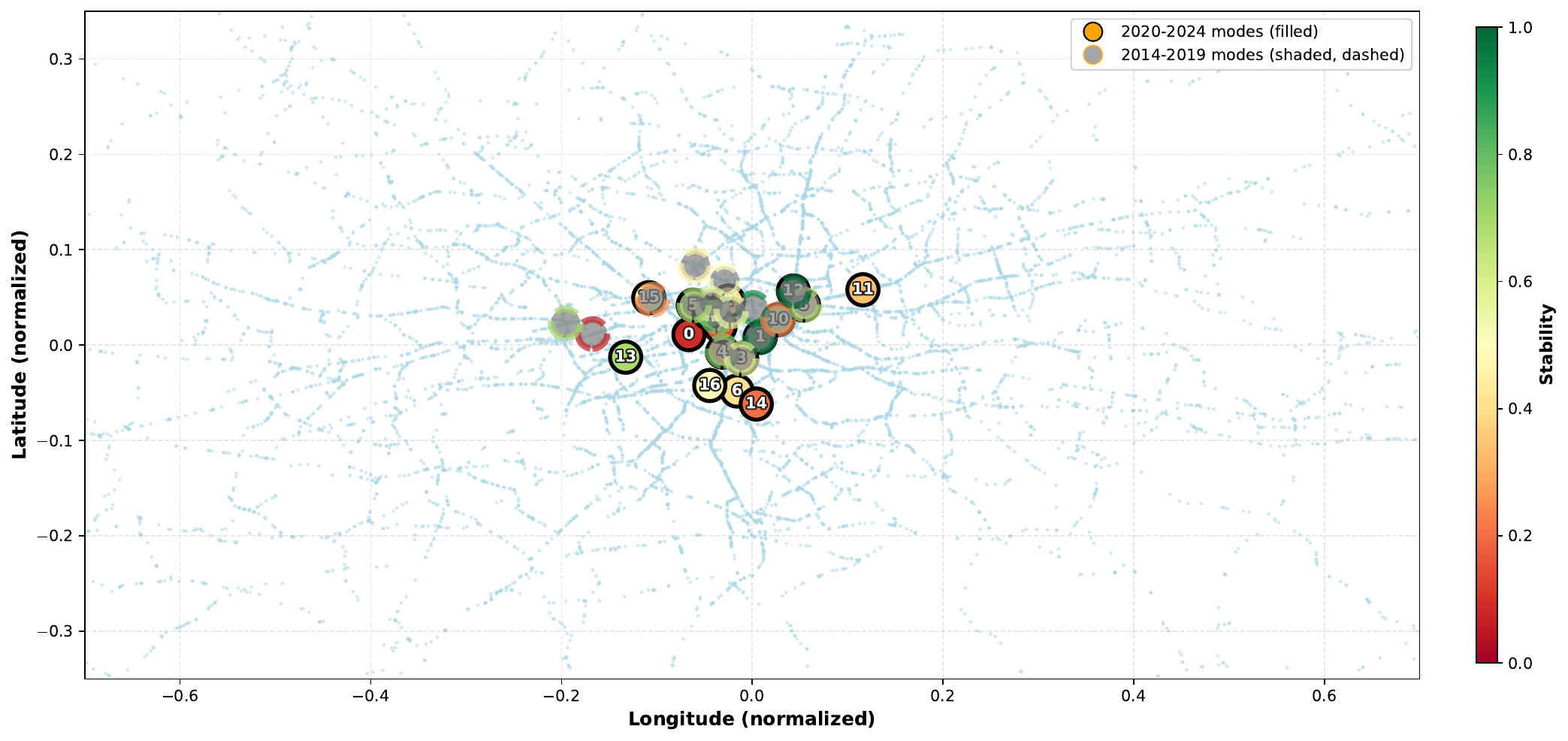}
\caption{Comparison of 2014-2019 hotspots and 2020-2024 hotspots in the zoomed-in window.}
\label{fig:hotspot_comparison_2019_vs_2024}
\end{figure}

\end{appendix}

\newpage

\end{document}